\documentclass[11pt,a4paper]{article}

\usepackage{geometry}
\usepackage{setspace}

\usepackage{bbm}
\usepackage{graphics} % for pdf, bitmapped graphics files
\usepackage{epsfig} % for postscript graphics files
\usepackage{mathptmx} % assumes new font selection scheme installed
\usepackage{times} % assumes new font selection scheme installed
\usepackage{amsmath} % assumes amsmath package installed
\usepackage{amssymb}  % assumes amsmath package installed
\usepackage{amsfonts}
\usepackage[centerlast]{caption}
\usepackage{graphicx}          % Include this line if your

\usepackage{tabularx}
\usepackage{booktabs}

\usepackage{amssymb}
\usepackage{subfigure}

\usepackage{amsthm}

\newtheorem{theorem}{Theorem}
\newtheorem{proposition}{Proposition}
\newtheorem{definition}{Definition}
\newtheorem{lemma}{Lemma}

\newtheorem{remark}{Remark}
\newtheorem{example}{Example}

\newtheorem{fact}{Fact}

\usepackage{amsmath}

\begin{document}

\title{Destructive nodes in multi-agent controllability
\thanks{This work was supported by the National Natural Science Foundation of China (Nos. 61374062, 61174131), Science Foundation of Shandong Province for Distinguished Young Scholars (No. JQ201419) and the Natural Sciences and Engineering Research Council of Canada.}}

\author{Zhijian Ji
\thanks{Zhijian Ji and Haisheng Yu are with the College of Automation Engineering, Qingdao University, Qingdao, 266071,
China. Corresponding author: Zhijian Ji. E-mails: jizhijian@pku.org.cn (Z. Ji), yhsh@qdu.edu.cn (H. Yu)}, 
Tongwen Chen
\thanks{Tongwen Chen is with the Department of Electrical and Computer Engineering, University of Alberta, Edmonton, Alberta, Canada T6G 2V4. E-mail:tchen@ualberta.ca (T. Chen).}, 
        and  Haisheng Yu$^{\dag}$}

\maketitle

\begin{abstract}                          % Abstract of not more than 200 words.
In this paper, several necessary and sufficient graphical conditions are derived for the controllability of multi-agent systems by taking advantage of the proposed concept of controllability destructive nodes. A key step of arriving at this result is the establishment of a relationship between topology structures of the controllability destructive nodes and a specific eigenvector of the Laplacian matrix. The results on double, triple and quadruple controllability destructive nodes constitute a novel approach to study the controllability. In particular, the approach is applied to the graph consisting of five nodes to get a complete graphical characterization of controllability.
\end{abstract}

%\end{frontmatter}

\section{Introduction}

Designing control strategies directly from network topologies is challenging, which contributes to an efficient manipulation of networks and a better understanding of the nature of systems. This requires research of the interplay between network topologies and system dynamics \cite{Yuan2013}.  Recently, considerable efforts have been made along this line in the multi-agent literature to understand how communication topological structures are related to controllability, which is also the focus here, where destructive nodes are defined to characterize controllability-relevant topologies. 

Multi-agent controllability was formulated under a leader-follower framework in which the influence over network is achieved by exerting control inputs upon leaders \cite{Tanner2004a}. A system is controllable if followers can be steered to proper positions to form any desirable configuration by regulating the movement of leaders. The earliest necessary and sufficient algebraic condition was presented in \cite{Tanner2004a}, which was expressed in terms of eigenvalues and eigenvectors of submatrices of Laplacian.  Another one was given in \cite{Rahmani2009}, which related controllability to the existence of a common eigenvalue of the system matrix and the Laplacian. Besides, a relationship between controllability and the eigenvectors of Laplacian was presented in \cite{JiWCon}, which provided a method of determining leaders from the elements of eigenvectors. Armed with these results, the virtue that leaders should have was characterized from both algebraic and graphical perspectives  \cite{Ji2012b}. Other algebraic conditions exist in, e.g., \cite{Rahimian2013,Lou2012,Zhangshuo2014,JiIJC,Liu2008,Lozano2008}. Recently, a unified protocol design method was proposed for controllability in  \cite{JiTAC2015}.  

Algebraic conditions lay the foundation for understanding interactions between topological structures and controllability. Previous work has suggested that this issue is quite involved, even for the simplest path graph \cite{Parlangeli2012}.  Special topologies were studied first, such as grid graphs \cite{Notarstefano2013}, symmetric structures \cite{Rahmani2006,Martini2008}, Cartesian product networks \cite{Chapman2012}, multi-chain topologies \cite{Egerstedt2012a,Cao2013} and tree graphs \cite{Ji2012b}. Controllability can be fully addressed by analyzing the eigenvectors of Laplacian, see e.g., \cite{Parlangeli2012,Notarstefano2013}. It can also be tackled through topological construction which sometimes relates to the partition of graphs. For example, topologies were designed by using the vanishing coordinates based partition \cite{Ji2012b} and an eigenvector based partition \cite{Jisubmitted2013}. In particular, the construction of uncontrollable topologies not only facilitates the design of control strategies but also deepens understanding of controllable ones \cite{Cao2013,JiWCon}. Recently, it was proved, via a proper design of protocols, that the controllability of single-integrator, high-order and generic linear multi-agent systems is uniquely determined by the topology structure of the communication graph  \cite{JiTAC2015}.  

The above work guides a further study of this issue. The topology structures of three kinds of the so-called controllability destructive nodes are discriminated and defined. Each structure depicts a topological relationship of destructive nodes to leader nodes so that leaders cannot distinguish the former, and thus destroys the controllability. Moreover, necessary and sufficient graphical conditions are derived by taking advantage of the concept of controllability destructive nodes. The results exhibit a new method of tackling controllability by which a complete graphical characterization of controllability is given for graphs consisting of five nodes.

\section{Preliminaries}

Consider a set of $n+l$ single integrator agents given by
\begin{equation}\label{singmul}
\begin{cases}
\dot x_i    =  u_i ,\qquad\;i = 1, \ldots ,n;\\
\dot z_j   =   u_{n+j},\quad j=1,\ldots,l,
\end{cases}
\end{equation}
where $n$ and $l$ are the number of followers and leaders, respectively;  $x_i$ and $z_j$ are the states of the $i$th and $(n+j)$th agent, respectively. Let $z_{1},\cdots,z_{l}$ act as leaders and be influenced only via  external control
inputs. $\mathcal {N}_i=\{j\,|\;v_i\sim v_j; j\neq i\}$ represents the neighboring set of
$v_i$ and  `$\sim$' denotes the neighboring relation. The followers are governed by neighbor based rule
\begin{equation}\label{ui}
u_i  =\sum\limits_{j \in \mathcal{N}_i } {(x_j  - x_i )}+
  \sum\limits_{(n+j) \in \mathcal{N}_i } {(z_j-x_i )},
\end{equation}
where $j\in\{1,\ldots,n\};$ $ (n+j)\in\{n+1,\ldots,n+l\}.$ $x,$ $z$ denote the stack vectors of $x_i$'s, $z_j$'s, respectively. The information flow between agents is incorporated in a graph $\mathcal{G},$ which consists of a node set $\mathcal{V}=\{v_1,\ldots,v_{n+l}\}$ and an edge set $\mathcal{E}=\{(v_i,v_j)\in\mathcal{V}\times\mathcal{V}|\,v_i\sim v_j\},$ with nodes representing agents and edges indicating the interconnections between them. $\mathcal{L}=D-A$ is the Laplacian, where $A$ is the adjacency matrix of $\mathcal{G}$ and $D$ is the diagonal matrix with diagonal entries $d_i=|\mathcal{N}_i|,$ the cardinality of $\mathcal{N}_i.$ Under (\ref{ui}), the followers'
dynamics is
\begin{equation}\label{ueqfolw}
\dot x=-\mathcal{F}x-\mathcal{R}z,
\end{equation}
where $\mathcal{F}$ is obtained from $\mathcal{L}$ after deleting the last $l$ rows
and $l$ columns; $\mathcal{R}$ consists of the first $n$ elements of the deleted columns. Since (\ref{ueqfolw}) captures the followers' dynamics, the controllability of a multi-agent system can be realized through (\ref{ueqfolw}).
A path of  $\mathcal{G}$ is a sequence of consecutive edges. $\mathcal{G}$ is connected if there is a path between any distinct nodes. A subgraph of $\mathcal{G}$ is a graph whose vertex set is a subset of $\mathcal{V}$ and whose edge set is a subset of $\mathcal{E}$ restricted to this subset. A subgraph is induced from $\mathcal{G}$ if it is obtained by deleting a subset of nodes and all the edges connecting to those nodes. An induced subgraph, which is maximal and connected, is said to be a connected component. Controllability can be studied under the assumption that $\mathcal{G}$ is connected~\cite{JiWCon}.
Let agents $n+1,\ldots,n+l$ play leaders' role. Define 
\begin{align*}
\mathcal{N}_{kf}\mathop  = \limits^\Delta &  \{ {i}|{v_i} \sim {v_k}, {v_i}~\mbox{is a node of follower subgraph}~ \mathcal{G}_f\},\\
\mathcal{N}_{kl}\mathop  = \limits^\Delta &  \{ {j}|{v_j} \sim {v_k},{v_j}~\mbox{is a node of leader subgraph}~ \mathcal{G}_l\}.
\end{align*}
Then $
{\mathcal{N}_k} = {\mathcal{N}_{kf}} \cup {\mathcal{N}_{kl}},$
$
\mathcal{N}_{kf} \cap \mathcal{N}_{kl}=\Phi,
$
where $\Phi$ is the empty set.  Here to focus on subsequent problem: \emph{identify a number of nodes so that the topology associated with them destroys the controllability of the whole graph.}

\begin{proposition}\label{singPro}
The multi-agent system with single-integrator dynamics (\ref{singmul}) is controllable if and only if there does not exist some $\beta$ such that any of the following statements \emph{i)  ii)  iii)  iv)} is satisfied: 
\begin{itemize}
\item[\emph{i)}] $\beta$ is an eigenvalue of $\mathcal{F}$ associated with eigenvector $y=$ $[y_1,\ldots,y_n]^T$ and $y$ is orthogonal to all columns of $\mathcal{R};$

\item[\emph{ii)}]  $\overline{y}=[y_1,\ldots,y_n,0, \ldots ,0]^T$ is an eigenvector of the Laplacian $\mathcal{L}$ associated with the eigenvalue at $\beta;$

\item[\emph{iii)}] $\mathcal{F}$ and $\mathcal{L}$ share a common eigenvalue at $\beta;$

\item[\emph{iv)}] the following equations hold. 
\begin{align}
{d_k}{y_k} - \sum\limits_{i \in {\mathcal{N}_{kf}}} {{y_i}}  =& \beta {y_k}, \,\, k=1,\ldots,n.\label{firsfoll}\\
\sum\limits_{i \in {\mathcal{N}_{kf}}} {{y_i}}  =& 0,\,\, k=n+1,\ldots,n+l.\label{seclead}
\end{align}
\end{itemize}
\end{proposition}
\begin{proof}
ii) and iii) were proved respectively in \cite{JiWCon} and \cite{JiMengACC2007}. The remaining is to show that the four statements are equivalent. i)$\Leftrightarrow$ii) and ii)$\Leftrightarrow$iii) follow from $\mathcal{L}\bar y =\beta \bar y$ and Theorem 9.5.1 of \cite{Godsil2001}. Next we show ii)$\Leftrightarrow$iv). $\mathcal{L}\overline{y}=\beta \overline{y}$ yields 
$
\mathcal{F}y = \beta y, {\mathcal{R}^T}y = 0,%\label{common2}
$
which respectively leads to (\ref{firsfoll}) and (\ref{seclead}).  
On the contrary, since $y_i=0$ for $i=n+1$, $\ldots,$ $n+l;$
$
\sum_{i \in {\mathcal{N}_{kl}}} {{y_i}}  = 0.
$
Then, by  (\ref{firsfoll}), for  $k=1,\ldots,n,$ 
$
{d_k}{y_k} - \sum_{i \in {\mathcal{N}_k}} {{y_i}}  = {d_k}{y_k} - \sum_{i \in {\mathcal{N}_{kf}}} {{y_i}}  - \sum_{i \in {\mathcal{N}_{kl}}} {{y_i}} 
 = \beta {y_k}.
$
For $k=n+1$ to $n+l,$ since $y_k=0$ and $
\sum_{i \in {\mathcal{N}_{kl}}} {{y_i}}  = 0,$ by (\ref{seclead}),
$
{d_k}{y_k} - \sum\limits_{i \in {\mathcal{N}_k}} {{y_i}} = \beta {y_k}
$
also holds. Thus the eigen-condition is met for each $k,$ i.e., $\mathcal{L}\overline{y}=\beta \overline{y}.$ 
\end{proof}

\section{Controllability destructive nodes} 

%Below the problem is studied via constructing uncontrollable topologies by using the identified controllability destructive nodes. 

\subsection{Double destructive nodes}
%Let $\overline{y} = [0, \ldots ,0,{y_p},0, \ldots ,0,{y_q},0, \ldots ,0]^T,$ $y_p, y_q\neq 0.$ 

\begin{definition}\label{DCDdef}
$v_p$ and $v_q$ are said to be double controllability destructive (DCD) nodes if for any node $v_k$ other than $v_p$ and $v_q,$ $k \in \{ 1, \cdots, n+l\},$ 
$\mathcal{N}_{k}$  contains either both indices $p$ and $q$ or neither of them. 
\end{definition}

\begin{lemma} \label{doulem}
Let $\mathcal{G}$ be a communication graph with leader nodes selected from $\mathcal{V}\setminus\{v_p,v_q\}.$ Then $\bar y = [0,\cdots,0,$ $y_p,0, \cdots 0, y_q,0, \cdots ,0]^T$ with ${y_p},  {y_q} \neq 0$ is an eigenvector of $\mathcal{L}$ if and only if for any $k \neq p, q;$ $k \in \{ 1, \cdots, n+l \};$  $\mathcal{N}_{kf}$ contains either both $p$ and $q$ or neither of them. Moreover, if $p \in \mathcal{N}_{pf},$ ${y_p} =  - {y_q}$ and $d_p=d_q,$ and the corresponding eigenvalue $\lambda  = {d_p} + 1;$ otherwise, $\lambda=d_q.$
\end{lemma}

\begin{proof} 
The special form of $\bar y$ and the selection of leaders lead to $\sum_{i \in {\mathcal{N}_{kl}}} {{y_i}}  = 0.$

(\emph{Necessity})
$\mathcal{L}\bar y=\lambda \bar y$ means %that for any  $k=1,\ldots,n+l$
\begin{equation}
{d_k}{y_k} - \sum\limits_{i \in {\mathcal{N}_k}} {{y_i}}  = \lambda {y_k}, \,\,\,\,k=1,\ldots,n+l.\label{eigcond}
\end{equation}

%\noindent\textbf{Case I.} 
For $k\neq p, q,$ since $y_k=0,$ it follows that 
\begin{equation}
{d_k}{y_k} - \sum\limits_{i \in {\mathcal{N}_k}} {{y_i}}  = \sum\limits_{i \in {\mathcal{N}_{kf}}} {{y_i}} \label{simcondt}
\end{equation} 
Combining (\ref{eigcond}) with (\ref{simcondt}) yields that for any $\lambda$  
\begin{equation}
\sum\limits_{i \in {\mathcal{N}_{kf}}} {{y_i}}  = 0.\label{reequaio}
\end{equation}
$\mathcal{N}_{kf}(k\neq p, q)$ has three cases: i) $p,q\in\mathcal{N}_{kf}.$ In this case, the special form of $\bar y$ implies $\sum_{i \in {\mathcal{N}_{kf}}} {{y_i}}  = {y_p} + {y_q} .$ By (\ref{reequaio}),  $y_p=-y_q.$ ii) only $p(\mbox{or}\, q)\in\mathcal{N}_{kf}$. Then $\sum_{i \in {\mathcal{N}_{kf}}} {{y_i}}  = {y_p}(\mbox{or}\, y_q) \neq 0$. This case cannot occur since (\ref{reequaio}) is not met. iii) $p,q\notin\mathcal{N}_{kf}.$ In this case, $\sum_{i \in {\mathcal{N}_{kf}}} {{y_i}}  = 0.$ Thus there exists at least one $k\neq p, q$ with $p,q\in\mathcal{N}_{kf}.$ Otherwise, for any $k\neq p, q,$ the above discussion means $p,q\notin\mathcal{N}_{kf}.$ That is, $v_p, v_q$ are isolated from all the other nodes, which contradicts with the connectedness of $\mathcal{G}.$ So, if $\bar y$ is an eigenvector of $\mathcal{L},$ then for any $k\neq p, q,$ either $p,q\in\mathcal{N}_{kf};$ or $p,q\notin\mathcal{N}_{kf}.$ If $p,q\in\mathcal{N}_{kf}$ occurs, $y_p=-y_q.$ 

For $k= p, q,$ (\ref{eigcond}) and $\sum_{i \in {\mathcal{N}_{kl}}} {{y_i}}  = 0$ yield that
\begin{equation}
({d_k} - \lambda ) \cdot {y_k} = \sum\limits_{i \in {\mathcal{N}_{kf}}} {{y_i}},\,\,\,\, k= p, q.\label{sencase}
\end{equation}
If $p\in\mathcal{N}_{qf},$ then $\sum_{i \in {\mathcal{N}_{qf}}} {{y_i}}  = {y_p}.$ By (\ref{sencase}), $({d_q} - \lambda ) {y_q} =y_p=-y_q.$ So $y_q\neq 0$ results in
$
\lambda  = {d_q} + 1.
$
Since $\mathcal{G}$ is undirected, $p\in\mathcal{N}_{qf}$ is equivalent to $q\in\mathcal{N}_{pf}.$ The same arguments show 
$
\lambda  = {d_p} + 1.
$
As a consequence, $d_p=d_q.$ 
If $p\notin\mathcal{N}_{qf},$ $\sum_{i \in {\mathcal{N}_{qf}}} {{y_i}}  = 0$ follows from  the special form of $\bar y.$ Thus 
$
{d_q}{y_q} - \sum_{i \in {\mathcal{N}_q}} {{y_i}}  = {d_q}{y_q}.
$
By (\ref{eigcond}),  
$
{d_q}{y_q} = \lambda {y_q}.
$
Since $y_q\neq 0,$ 
$
\lambda=d_q.
$
Similar arguments to $q\notin\mathcal{N}_{pf}$ yields $\lambda=d_p.$ The necessity proof is completed.

(\emph{Sufficiency}) For $p\notin\mathcal{N}_{qf},$ if $p,q\in\mathcal{N}_{kf}(k\neq p, q),$ then 
\begin{align}
{d_k}{y_k} - \sum\limits_{i \in {\mathcal{N}_k}} {{y_i}}  =& {d_k}\cdot 0 - \sum\limits_{i \in {\mathcal{N}_{kf}}} {{y_i}}  - \sum\limits_{i \in {\mathcal{N}_{kl}}} {{y_i}}\nonumber \\
 =& -(y_p+y_q),\,\,\,\,\,\, k\neq p, q. \label{leafir}
\end{align}
$y_p=-y_q$ is required to satisfy the eigen-condition in (\ref{eigcond}) for the eigenvalue at $\lambda=d_p.$ 
Since $p,q\in\mathcal{N}_{kf}$ occurs at least for one $k\neq p,q$ (otherwise $\mathcal{G}$ is not connected), $y_p=-y_q$ is a prerequisite for $\bar y$ to be an eigenvector of $\mathcal{L}.$ 
If $p,q\notin\mathcal{N}_{kf} (k\neq p, q),$ then $\sum_{i \in {\mathcal{N}_{kf}}} y_i  =\sum_{i \in {\mathcal{N}_{kl}}} y_i  =0.$ The eigen-condition also holds for any number $\lambda.$
When $k=p,q,$ since the valency of $v_p$ and $v_q$ is equal, $d_p=d_q.$ It  follows from $p\notin\mathcal{N}_{qf}, q\notin\mathcal{N}_{pf}$ that $\sum_{i \in {\mathcal{N}_{kl}}} {{y_i}}=\sum_{i \in {\mathcal{N}_{kf}}} {{y_i}}=0(k=p,q).$ Then
$
{d_k}{y_k} - \sum_{i \in {\mathcal{N}_{k}}} {{y_i}}  = d_k y_k-0=\lambda y_k; k=p,q,
$
where $\lambda=d_p=d_q.$ Hence, with given leaders, the eigen-condition 
is met for each $k=1,\ldots, n+l.$ 
Thus $\overline{y}$ is an eigenvector of $\mathcal{L}$ with the eigenvalue at $\lambda=d_p.$ 

For $p\in\mathcal{N}_{qf}.$ $
\sum_{i \in {\mathcal{N}_{pl}}} {{y_i}}=0, \sum_{i \in {\mathcal{N}_{pf}}} {{y_i}}=y_q.
$
Therefore
$
{d_p}{y_p} - \sum_{i \in {\mathcal{N}_p}} {{y_i}}  
 = (\lambda+1) {y_p},
$
where $\lambda=d_p=d_q.$
Similarly, 
$
{d_q}{y_q} - \sum_{i \in {\mathcal{N}_q}} {{y_i}}=(\lambda+1) {y_q}.
$
The remaining proof is in the same vein as that of $p\notin\mathcal{N}_{qf}$ with the eigenvalue $\lambda$ replaced by $\lambda+1.$ 
\end{proof}

\begin{theorem}\label{doubDcd}
There exist a group of leaders selected from $\Gamma _{p,q}$ such that the multi-agent system with single-integrator dynamics (\ref{singmul}) is controllable if and only if the follower node set does not contain DCD nodes $v_p$ and $v_q,$ where $p\neq q;$ ${\Gamma _{p,q}}\mathop  = \limits^\Delta  \{ 1, \ldots,$ $n + l\}\setminus \{ p,q\}.$
\end{theorem}
\begin{proof}
(\emph{Necessity}) Suppose by contradiction that the follower subgraph $\mathcal{G}_f$ contains DCD nodes $v_p, v_q.$ Lemma \ref{doulem} shows that $\mathcal{L}$ has an eigenvector 
$\bar y = [0, \cdots ,0,$ $y_p,0, \cdots 0, y_q,0, \cdots ,0]^T$ with ${y_p} =  - {y_q} \neq 0.$
By Proposition \ref{singPro}, system (\ref{singmul}) is uncontrollable with any leaders selected from $\Gamma _{p,q}.$ This contradicts the assumption. 

(\emph{Sufficiency}) Suppose by contradiction that the system is uncontrollable with any leaders selected from ${\Gamma _{p,q}}.$ Then the system is uncontrollable with all the elements of $\Gamma _{p,q}$ playing leaders' role. By Proposition \ref{singPro}, $\mathcal{L}$ has an eigenvector 
$\bar y = {[0, \cdots ,0,y_p,0, \cdots 0, y_q,0, \cdots ,0]^T}.$ Next we show $y_p, y_q\neq 0.$
Suppose by contradiction $y_p=0,$ then $y_q\neq 0$ because $\bar y$ is an eigenvector. Since the graph is connected, $\lambda=0$ is a simple eigenvalue associated with the all-one eigenvector $\textbf{1}.$ Thus the eigenvalue $\beta$ associated with $\bar y$ is not zero. In addition, there exist at least one $k\neq q$ with $k\in\mathcal{N}_q;$ otherwise, $v_q$ will be isolated from all the other nodes. The special form of $\bar y$ then results in $\sum\nolimits_{i \in {\mathcal{N}_{kl}}} {{y_i}}  = 0,\sum\nolimits_{i \in {\mathcal{N}_{kf}}} {{y_i}}  = {y_q}.$ Since $y_k=0,$ 
$
{d_k}{y_k} - \sum_{i \in {\mathcal{N}_k}} {{y_i}}  
 =-y_q.
$
The eigen-condition in (\ref{eigcond}) is not met for $v_k$ since $y_k=0$ and $y_q\neq 0.$ This contradicts with the fact that $\bar y$ is an eigenvector. Therefore $y_p\neq 0.$ Similar arguments yield $y_q\neq 0.$ Finally, it follows from Lemma \ref{doulem} that $v_p$ and $v_q$ are DCD nodes since $\bar y$ with $y_p, y_q\neq 0$ is an eigenvector of $\mathcal{L}$. This is in contradiction with the assumption. The proof is completed.   
\end{proof}

\subsection{Triple destructive nodes}

%Below leaders are arbitrarily selected from $\mathcal{V}\setminus\{v_p,v_q,v_r\}.$

\begin{definition}\label{TCDdef}
$v_p,v_q,v_r$ are said to be triple controllability destructive (TCD) nodes if for any $v_k$ other than $v_p,v_q,v_r;$ $\mathcal{N}_{kf}$ contains either all $p,q,r$ or none of them; and for $v_p,v_q,v_r,$ any of the following four cases occurs:
\begin{itemize}
\item for any $k\in\{p,q,r\},$ $\mathcal{N}_{kf}$ contains the other two indices of $p, q, r;$

\item there is a $k\in\{p,q,r\}$(say $k=p$) with $\mathcal{N}_{pf}$ containing $q,r$ and each of $\mathcal{N}_{qf},\mathcal{N}_{rf}$ contains only $p$ in $\{p,q,r\};$

\item there is a $k\in\{p,q,r\}$(say $k=p$) with $\mathcal{N}_{kf}$ containing one and only one of the other two indices of $p, q, r;$ and its single neighbor node of $p, q, r$(say $q$) also has $k$ as its single neighbor node in $\{p, q, r\};$

\item for any $k\in\{p,q,r\},$ $\mathcal{N}_{kf}$ contains none of $p,$ $q$ and $r.$ 
\end{itemize}
\end{definition}
\begin{remark}
Definition \ref{TCDdef} has no limit as to whether $\mathcal{N}_{kf}$ contains an index $l (l\neq p,q,r).$ It identifies four topologies I to IV(see Fig.\ref{tridet1and2}) which correspond to, respectively, the above first to fourth case of $\mathcal{N}_{kf}$ of $v_p, v_q, v_r.$ 
\end{remark}
\begin{figure}[H]
\begin{center}
\subfigure[]{\includegraphics[width=1.62cm]{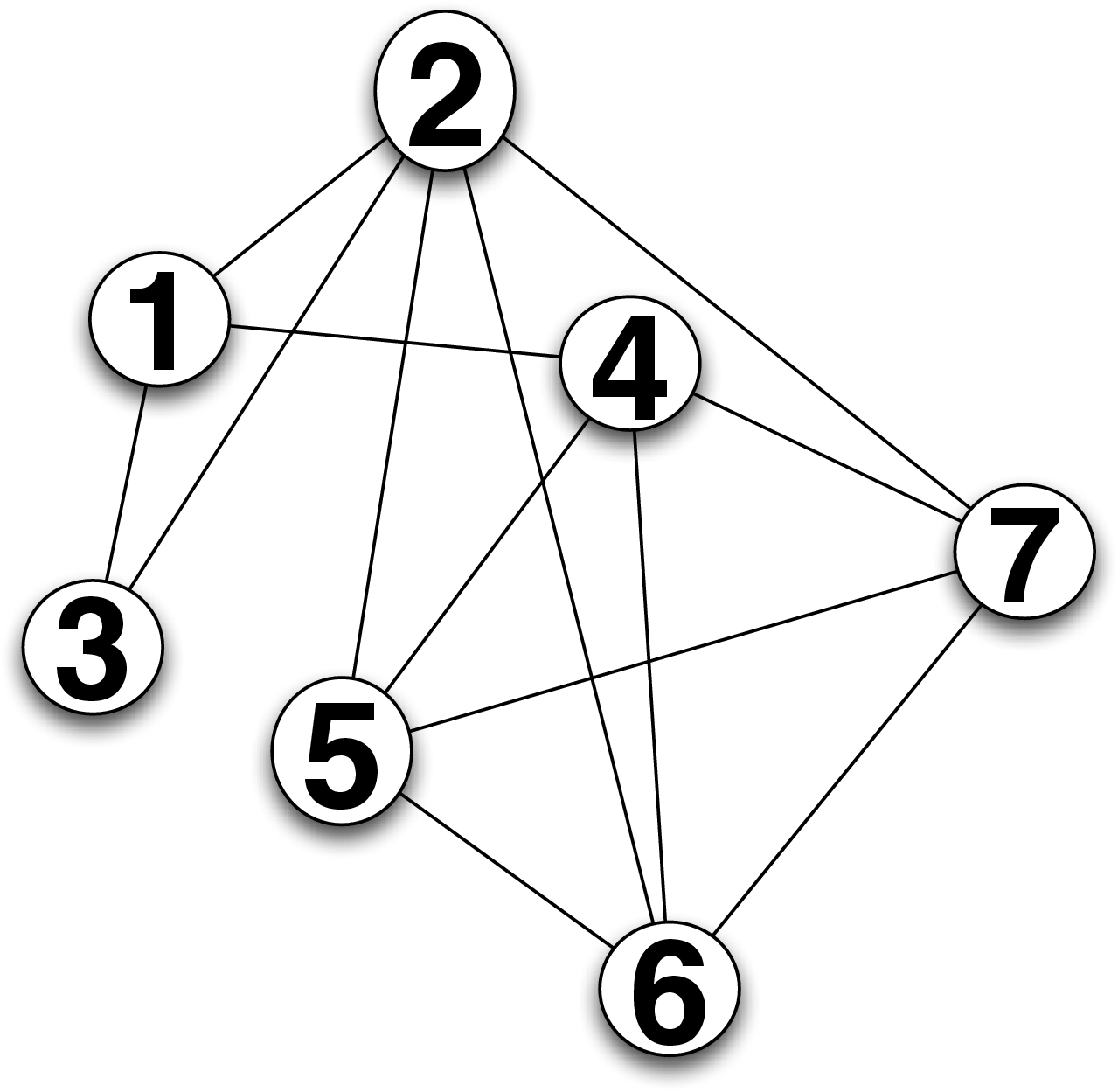}}
\subfigure[]{\includegraphics[width=1.62cm]{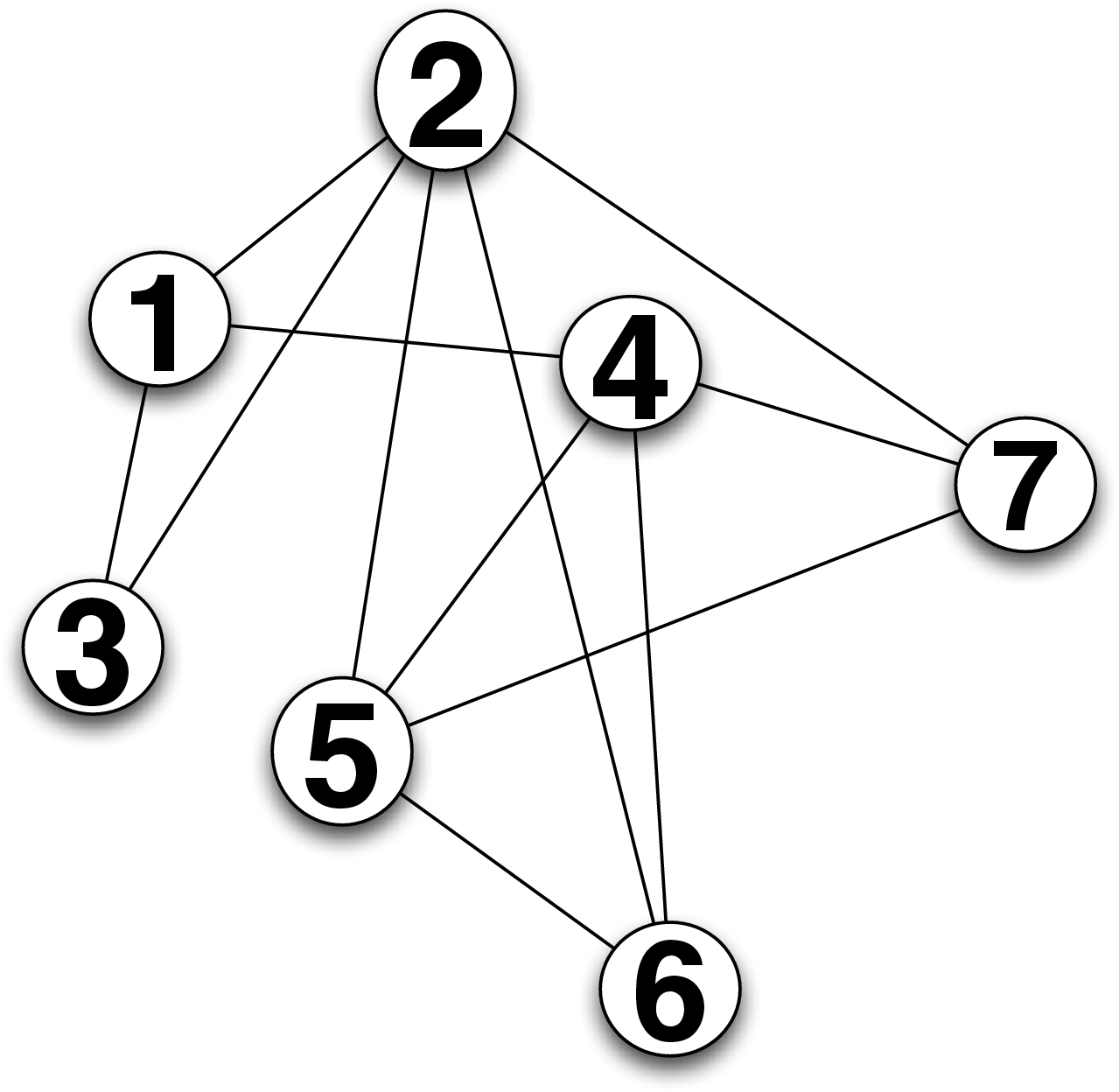}}
\subfigure[]{\includegraphics[width=1.62cm]{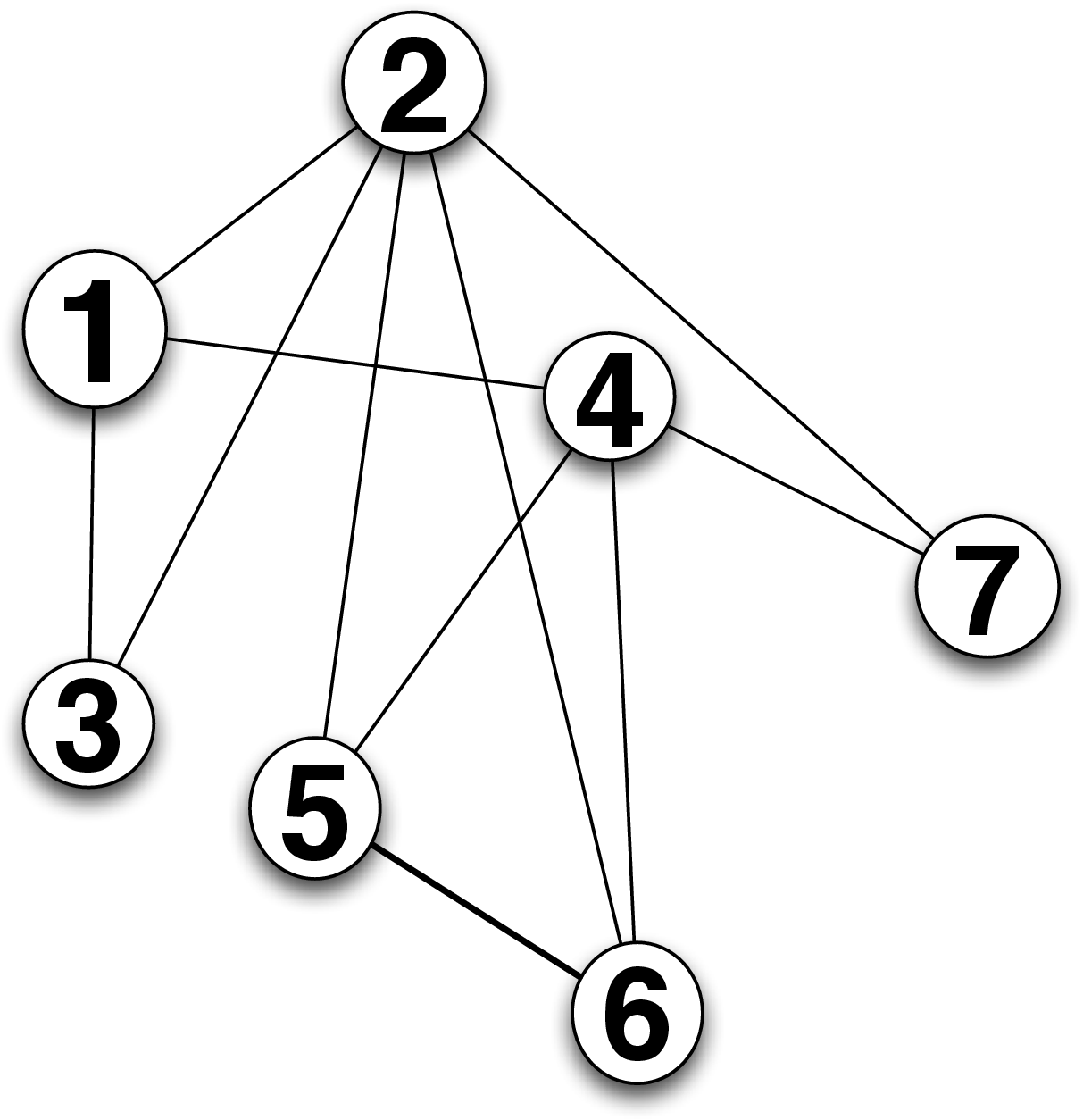}}
\subfigure[]{\includegraphics[width=1.62cm]{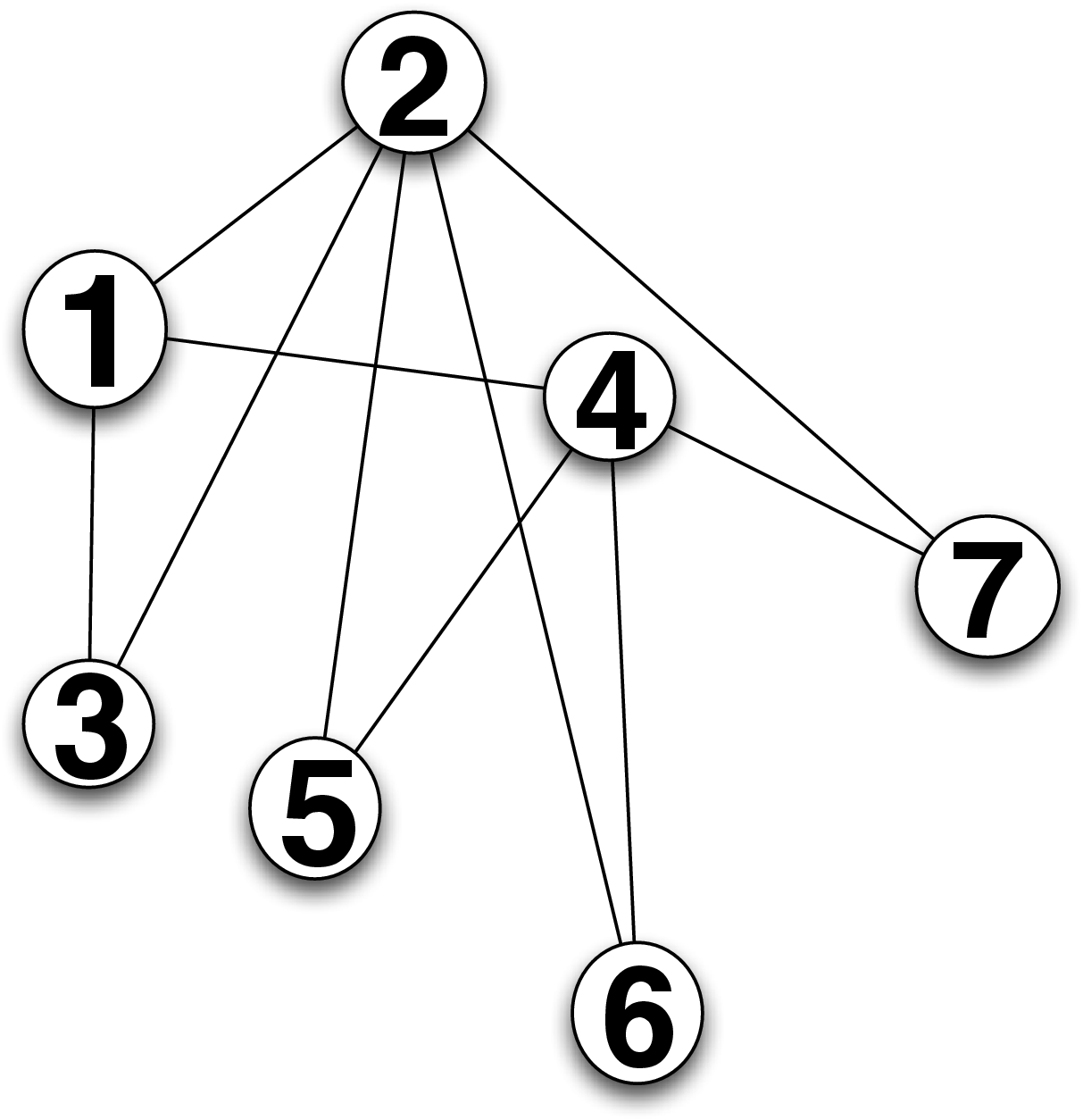}}
\caption{(a),(b),(c),(d) are respectively a topology I,II,III and IV with $v_5, v_6, v_7$ being the TCD nodes.}
\label{tridet1and2}
\end{center}
\end{figure}
%The following lemma is crucial to the development of the main result.

\begin{lemma} \label{trilem}
Let $\mathcal{G}$ be a communication graph with leader nodes arbitrarily selected from $\mathcal{V}\setminus\{v_p,v_q,v_r\}.$ Then $\bar y = [0, \ldots ,{y_p},0, \ldots ,{y_q},0, $ $\ldots ,{y_r},$ $0, \ldots ,0]^T$ with $y_p, y_q, y_r\neq 0$ and all the other elements being zero is an eigenvector of $\mathcal{L}$ if and only if $v_p,v_q,v_r$ are TCD nodes. Moreover, $y_p+y_q+y_r=0, y_k\neq 0, k=p, q, r,$ and 
\begin{itemize}
\item for topology I, $d_p=d_q=d_r$ and the corresponding eigenvalue is $d_p+1;$

\item for topology II, $y_q=y_r,$ $d_p=d_q+1=d_r+1$ and the corresponding eigenvalue is $d_p+1;$

\item  for topology III, $y_p=y_q,$ $d_p=d_q=d_r+1$ and the corresponding eigenvalue is $d_r;$

\item  for topology IV, $d_p=d_q=d_r$ and the corresponding eigenvalue is $d_r.$
\end{itemize}
\end{lemma}

\begin{proof} 
As in Lemma \ref{doulem}, $\sum_{i \in {\mathcal{N}_{kl}}} {{y_i}}  = 0$ for any $k.$ 

(\emph{Necessity})
The eigen-condition in (\ref{eigcond}) is met for each $k.$ 
\noindent\textbf{Case I.} $k\neq p, q, r.$ In this case, $y_k=0.$ Then  
\begin{equation}
{d_k}{y_k} - \sum\limits_{i \in {\mathcal{N}_k}} {{y_i}} = -\sum\limits_{i \in {\mathcal{N}_{kf}}} {{y_i}} \label{triplesimcond}
\end{equation}
Combining (\ref{eigcond}) with (\ref{triplesimcond}) yields 
\begin{equation}
\sum\limits_{i \in {\mathcal{N}_{kf}}} {{y_i}}  = 0.\label{trireequaio}
\end{equation}
Each $\mathcal{N}_{kf} (k\neq p,q,r)$ falls into one of the four cases. 
\begin{itemize}
\item[a)] $p,q,r\in\mathcal{N}_{kf}.$ Since $\sum_{i \in {\mathcal{N}_{kf}}} {{y_i}}  = {y_p} + {y_q}+y_r,$ by (\ref{trireequaio}) 
\begin{equation}
{y_p} + {y_q}+y_r=0.\label{tripleequ}
\end{equation}

\item[b)] any two and only two of $p,q,r$ belong to $\mathcal{N}_{kf}.$ Suppose $p,q\in\mathcal{N}_{kf},$ then $\sum_{i \in {\mathcal{N}_{kf}}} {{y_i}}  = {y_p}+y_q.$ By (\ref{trireequaio})\begin{equation}
{y_p} + {y_q}=0.\label{doubelequ}
\end{equation}
(\ref{tripleequ}) and (\ref{doubelequ}) cannot be met simultaneously, or else, $y_r=0.$ This contradicts with $y_r\neq 0.$ 
If there is another $k\neq p,q,r$ with $\mathcal{N}_{kf}$ containing $p,r,$ by (\ref{trireequaio}) 
\begin{equation}
{y_p} + {y_r}=0.\label{doubelano}
\end{equation}
From (\ref{doubelequ}) (\ref{doubelano}), $y_p=-y_q=-y_r.$ If (\ref{doubelequ}) (\ref{doubelano}) are met simultaneously, there does not exist the third $k\neq p, q, r$ with $\mathcal{N}_{kf}$ containing $q,r.$ Otherwise, 
\begin{equation}
{y_q} + {y_r}=0.\label{doubthird}
\end{equation}
This however is impossible because ${y_q} + {y_r}=0$ and $y_p=-y_q=-y_r$ lead to $y_q=y_r=0,$ which is incompatible with $y_k\neq 0, k=p,q,r.$ Hence, at most two of (\ref{doubelequ}), (\ref{doubelano}) and (\ref{doubthird}) take place. 

\item[c)] any one and only one of $p,q,r$ belongs to $\mathcal{N}_{kf},$ say $p\in\mathcal{N}_{kf}.$ In this case, $\sum_{i \in {\mathcal{N}_{kf}}} {{y_i}}  = {y_p}.$ To satisfy (\ref{trireequaio}), it requires  
$
{y_p} =0.
$
This is impossible because $y_p\neq 0.$ 

\item[d)] none of $p,q,r$ belongs to $\mathcal{N}_{kf}.$ In this case, the special form of $\bar y$ implies $\sum_{i \in {\mathcal{N}_{kf}}} {{y_i}}  = 0,$ i.e., (\ref{trireequaio}) is met. 
 \end{itemize}
Since (\ref{tripleequ}) (\ref{doubelequ}) cannot be met simultaneously, a) and b) cannot occur at once. That is, there do not exist different $v_{k_1}, v_{k_2}$ in $\mathcal{G}$ with $v_{k_1}$ and $v_{k_2}$ consistent with cases a) and b), respectively. 
Thus, with given $k\neq p, q, r;$ $\mathcal{N}_{kf}$ conforms to one and only one of the following cases: {\bf{i)}} at least one of a), d) occurs; {\bf{ii)}} at least one of b), d) occurs.

%\begin{itemize}
%\item[i)] at least one of a) and d) occurs. 
%
%\item[ii)] at least one of b) and d) occurs.
%\end{itemize}

\noindent\textbf{Case II.} $k=p, q, r.$ Since $\sum_{i \in {\mathcal{N}_{kl}}} {{y_i}}  = 0,$ by (\ref{eigcond})
\begin{equation}
({d_k} - \lambda ) {y_k} = \sum\limits_{i \in {\mathcal{N}_{kf}}} {{y_i}}. \label{trisencase}
\end{equation}
%There are three circumstances for $\mathcal{N}_{kf}.$ 
\begin{itemize}
\item[1)] There is at least one
$k\in\{p, q, r\}$ with $\mathcal{N}_{kf}$ containing the other two indices of $p, q, r.$    
{\bf{1a)}} only one $k\in\{p,q,r\}$ is of this kind. 
{\bf{1b)}} there are two $k's\in\{p,q,r\}$ of this kind. (a) (b) of Fig. \ref{pic1a1b} correspond to 1a) and 1b), respectively. {\bf{1c)}} each $k\in\{p, q, r\}$ is of this kind. Note that 1b) and 1c) are equivalent.

\item[2)] There is at least one
$k\in\{p, q, r\}$ with $\mathcal{N}_{kf}$ containing one and only one of the other two indices of $\{p, q, r\}.$ {\bf{2a)}} only one $k\in\{p,q,r\}$(say $k=p$) is of this kind and its single neighbor node in $\{p, q, r\},$ say $q,$ also has $k$ as its single neighbor node in $\{p, q, r\}.$ {\bf{2b)}} there are two $k's\in\{p,q,r\}$ of this kind. 1a) coincides with 2b). That each $k\in\{p,q,r\}$ is of this kind does not occur. 

\item[3)] For each $k=p,q,r;$ $\mathcal{N}_{kf}$ contains none of the other two indices of $p,q,r.$ {\bf{3a)}} only one $k\in\{p,q,r\}$ is of this kind, which coincides with 2a).
{\bf{3b)}} there are two $k's\in\{p,q,r\}$ of this kind (see (d) of Fig. \ref{pic1a1b}).  
\end{itemize}
\begin{figure}[H]
\begin{center}
\subfigure[]{\includegraphics[width=1.7cm]{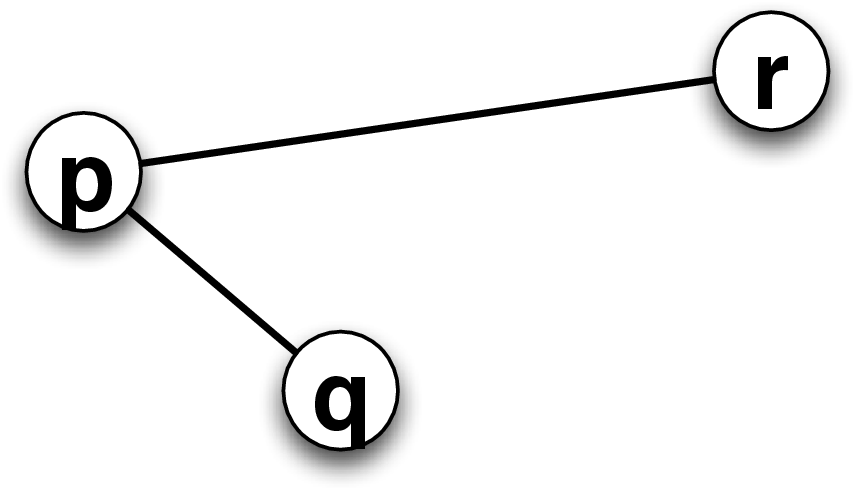}}
\subfigure[]{\includegraphics[width=1.7cm]{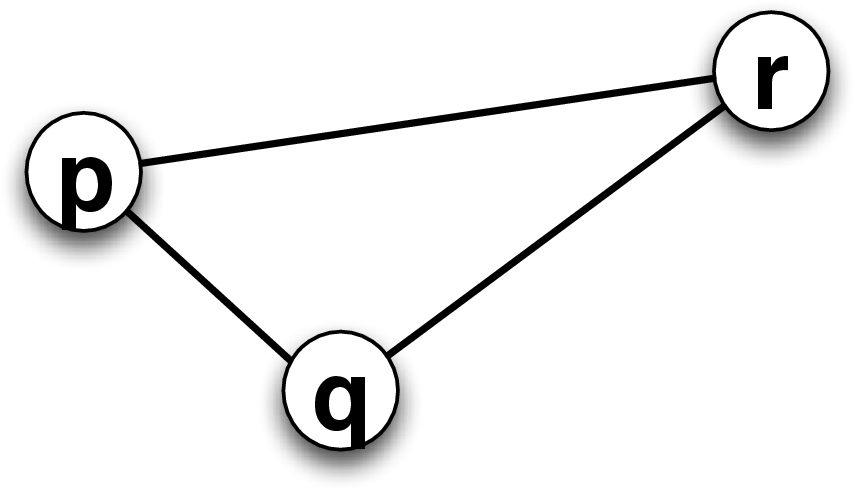}}
\subfigure[]{\includegraphics[width=1.7cm]{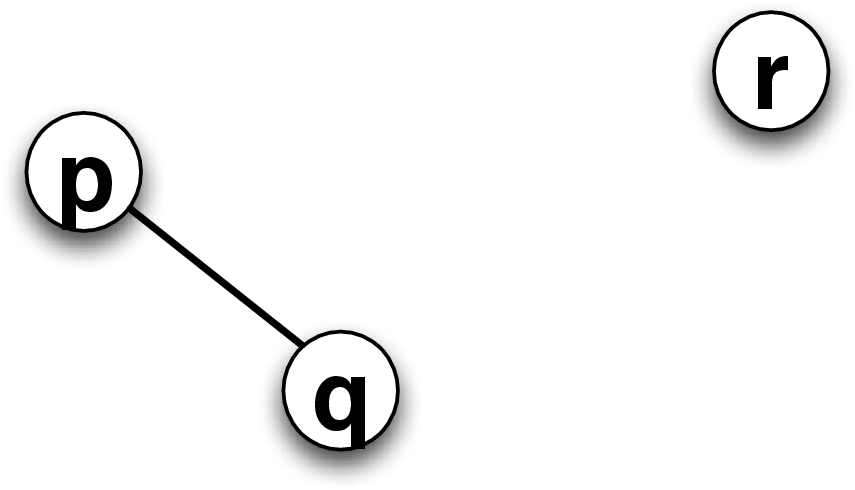}}
\subfigure[]{\includegraphics[width=1.7cm]{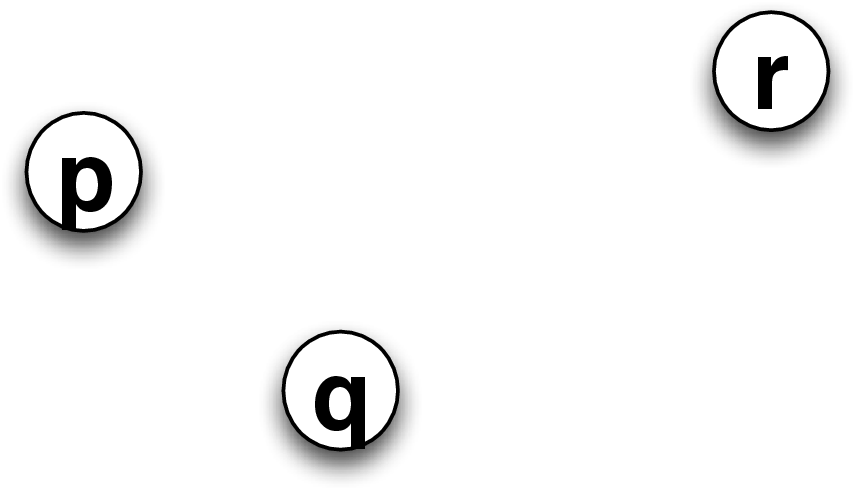}}
\subfigure[]{\includegraphics[width=1.9cm]{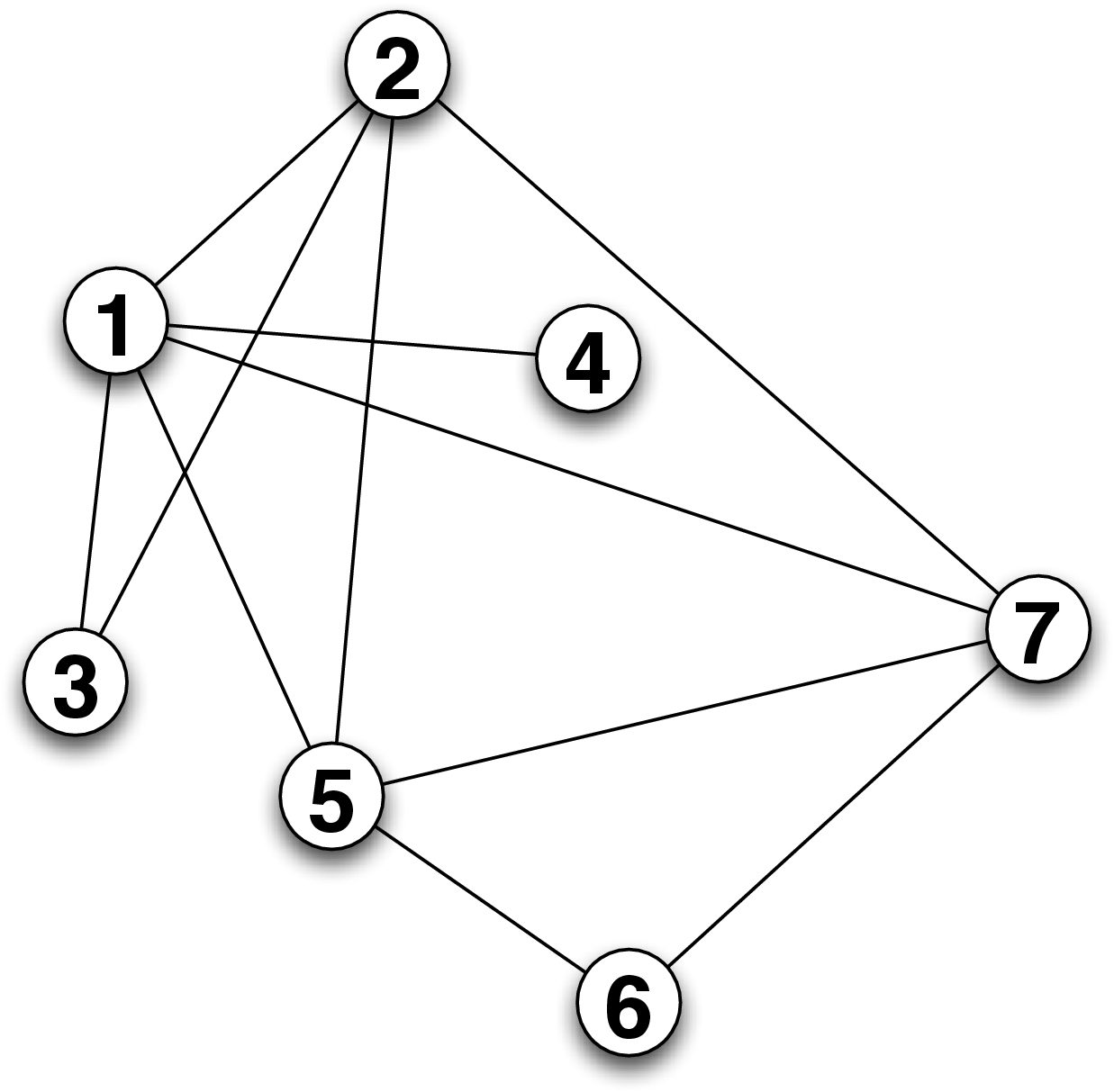}}
\subfigure[]{\includegraphics[width=1.9cm]{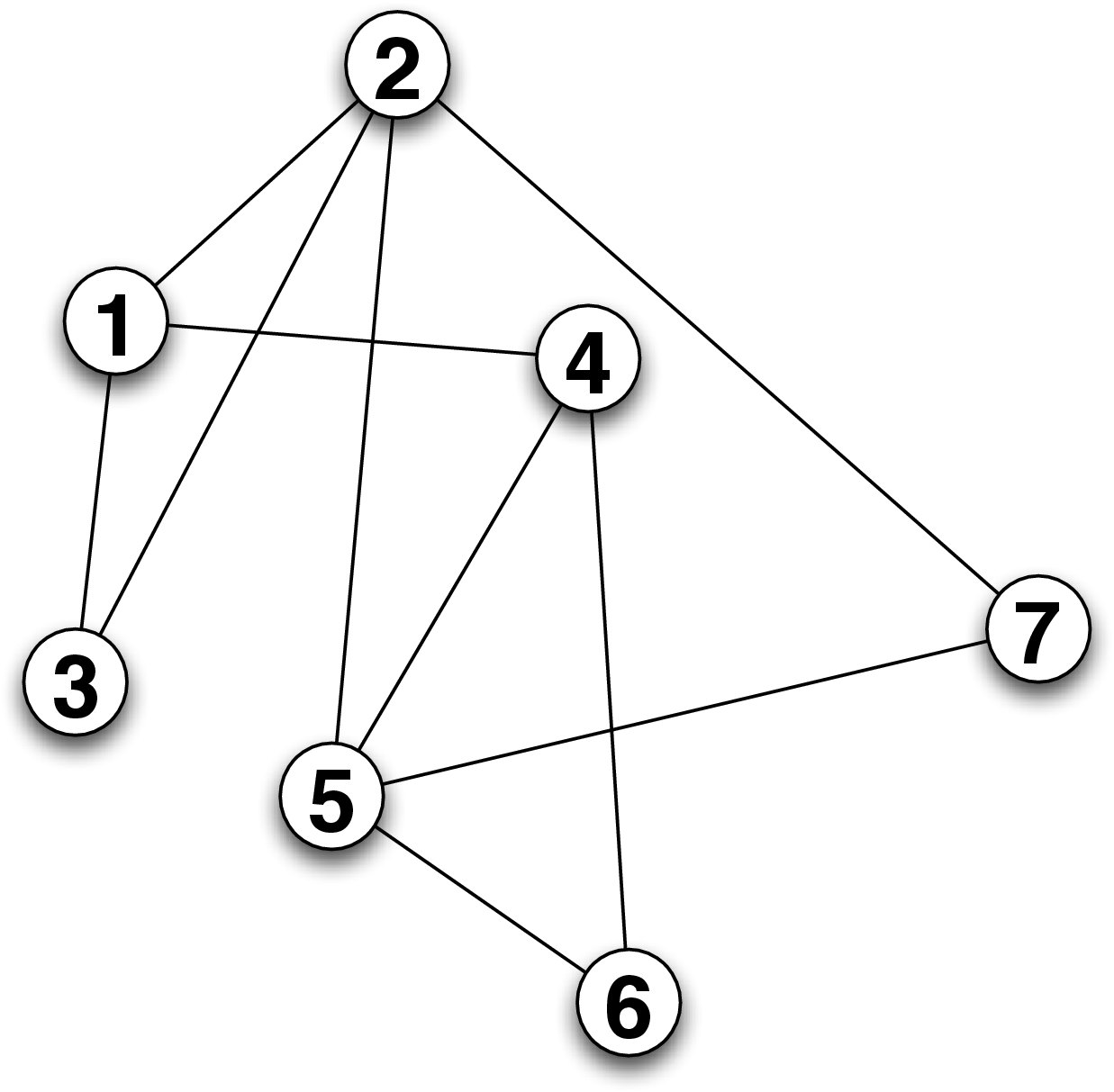}}
\subfigure[]{\includegraphics[width=1.9cm]{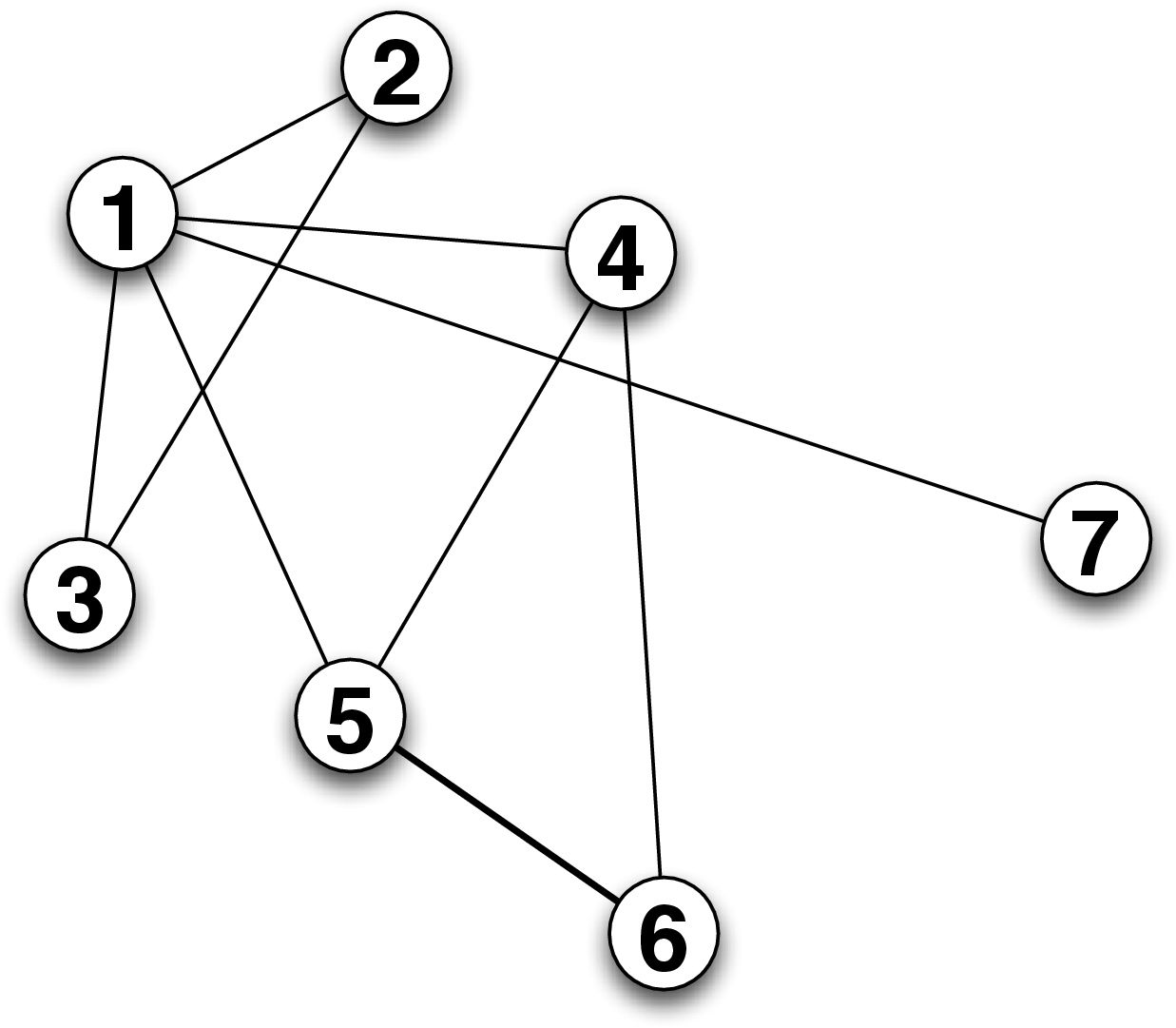}}
\subfigure[]{\includegraphics[width=1.9cm]{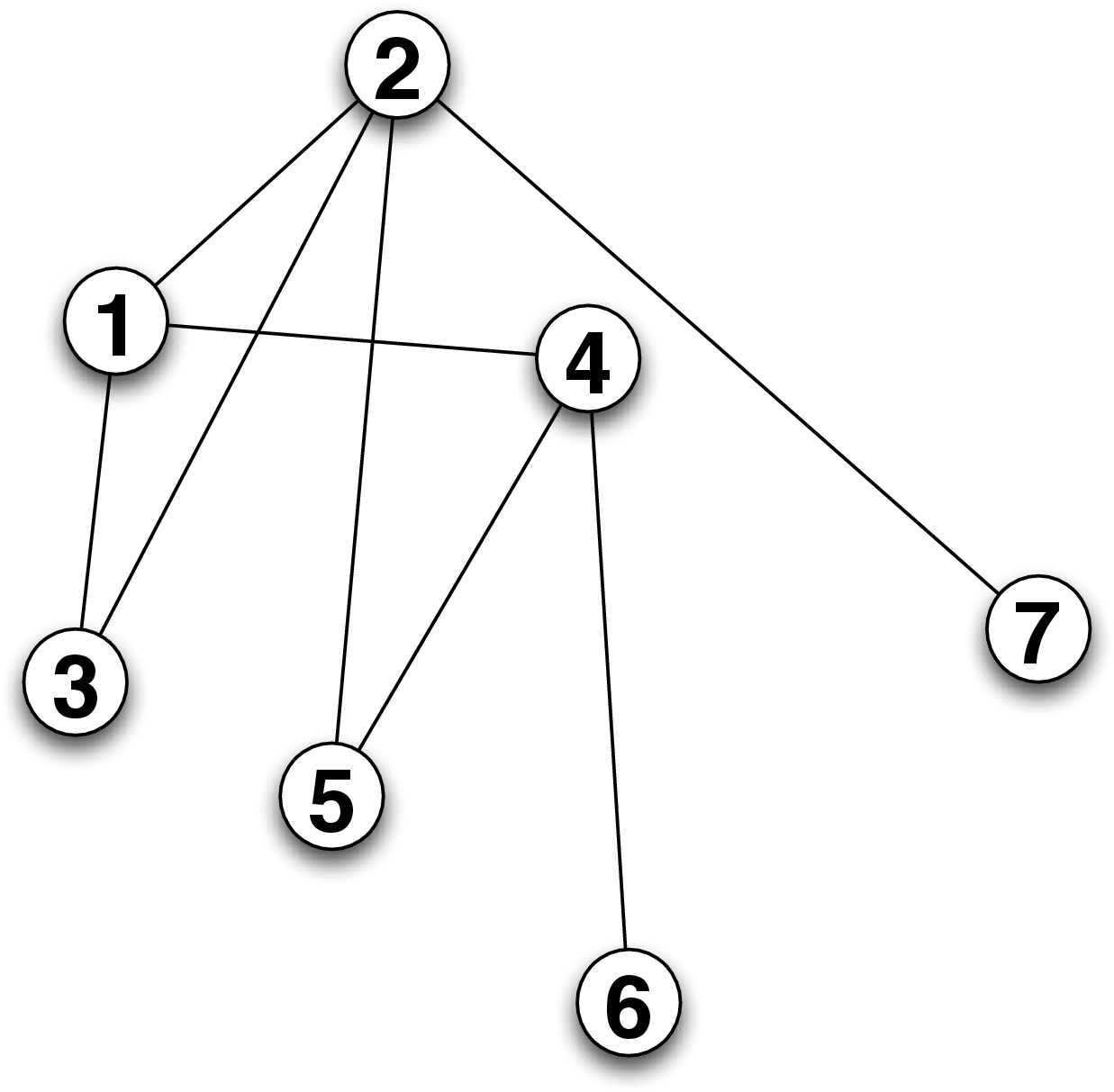}}
\caption{(a)(b)(c)(d), respectively, correspond to 1a) with $k=p;$ 1b)(or 1c)); 2a) with $k=p$(or $q$) and 3b). (e)(f)(g)(h) are respectively topologies V,VI,VII and VIII with $v_5, v_6, v_7$ the destructive nodes}
\label{pic1a1b}
\end{center}
\end{figure}
Item {\bf{i)}} of Case I,  together with 1b), 1a), 2a), 3b) of Case II, respectively, results in topologies I to IV (see Fig. \ref{tridet1and2}). 
If the `item {\bf{i)}} of Case I' is replaced by `item {\bf{ii)}} of Case I', then topologies V to VIII are generated(see (e) to (h) of Fig. \ref{pic1a1b}). So, if $\bar y$ is an eigenvector of $\mathcal{L},$ then $v_p,v_q,v_r$ have maximum of eight possible topologies. Moreover, it will be shown that topologies V to VIII are redundant. 
\begin{fact}\label{factvprf}
If $\bar y$ is an eigenvector of $\mathcal{L},$ then $v_p,v_q,v_r$ cannot have topology structures V, VI, VII and VIII. 
\end{fact}
\textbf{Case 1.}
$k\neq p, q, r.$ It is to be proved by contradiction first for V. In this case, (\ref{triplesimcond}) holds.
Since the graph is connected, one of $v_p,v_q,v_r,$ say $v_q$ in subsequent arguments, must have a neighbor in $\mathcal{V}\setminus\{v_p,v_q,v_r\}.$ By the topology structure of V, there is a node of $v_p, v_q, v_r,$ say $v_p$ with $v_p, v_q$ sharing at least one common neighbor node in  $\mathcal{V}\setminus\{v_p,v_q,v_r\}.$  Suppose this node is $v_k,$  then $p,q\in\mathcal{N}_{kf}.$ Since a) of Case I does not arise, $\sum_{i \in {\mathcal{N}_{kf}}} {{y_i}}  = {y_p} + {y_q}.$ Then by (\ref{eigcond}) and (\ref{triplesimcond}), (\ref{doubelequ}) holds. 
%\begin{equation}
%{y_p} + {y_q}=0.\label{pqeicon}
%\end{equation}
Now there are two situations for $v_p,v_r.$ One is that there is another $k\neq p,q,r$ with $v_k$ incident to 
both $v_p$ and $v_r;$ the other is that none of $v_k (k\neq p,q,r)$ is incident to both $v_p$ and $v_r$. For the first situation, similar arguments to (\ref{doubelequ}) yield that the eigen-condition requires (\ref{doubelano}) to be true. 
%\begin{equation}
%{y_p} + {y_r}=0.\label{preicon}
%\end{equation}
$\{v_p,v_q\}$ and $\{v_p,v_r\}$ cannot be incident to the same $v_k (k\neq p,q,r)$ because  a) of Case I does not arise in topology V. 
For $k\neq p,q,r,$ with $\mathcal{N}_{kf}$ containing none of $p,q,r;$ $\sum_{i \in {\mathcal{N}_{kf}}} {{y_i}}  = 0.$ It follows from $y_k=0 (k\neq p,q,r)$ and (\ref{triplesimcond}) that for these $k's$ the eigen-condition (\ref{eigcond}) is met.

\textbf{Case 2.} $k= p, q, r.$ Let us first consider the first situation of $v_p, v_r.$ Since $\sum_{i \in {\mathcal{N}_{kl}}} {{y_i}}  = 0,$ one has
\begin{equation}
{d_k}{y_k} - \sum\limits_{i \in {\mathcal{N}_k}} {{y_i}}  = {d_k}  y_k - \sum\limits_{i \in {\mathcal{N}_{kf}}} {{y_i}}.\label{caengentric}
\end{equation}
In topology V, each $\mathcal{N}_{kf} (k=p,q,r)$ contains two indices of $p, q, r,$ which are different from $k.$ Thus, for a $k\in\{p, q, r\},$ say $k=p,$   
$
\sum_{i \in {\mathcal{N}_{kf}}} {{y_i}}  = {y_q} + {y_r}. 
$
By (\ref{doubelequ}) and (\ref{doubelano}), $y_p=-y_q=-y_r.$ So 
$
\sum_{i \in {\mathcal{N}_{kf}}} {{y_i}}  = -2{y_p}.
$
By (\ref{caengentric})
\begin{equation}
{d_p}{y_p} - \sum\limits_{i \in {\mathcal{N}_p}} {{y_i}}  = (d_p+2) y_p\label{dpplus2}
\end{equation}
Thus, for $k=p,$ the eigen-condition is met for $\lambda=d_p+2.$
For $k=q,$  
$
\sum_{i \in {\mathcal{N}_{qf}}} {{y_i}}  = {y_p} + {y_r}=0.
$
From (\ref{caengentric}) 
\begin{equation}
{d_q}{y_q} - \sum\limits_{i \in {\mathcal{N}_q}} {{y_i}}  = {d_q} y_q.\label{dqeicon}
\end{equation}
Similarly, for $k=r,$
$
\sum_{i \in {\mathcal{N}_{rf}}} {{y_i}}  = {y_p} + {y_q}=0.
$
Thus
\begin{equation}
{d_r}{y_r} - \sum\limits_{i \in {\mathcal{N}_r}} {{y_i}}  = {d_r} y_r.\label{dreicon}
\end{equation}
To satisfy (\ref{dpplus2}), (\ref{dqeicon}) and (\ref{dreicon}) simultaneously, it requires $d_p+2=d_q=d_r.$ 
Below shows that this is not possible. 
If there is a node $v_{h}$ in $\mathcal{V}\setminus\{v_p,v_q,v_r\}$ which is incident to both $v_q$ and $v_r,$ 
then (\ref{doubthird}) should also be met. However the arguments of b) of Case I show that (\ref{doubelequ})(\ref{doubelano})(\ref{doubthird}) cannot be satisfied simultaneously. Hence this cannot be happening. In this situation, to satisfy $d_q=d_r,$ the number of $v_k$ in $\mathcal{V}\setminus\{v_p,v_q,v_r\}$ which is incident to both $v_p$ and $v_q$ is required to be equal to the number of $v_h$ in $\mathcal{V}\setminus\{v_p,v_q,v_r\}$ which is incident to both $v_p$ and $v_r,$ where $k\neq h.$ As a consequence, $d_p\geq d_q.$ Accordingly $d_p+2> d_q.$ Hence (\ref{dpplus2})(\ref{dqeicon})(\ref{dreicon}) cannot be met at the same time, and accordingly $\bar y$ is not an eigenvector of Laplacian. This contradicts with the assumption. 

Next, for the second situation of $\{v_p,v_r\},$ i.e., none of $v_k (k\neq p,q,r)$ is incident to both $v_p$ and $v_r,$ (\ref{doubelequ}) still holds. In this situation, we further distinguish between two circumstances: one is that there is a $v_k\in\mathcal{V}\setminus\{v_p,v_q, v_r\}$ which is incident to both $v_q$ and $v_r,$ the other is the reversal. For the first circumstance, relabelling $v_p$ as $v_q$ and vice-versa, the proof is the same as that of the aforementioned first situation of $\{v_p, v_r\}.$
For the second circumstance, it can be seen that $d_p=d_q.$ By (\ref{caengentric}) and (\ref{doubelequ}), 
$
{d_r}{y_r} - \sum_{i \in {\mathcal{N}_r}} {{y_i}}  = {d_r}{y_r} - ({y_p} + {y_q})
 = {d_r}{y_r}.
$
Hence, to satisfy the eigen-condition, it requires 
$
\lambda  = {d_r} .
$
Consider the eigen-condition of $v_p.$ By (\ref{caengentric}), 
$
{d_p}{y_p} - \sum_{i \in {\mathcal{N}_p}} {{y_i}}  = {d_p}{y_p} - ({y_q} + {y_r})
 = ({d_p} + 1){y_p} - {y_r}.
$
To satisfy the eigen-condition, it requires
\begin{equation}
({d_p} + 1){y_p} - {y_r} = \lambda {y_p}\label{vpeicon}
\end{equation}
With $\lambda=d_r,$ the above equation means ${y_r} = ({d_p}+1-d_r){y_p}.$ Thus, 
for node $v_q,$ 
$
\sum_{i \in {\mathcal{N}_{qf}}} {{y_i}}  ={y_p} + {y_r}
=({d_p}+2-d_r) y_p.
$
By (\ref{caengentric}) and (\ref{doubelequ}),  
$
{d_q}{y_q} - \sum_{i \in {\mathcal{N}_q}} {{y_i}}  = {d_q} y_q+({d_p}+2-d_r) y_q
=(2d_q+2-d_r)y_q.
$
Hence, to satisfy the eigen-condition, it requires $2d_q+2-d_r=\lambda=d_r,$ i.e., $d_q+1=d_r.$ 
However, it will be shown $d_q>d_r.$ Since none of $v_k (k\neq p,q,r)$ is incident to both $v_p$ and $v_r$ and a) b) of Case I cannot arise simultaneously, then a node $v_h$ in $\mathcal{V}\setminus\{v_p,v_q,v_r\}$ which is incident to $v_r$ is also incident to $v_q.$ In addition, there is already at least one $v_k$ in $\mathcal{V}\setminus\{v_p,v_q,v_r\}$ which is incident to $v_q$ and $v_p.$ Hence $d_q>d_r$ and 
accordingly $\bar y$ cannot be an eigenvector of $\mathcal{L}.$ This contradicts with the assumption. 

For topology VI, only the proof different from that of topology V is given. As topology V,  it can be assumed without loss of generality that $v_p,v_q$ share at least one common node in $\mathcal{V}\setminus\{v_p,v_q,v_r\}.$ Consider the 
first situation of $\{v_p,v_r\},$ i.e., there is a $v_k (k\neq p,q,r)$ incident to 
both $v_p$ and $v_r.$ In this situation, (\ref{doubelequ}) and (\ref{doubelano}) still hold for $k=p,q,r.$ Then $y_p=-y_q=-y_r.$ 
For $k=p,$ (\ref{dpplus2}) still holds. For $k=q,$ 
$
\sum_{i \in {\mathcal{N}_{qf}}} {{y_i}}  = {y_p} =-y_q.
$
Thus 
\begin{equation}
{d_q}{y_q} - \sum\limits_{i \in {\mathcal{N}_q}} {{y_i}}  = ({d_q}+1) y_q.\label{dqeiconvi}
\end{equation}
Similarly, for $k=r,$ 
\begin{equation}
{d_r}{y_r} - \sum\limits_{i \in {\mathcal{N}_r}} {{y_i}}  = ({d_r}+1) y_r.\label{dreiconvi}
\end{equation}
The remaining discussion is the same as topology V.  
Next consider the second situation of $\{v_p,v_r\}.$ In this case, (\ref{doubelequ}) still holds. It can be seen that for $v_r,$
$
{d_r}{y_r} - \sum_{i \in {\mathcal{N}_r}} {{y_i}}  = {d_r}{y_r} - y_p.
$
The eigen-condition requires ${d_r}{y_r} - y_p=\lambda y_r,$ i.e., $y_p=(d_r-\lambda)y_r.$ For $v_p,$ it still requires equation (\ref{vpeicon}). So 
$
{y_r} = ({d_p} + 1 - \lambda ){y_p}
 = ({d_p} + 1 - \lambda )(d_r - \lambda ){y_r}. 
$
Since $y_r\neq 0$ 
\begin{equation}
({d_p} + 1 - \lambda )(d_r - \lambda ) = 1.\label{vpeicondivi}
\end{equation}
For $v_q,$ since (\ref{doubelequ}) still holds, 
$
{d_q}{y_q} - \sum_{i \in {\mathcal{N}_q}} {{y_i}}  = {d_q} y_q-y_p
=({d_q}+1) y_q. 
$
Thus, to satisfy the eigen-condition, it requires $\lambda={d_q}+1.$ By (\ref{vpeicondivi}), 
$({d_p} - {d_q})({d_r} - {d_q} - 1) = 1,$ which cannot be satisfied because $d_q>d_r$(as topology V) and $d_p,d_q$ are all integers. Accordingly, $\bar y$ cannot be an eigenvector of $\mathcal{L}.$ This contradicts with the assumption.

%For topologies VII and VIII, the second situation of $\{v_p,v_r\}$ does not occur. Otherwise, $v_r$ will be isolated from all the other nodes.  

For topology VII and the 
first situation of $\{v_p,v_r\},$  there does not exist node $v_{h}$ in $\mathcal{V}\setminus\{v_p,v_q,v_r\}$ which is incident to both $v_q$ and $v_r$ because (\ref{doubelequ})(\ref{doubelano})(\ref{doubthird}) cannot be satisfied simultaneously. Hence $d_p>d_r$ and  $d_p>d_q.$ Note that $\sum_{i \in {\mathcal{N}_{pf}}} {{y_i}}=y_q=-y_p.$ By (\ref{caengentric}), 
$
{d_p}{y_p} - \sum_{i \in {\mathcal{N}_p}} {{y_i}}  = (d_p+1) y_p.
$
Similarly, for $k=q,$ 
$
{d_q}{y_q} - \sum_{i \in {\mathcal{N}_q}} {{y_i}}  = ({d_q}+1) y_q.
$
Since $d_p+1>d_q+1,$ the eigen-condition of $v_p,v_q$ cannot be met for the same eigenvalue. For the 
second situation of $\{v_p,v_r\},$ $d_q>d_r.$ Since 
$
{d_q}{y_q} - \sum_{i \in {\mathcal{N}_q}} {{y_i}}  = ({d_q}+1) y_q;
$
$
{d_r}{y_r} - \sum_{i \in {\mathcal{N}_r}} {{y_i}}  =d_r y_r 
$
and $d_q+1>d_r,$  the eigen-condition of $v_q,v_r$ cannot be met for the same eigenvalue as well. This contradicts the assumption that $\bar y$ is an eigenvector.   

For topology VIII, $\sum_{i \in {\mathcal{N}_{kl}}} {{y_i}}=\sum_{i \in {\mathcal{N}_{kf}}} {{y_i}}=0 (k=p,q,r).$ By (\ref{caengentric})
\begin{equation}
{d_k}{y_k} - \sum\limits_{i \in {\mathcal{N}_k}} {{y_i}}  = {d_k}  y_k\label{eigenconviii} 
\end{equation}
Since each $v_k (k=p,q,r)$ has no neighbor nodes in $\{v_p,v_q,v_r\}$ and $\mathcal{G}$ is connected, it has at least one neighbor node in $\mathcal{V}\setminus\{v_p,v_q,v_r\};$ or else, $v_k$ will be an isolated node. With $v_p,v_q$ sharing a common neighbor node in $\mathcal{V}\setminus\{v_p,v_q,v_r\},$ the previous arguments show that $v_q,v_r$ do not share a common neighbor node in $\mathcal{V}\setminus\{v_p,v_q,v_r\}$ if the first situation of $v_p,v_r$ arises. In this circumstance, $d_p>d_q$ and $d_p>d_r.$
By (\ref{eigenconviii}), the eigen-condition requires $d_p=d_q=d_r,$ which cannot be met since $d_p>d_q.$ 
If the second situation of $v_p,v_r$ arises, the connectedness of $\mathcal{G}$ means there exist at least one $v_{k}$ in $\mathcal{V}\setminus\{v_p,v_q,v_r\}$ which is incident to both $v_q$ and $v_r.$ Since this $v_k$ cannot be incident to $v_p,v_q$ simultaneously, $d_q>d_p$ and $d_q>d_r.$ By (\ref{eigenconviii}), the eigen-condition cannot be met simultaneously for $v_p,v_q,v_r.$ This contradicts the assumption that $\bar y$ is an eigenvector. 
Above all, if $\bar y$ is an eigenvector of $\mathcal{L},$ then the topology of $v_p ,v_q, v_r$ accords with one of I to IV, i.e.,  they constitute a set of TCD nodes.

(\emph{Sufficiency of Lemma \ref{trilem}}) 
Firstly, suppose $v_p, v_q, v_r$ are TCD nodes with topology I. The corresponding topology structure means 
$d_p=d_q=d_r.$ For $k\neq p, q, r;$ the special form of $\bar y$ yields $\sum_{i \in {\mathcal{N}_{kl}}} {{y_i}}  = 0$ and $y_k=0.$ Then (\ref{triplesimcond}) holds. Since the topology structure of $v_p, v_q, v_r$ accords with type I, for any $k\neq p,q,r,$ either $p,q,r\in\mathcal{N}_{kf}$ or $p,q,r\notin\mathcal{N}_{kf}.$
If $p,q,r\in\mathcal{N}_{kf},$ then $\sum_{i \in {\mathcal{N}_{kf}}} {{y_i}}  = {y_p} + {y_q}+y_r.$ Since  $y_p+y_q+y_r=0,$ by (\ref{triplesimcond}) 
\begin{equation}
{d_k}{y_k} - \sum\limits_{i \in {\mathcal{N}_k}} {{y_i}}  =0. \label{spcase1}
\end{equation}
If $p,q,r\notin\mathcal{N}_{kf},$ (\ref{spcase1}) still holds. Since $y_k=0 (k\neq p, q, r),$ $\lambda y_k=0.$ Then, for any $k\neq p, q, r$ and any number $\lambda,$ the eigen-condition (\ref{eigcond}) holds. 
For $k= p, q, r,$ it follows from $\sum_{i \in {\mathcal{N}_{kl}}} {{y_i}}  = 0$ that
\begin{equation}
{d_k}{y_k} - \sum\limits_{i \in {\mathcal{N}_k}} {{y_i}}  = {d_k} y_k - \sum\limits_{i \in {\mathcal{N}_{kf}}} {{y_i}}.\label{ca2eicon}
\end{equation}
Since $\mathcal{N}_{kf}$ contains the other two indices of $p, q, r,$ for any given $k\in\{p, q, r\},$ say $k=p,$ it follows  
$
\sum_{i \in {\mathcal{N}_{kf}}} {{y_i}}  = {y_q} + {y_r}.
$
By $y_p+y_q+y_r=0$ and (\ref{ca2eicon}),
$
{d_k}{y_k} - \sum\limits_{i \in {\mathcal{N}_k}} {{y_i}}  = (d_{k}+1)y_k.
$
Thus, for any $k,$ the eigen-condition (\ref{eigcond}) is met for $\lambda=d_p+1.$ So the result holds for topology I. 

Secondly, if $v_p, v_q, v_r$ are TCD nodes with topology II, the associated topology structure implies 
$\mathcal{N}_{pl}=\mathcal{N}_{ql}=\mathcal{N}_{rl}$ and $\mathcal{N}_{pf}\setminus\{p,q,r\} =\mathcal{N}_{qf}\setminus\{p,q,r\}=\mathcal{N}_{rf}\setminus\{p,q,r\}.$ Moreover, since $q,r\in\mathcal{N}_{pf},$ $p\in\mathcal{N}_{qf},$ $p\in\mathcal{N}_{rf}$ and $\mathcal{N}_k=\mathcal{N}_{kl}+\mathcal{N}_{kf},$ it follows that 
$d_p=d_q+1=d_r+1.$ For $k\neq p, q, r,$ the same arguments as topology I yield that the eigen-condition is met for any number $\lambda.$ For $k=p,$ since $\sum_{i \in {\mathcal{N}_{kl}}} {{y_i}}  = 0,$ $q,r\in\mathcal{N}_{pf}$ and $y_p+y_q+y_r=0,$
$
\sum_{i \in {\mathcal{N}_{pf}}} {{y_i}}  = {y_q}+y_r
=-y_p.
$
By (\ref{ca2eicon}) 
$
{d_p}{y_p} - \sum_{i \in {\mathcal{N}_p}} {{y_i}}  
= (d_p+1) y_p. \label{dppluforth}
$
For $k=q,$ since $p\in\mathcal{N}_{qf},$
$
{d_q}{y_q} - \sum_{i \in {\mathcal{N}_q}} {{y_i}}  = {d_q} y_q-y_p.
$
From $y_p+y_q+y_r=0$ and $y_q=y_r,$ one has 
$
{d_q}{y_q} - \sum_{i \in {\mathcal{N}_q}} {{y_i}}  = ({d_q}+2) y_q. \label{qforeign}
$
For $k=r,$ the same arguments as $k=q$ gives  
$
{d_r}{y_r} - \sum_{i \in {\mathcal{N}_r}} {{y_i}}  = ({d_r}+2) y_r. \label{rfortheign}
$
The previous arguments show that $\bar y$ is an eigenvector of $\mathcal{L}$ with $d_p+1$ the corresponding eigenvalue.  

Thirdly, if $v_p, v_q, v_r$ are TCD nodes with topology III, $d_p=d_q=d_r+1,$ which can be verified in the same way as the beginning part of proof of topology II. For $k\neq p,q,r,$ the same proof as that of topology I yields that the eigen-condition holds for any number $\lambda$ if $y_p+y_q+y_r=0.$ For $k=p,q,r,$ (\ref{caengentric}) holds. For $k=p,$ $\sum_{i \in {\mathcal{N}_{pf}}} {{y_i}}  = {y_q}$ and for $k=q,$ $\sum_{i \in {\mathcal{N}_{qf}}} {{y_i}}  = {y_p}.$ By (\ref{caengentric}) and $y_p=y_q,$ it follows 
$
{d_p}{y_p} - \sum\limits_{i \in {\mathcal{N}_p}} {{y_i}} =(d_p-1)y_p. 
$
Similarly, for $k=q,$ 
$
{d_q}{y_q} - \sum_{i \in {\mathcal{N}_q}} {{y_i}}  =(d_q-1)y_q. 
$
For $k=r,$ since $\sum_{i \in {\mathcal{N}_{rl}}} {{y_i}}  = \sum_{i \in {\mathcal{N}_{rf}}} {{y_i}}  = 0,$ 
$\sum_{i \in {\mathcal{N}_r}} {{y_i}}  = 0,$ it can be seen that 
$
{d_r}{y_r} - \sum_{i \in {\mathcal{N}_r}} {{y_i}}  = {d_r}{y_r}.
$
Since $d_p=d_q=d_r+1,$ the above arguments show that the eigen-condition holds for each $k$ and the 
corresponding eigenvalue is $\lambda=d_r.$

Finally, if $v_p, v_q, v_r$ are with topology IV, $d_p=d_q=d_r.$ In addition, for $k\neq p,q,r,$ the eigen-condition still holds for any number $\lambda$ if $y_p+y_q+y_r=0;$ and for $k=p,q,r,$ $\sum_{i \in {\mathcal{N}_{kl}}} {{y_i}}  = \sum_{i \in {\mathcal{N}_{kf}}} {{y_i}}  = 0.$ Thus $\sum_{i \in {\mathcal{N}_{k}}} {{y_i}}  =0 (k=p,q,r),$ and accordingly   
$
{d_k}{y_k} - \sum_{i \in {\mathcal{N}_k}} {{y_i}}  = {d_k}{y_k}.
$
Thus the eigen-condition is met for each $k$ if the eigenvalue $\lambda=d_p.$
Therefore, $\bar y$ is an eigenvector of $\mathcal{L}$ if $v_p,v_q,v_r$ are TCD nodes with one of topologies I to IV.
\hfill
\end{proof}

\begin{theorem}\label{tripDcd}
There exist a group of leaders selected from $\Gamma _{p,q,r}$ such that the multi-agent system with single-integrator dynamics (\ref{singmul}) is controllable if and only if the following two conditions are met simultaneously:
\begin{itemize}
\item the follower node set does not contain TCD nodes $v_p,v_q,v_r,$ where $p,q,r\in\{1,\ldots,n+l\},$ ${\Gamma _{p,q,r}}\mathop  = \limits^\Delta  \{ 1, \ldots,$ $n + l\}\setminus \{ p,q,r\}.$

\item any two of $v_p, v_q, v_r$ are not DCD nodes. 
\end{itemize}
\end{theorem}
\begin{proof}
(\emph{Necessity}) Suppose by contradiction that two of $v_p, v_q, v_r$ are DCD nodes, then necessity can be proved in the same vein as that of Theorem \ref{doubDcd}. In case $v_p, v_q, v_r$ are TCD nodes, the proof can be carried out in the same way by using Lemma \ref{trilem}. 

(\emph{Sufficiency}) Suppose by contradiction that the system is uncontrollable with any leaders selected from ${\Gamma _{p,q,r}}.$ Then the same arguments as the sufficiency proof of Theorem \ref{doubDcd} show that $\bar y = [0, \ldots ,{y_p},$ $0, \ldots ,{y_q},0,$ $ \ldots ,{y_r},$ $0, \ldots ,0]^T$ is an eigenvector of $\mathcal{L}.$ Next, it is to verify $y_p, y_q, y_r\neq 0.$
Firstly, we show that two of $y_p, y_q, y_r$ cannot be zero. 
Suppose by contradiction that two of $y_p, y_q, y_r$ take zero, say $y_p=y_q=0.$ Then $y_r\neq 0,$ or else $\bar y$ is a zero vector. Since $\mathcal{G}$ is connected, $\lambda=0$ is a simple eigenvalue associated with the all one eigenvector  $\textbf{1}.$ Thus the eigenvalue $\beta$ associated with $\bar y$ is not zero. Since $\mathcal{G}$ is connected, there is a $k\neq r$ with $k\in\mathcal{N}_r,$ i.e., the corresponding $v_k$ is incident to $v_r.$ Otherwise, $v_r$ turns to be an isolated node. The special form of $\bar y$ then leads to $\sum\nolimits_{i \in {\mathcal{N}_{kl}}} {{y_i}}  = 0,\sum\nolimits_{i \in {\mathcal{N}_{kf}}} {{y_i}}  = {y_r}.$ From $y_k=0,$ one has 
$
{d_k}{y_k} - \sum_{i \in {\mathcal{N}_k}} {{y_i}}  = -y_r.
$
Since $y_k=0, y_r\neq 0,$ this equation means that the eigen-condition (\ref{eigcond}) of $v_k$
is not met. This contradicts with the condition that $\bar y$ is an eigenvector. So any two of $y_p, y_q, y_r$ cannot take the value of zero. 
Secondly, suppose there is one and only one of $y_p, y_q, y_r$ taking zero, say $y_p=0$ and $y_q\neq 0, y_r\neq 0.$ By Lemma \ref{doulem}, the corresponding $v_q, v_r$ constitute a pair of DCD nodes. This contradicts with the condition that any two of $v_p, v_q, v_r$ are not DCD nodes. 
Since $y_p, y_q, y_r\neq 0,$ Lemma \ref{trilem} shows that $v_p, v_q, v_r$ constitute a triple of TCD nodes. This also contradicts with the condition. \hfill
\end{proof}

\subsection{quadruple destructive nodes}

\subsubsection{A design method for QCD nodes}
Below ${s_1},{s_2},{t_1},{t_2}$ are used to represent the indices of the desired quadruple controllability destructive (QCD) nodes. Let $\eta$ be a vector with entries $\eta_p=\eta_q=0$ and 
\begin{equation}
\eta_{s_1}=\eta_{s_2}=-\eta_{t_1}=-\eta_{t_2}\neq 0\label{eielem}
\end{equation}
where $p, q, s_1,s_2,t_1,t_2$ are distinct and  all the other entries of $\eta$ are zero. The node set of $\mathcal{G}$ can be broken down into four parts: $\{v_p, v_q\}, \{v_{s_1},v_{s_2}\},$ $ \{v_{t_1},v_{t_2}\}$ and the others. In subsequent topology design procedure, $v_p, v_q$ are fixed in advance to assist in designing neighbor relationship of $\{v_{s_1},v_{s_2}\}$ and $ \{v_{t_1},v_{t_2}\}.$ The neighbor topology structure of $\{v_{s_1},v_{s_2}\}$ to $\{v_p, v_q\}$ and $\{v_{t_1},v_{t_2}\}$ is constructed below, where $v_{s_2}$ follows the same rule as $v_{s_1}.$ So the rule  is stated only for $v_{s_1}.$ 
A topology design procedure for QCD nodes is as follows:

\textbf{Case I}. $v_{s_1}$ has no neighbor relationship with $v_{s_2},$ and so has $v_{t_1}$  with $v_{t_2}.$ The design is divided into four steps:

\emph{Step 1} The construction of neighbor nodes of $v_{s_1}$ conforms to one of the following cases: 
\begin{enumerate}
\item[{i)}] $v_{s_1}$ is a neighbor of both $v_p$ and $v_q.$ In this case, $v_{s_1}$ is required to have neighbor relationship with only one of $v_{t_1}$ and $v_{t_2}.$ 

\item[{ii)}] $v_{s_1}$ has neighbor relationship to neither $v_p$ nor $v_q.$ In this case, $v_{s_1}$ is required to have neighbor relationship with both $v_{t_1}$ and $v_{t_2}.$ 
\end{enumerate}

\emph{Step 2} The design of the neighbor topology structure of $\{v_{t_1},v_{t_2}\}$ to $\{v_p, v_q\}$ and $\{v_{s_1},v_{s_2}\}$ is in the same vein as that of $\{v_{s_1},v_{s_2}\}$ to $\{v_p, v_q\}$ and $\{v_{t_1},v_{t_2}\}.$

\emph{Step 3} For $k=p,q,$ $\mathcal{N}_{kf}$ contains exactly one of $s_1, s_2$ and one of $t_1, t_2.$ 

\emph{Step 4} For $k\in\Omega \mathop  = \limits^\Delta \{1,\ldots,n+l\}\setminus\{p,q,s_1,s_2,t_1,t_2\},$ the design of neighbors of $v_{k}$ conforms to the following cases:
\begin{enumerate}
\item[\emph{a)}] $v_k$ is a neighbor of both $v_p$ and $v_q;$ 

\item[\emph{b)}] $v_k$ is a neighbor of all of $v_{s_1},v_{s_2},v_{t_1},v_{t_2};$  

\item[\emph{c)}] $v_k$ does not have neighbor relationship to any of $v_p,v_q,v_{s_1},v_{s_2},v_{t_1},v_{t_2};$

\item[\emph{d)}] $v_k$ has arbitrary neighbor relationship with any other nodes except $v_p,v_q,v_{s_1},v_{s_2},v_{t_1},v_{t_2}.$
\end{enumerate}
Any of \emph{a)}, \emph{b)}, \emph{c)}, \emph{d)} can be satisfied simultaneously. 

\textbf{Case II}. at least one of the following two cases occur: $v_{s_1}$ is a neighbor of $v_{s_2};$ or $v_{t_1}$ is a neighbor of $v_{t_2}.$
The remaining construction is the same as Case I.

\begin{remark}
The neighbor topology structure of $\{v_{s_1},v_{s_2}\}$ to $\{v_p,v_q\}$ is designed to be the same as that of $\{v_{t_1},v_{t_2}\}$ to $\{v_p,v_q\}.$  This kind of equivalence of neighbor topology between $\{v_{s_1},v_{s_2}\}$ and $\{v_{t_1},v_{t_2}\}$ makes leaders incapable to torn open them and therefore destroys controllability.
\end{remark}

%\begin{remark}
%The definitions of DCD and TCD nodes inspire the rule that for any node $v_k$ other than $v_{s_1},v_{s_2},$ $v_{t_1},v_{t_2};$ $\mathcal{N}_{kf}$ contains either $s_1,s_2,t_1,t_2$ or none of them. 
%This rule, however, is not always applicable to QCD nodes. Although $v_p, v_q$ in the above design %procedure do not comply with this rule, subsequent theorem exhibits that  $v_{s_1},v_{s_2},v_{t_1},v_{t_2},%$ with any of their topology structures designed via Steps 1-4, still constitute a set of QCD nodes. So the %topology structures of QCD nodes are even more complicated than those of DCD and TCD nodes.
%\end{remark}

\begin{theorem}\label{anotherconthm}
If system  (\ref{singmul}) is controllable, then the follower node set does not contain $v_{s_1},v_{s_2},v_{t_1},v_{t_2}$ with the topology structure of $v_{s_1},v_{s_2},v_{t_1},v_{t_2}$ agreeing with any of those designed via Steps 1-4,  where $s_1,s_2,t_1,t_2\in\{1,\ldots,n+l\}$ are distinct indices. 
\end{theorem}
\begin{proof}
The $\eta$ in (\ref{eielem}) is shown to be an eigenvector of $\mathcal{L}.$ The result will then follows from Proposition \ref{singPro}. %The proof is first presented for Case I.

%\begin{itemize}
For $k=s_1,s_2,$ if the neighbor nodes of $v_{k}$ to $\{v_p,v_q\}$ and $\{v_{t_1},v_{t_2}\}$ are  designed according to i) of Step 1, there are three neighbors of $v_{k}$ in $\{v_p,v_q,v_{t_1},v_{t_2}\}.$ In addition, denote by $\sigma$ the number of neighbor nodes of $v_{k}$ in $\mathcal{V}\setminus\{v_p,v_q,{v_{s_1}},{v_{s_2}},{v_{t_1}},{v_{t_2}}\}.$ Then the node degree of $v_{k}$ is $d_k=\sigma+3.$ Note that b) of Step 4 means that the value of $\sigma$ remains unchanged for each $v_{k}, k=s_1,s_2.$  
Since all the elements of $\eta$ are zero except $\eta_{s_1},\eta_{s_2},\eta_{t_1},\eta_{t_2};$ 
$\sum_{i \in {\mathcal{N}_{k}}} {{\eta _i}}=\eta_t,$ where $t=t_1~\mbox{or}~t_2$ depending on the specific situation of item i). Then $\eta_k=-\eta_t$ yields that 
\begin{align}
{d_k}{\eta _k} - \sum\limits_{i \in {\mathcal{N}_{k}}} {{\eta _i}}  
=&(d_k+1) {\eta_k}\nonumber\\
=& (\sigma+4) \eta_k, \quad k=s_1,s_2. \label{case2equ}
\end{align}
If the neighbors of $v_{s_k}$ are designed via ii) of Step 1, $d_k=\sigma+2.$ In this case, 
$ \sum_{i \in {\mathcal{N}_{k}}} {{\eta _i}}=\eta_{t_1}+\eta_{t_2}.$ By (\ref{eielem}),
$
{d_k}{\eta _k}-\sum_{i \in {\mathcal{N}_{k}}} {{\eta _i}}  = {d_k}{\eta _k} +2 \eta _k 
=(\sigma+4) {\eta _k}, \,k=s_1,s_2.
$
For $k=t_1,t_2,$ the neighbor nodes of $\{v_{t_1},v_{t_2}\}$ to $\{v_p, v_q\}$ and $\{v_{s_1},v_{s_2}\}$ is designed in the same way as that of $\{v_{s_1},v_{s_2}\}$ to $\{v_p, v_q\}$ and $\{v_{t_1},v_{t_2}\}.$ In addition, Step 4 means that the aforementioned $\sigma$ is also the number of neighbors of $v_{k}$ in $\mathcal{V}\setminus\{v_p,v_q,{v_{s_1}},{v_{s_2}},{v_{t_1}},{v_{t_2}}\}.$ Then the proof can be carried out in the same manner as the case of $k=s_1,s_2.$ Accordingly
\begin{equation}
{d_k}{\eta _k}-\sum\limits_{i \in {\mathcal{N}_{k}}} {{\eta _i}} =(\sigma+4) {\eta _k},\quad k=t_1,t_2.\label{commoneigen}
\end{equation}
For $k=p,q,$ it follows from Step 3 that 
\begin{equation}
\sum\limits_{i \in {\mathcal{N}_k}} {{\eta _i}}  = \sum\limits_{i \in {\mathcal{N}_{kf}}} {{\eta _i}}=\eta_s+\eta_t,\quad k=p,q,\label{kweip}
\end{equation}
where $s=s_1~\mbox{or}~s_2;$ $t=t_1~\mbox{or}~t_2$ depending on the specific situation of Step 3. By (\ref{eielem}), $\eta_s=-\eta_t.$ Then (\ref{kweip}) yields $\sum_{i \in {\mathcal{N}_k}} {{\eta _i}}  =0.$ By $\eta_k=0,$ (\ref{commoneigen}) also holds for $k=p,q.$  

For $k\in\Omega,$ Step 4 means 
$\sum_{i \in {\mathcal{N}_k}} {{\eta _i}}  = {\eta _{{s_1}}} + {\eta _{{s_2}}} + {\eta _{{t_1}}} + {\eta _{{t_2}}} = 0$ if b) is involved; 
and
$\sum_{i \in {\mathcal{N}_k}} {{\eta _i}}  =0$ if b) is not involved. 
This together with $\eta_k=0$ also leads to (\ref{commoneigen}) for $k\in\Omega.$
The above arguments show that $\eta$ is an eigenvector of $\mathcal{L}.$ 

For Case II, the above proof for Case I needs a bit of alteration. Below the discussion focuses on the situation that $v_{s_1}$ is a neighbor of $v_{s_2}.$ The result can be shown in the same way when $v_{t_1}$ is a neighbor of $v_{t_2}.$     
For $k=s_1,s_2,$ the node degree of $v_k$ is changed to be $\sigma+4$ and $ \sum_{i \in {\mathcal{N}_{k}}} {{\eta _i}}=0$ since there is an additional edge between $v_{s_1}$ and $v_{s_2}.$ Thus
(\ref{commoneigen}) holds for $k=s_1,s_2.$ 
If the neighbors of $v_{s_k}$ are designed according to ii) of Step 1, $d_k=\sigma+3.$ In this case, 
$ \sum_{i \in {\mathcal{N}_{k}}} {{\eta _i}}=\eta_t,$ where $t=t_1~\mbox{or}~t_2$ depending on the specific construction.   By (\ref{eielem}), (\ref{commoneigen}) still holds.
For $k=t_1,t_2,$ the proof is in the same manner as $k=s_1,s_2.$
The remaining proof is the same as Case I. This completes the proof. 
\end{proof}

\begin{example}
The example is to illustrate Theorem \ref{anotherconthm}. 
\begin{figure}[H]
\begin{center}
\subfigure[]{\includegraphics[width=2.19cm]{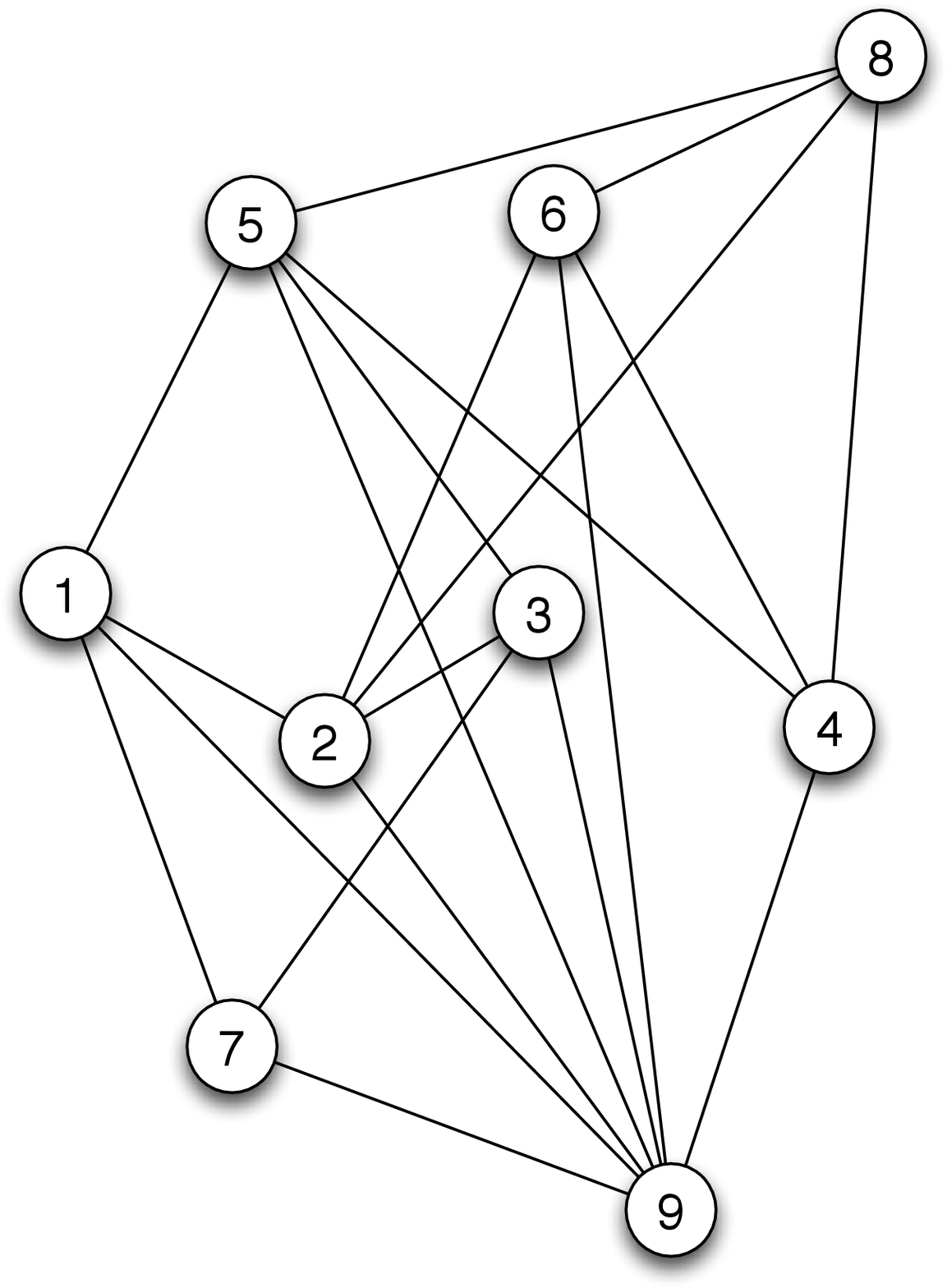}}
\subfigure[]{\includegraphics[width=2.19cm]{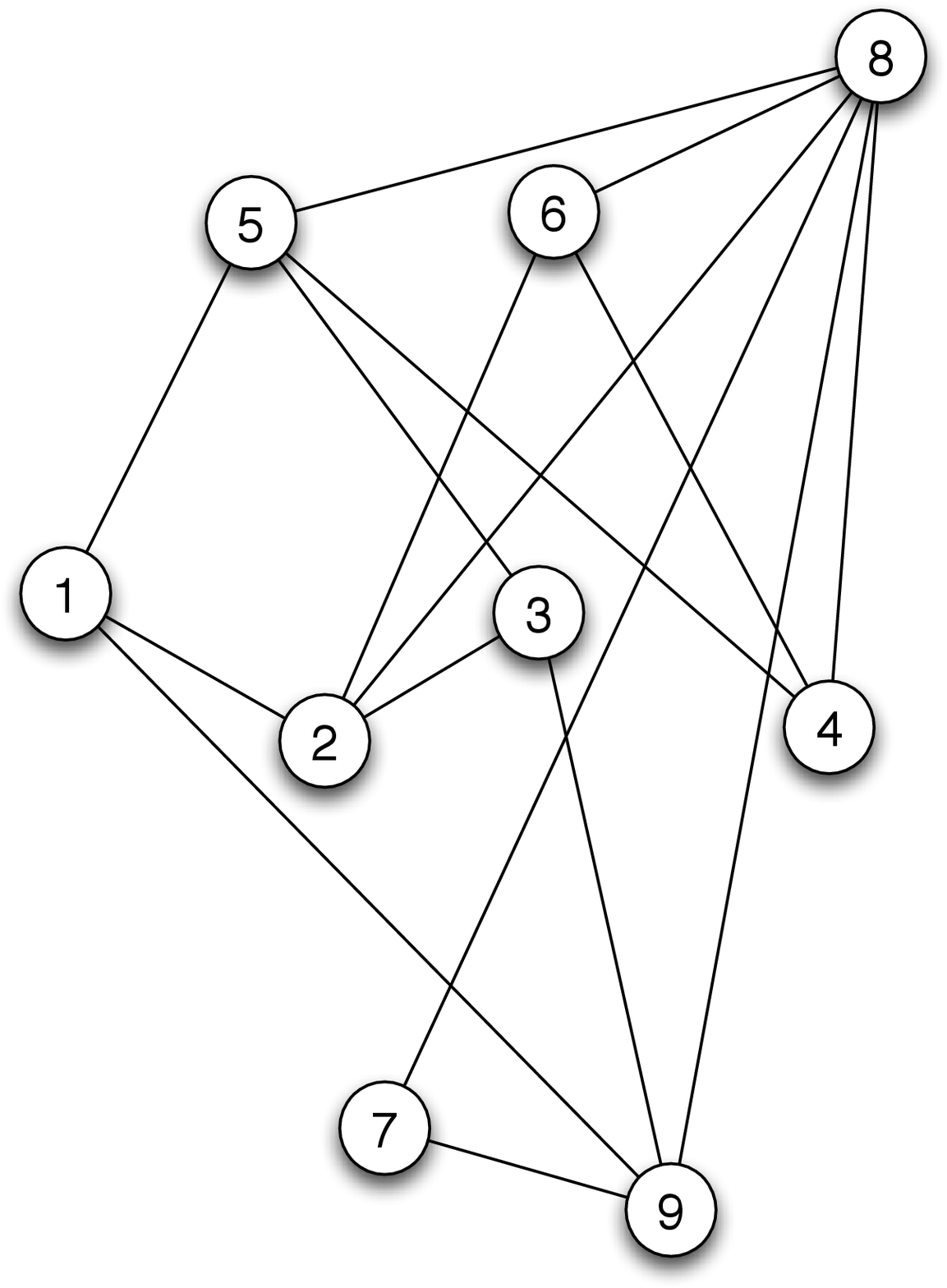}}
\subfigure[]{\includegraphics[width=2.19cm]{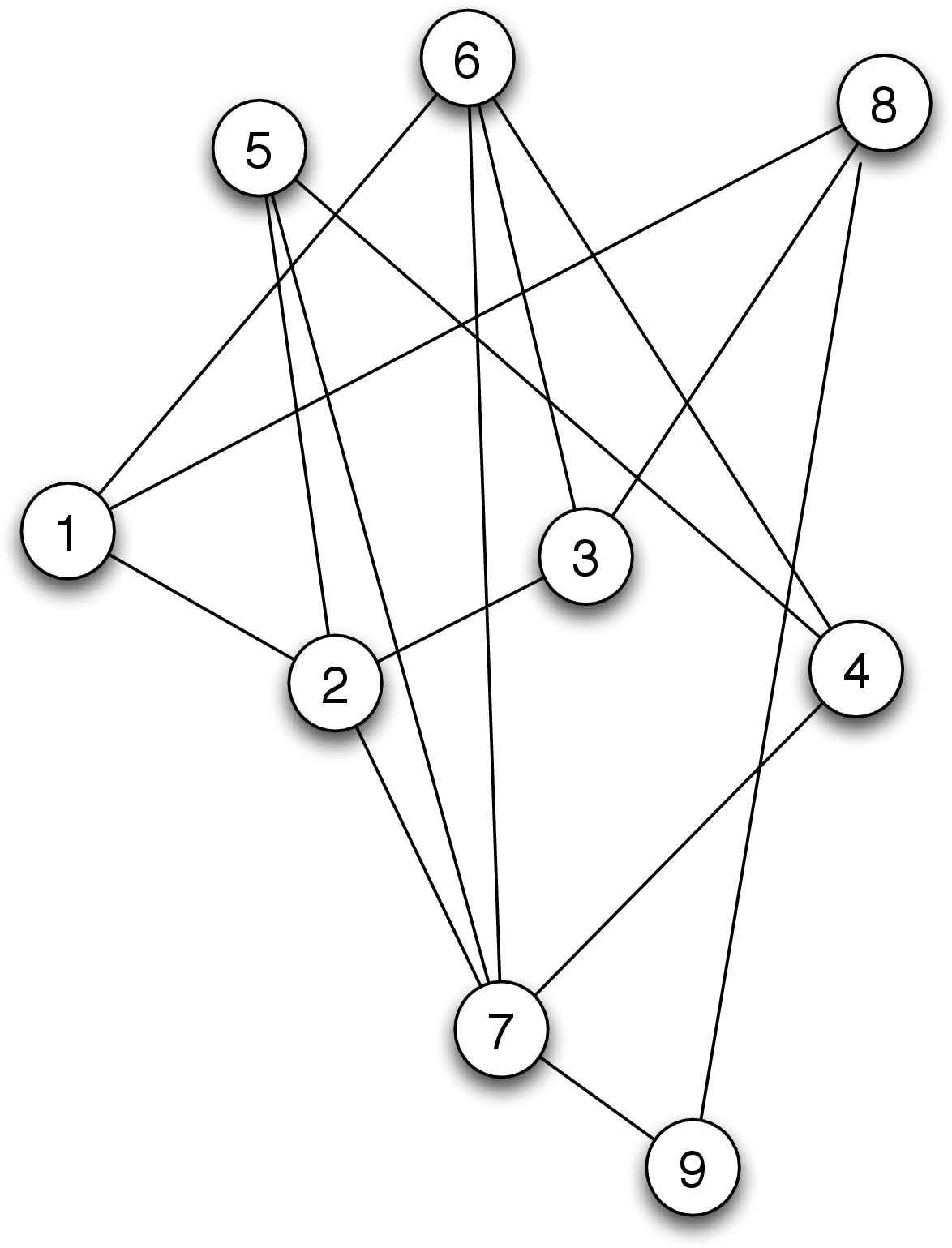}}
\subfigure[]{\includegraphics[width=2.19cm]{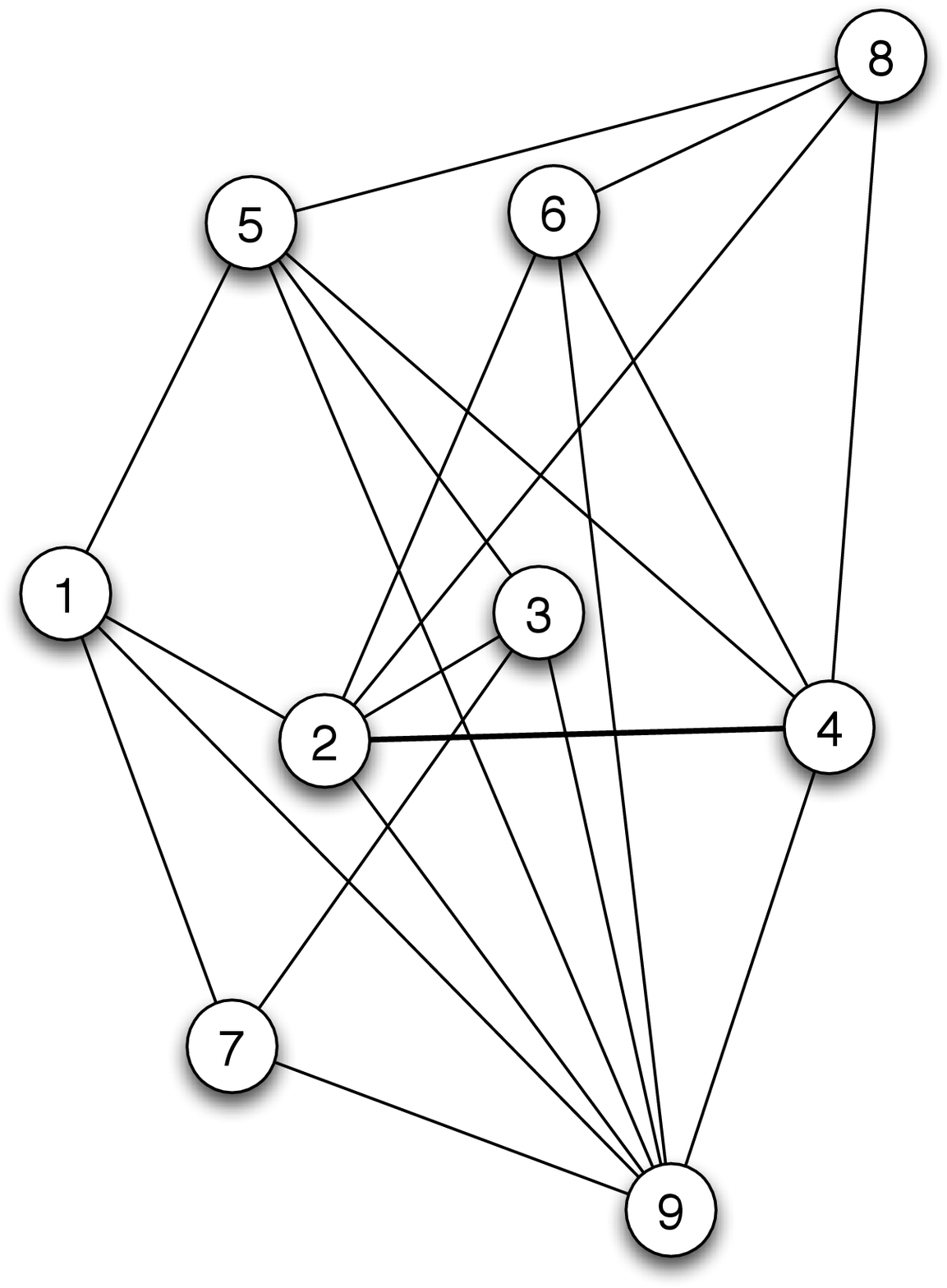}}
\subfigure[]{\includegraphics[width=2.19cm]{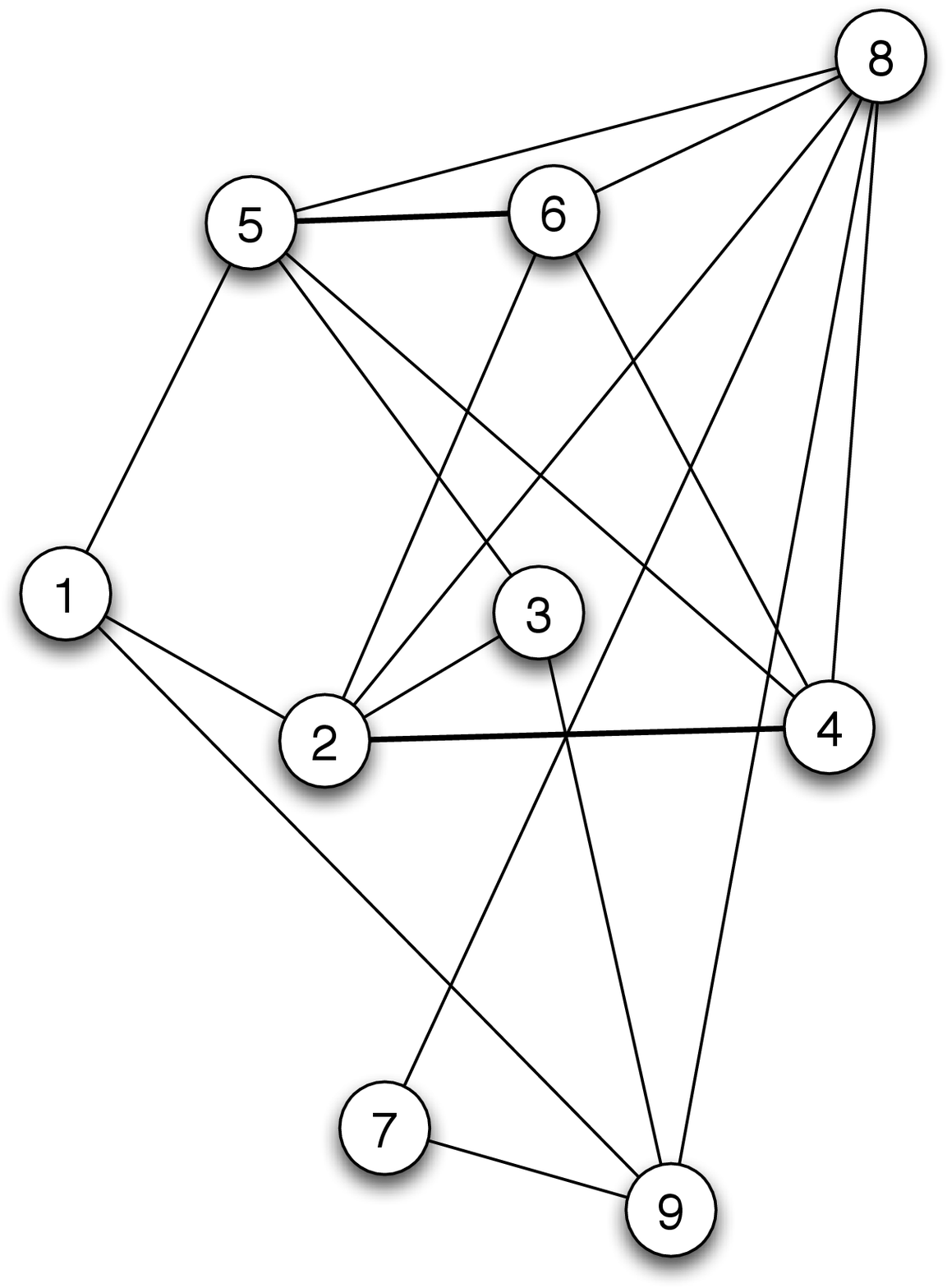}}
\subfigure[]{\includegraphics[width=2.19cm]{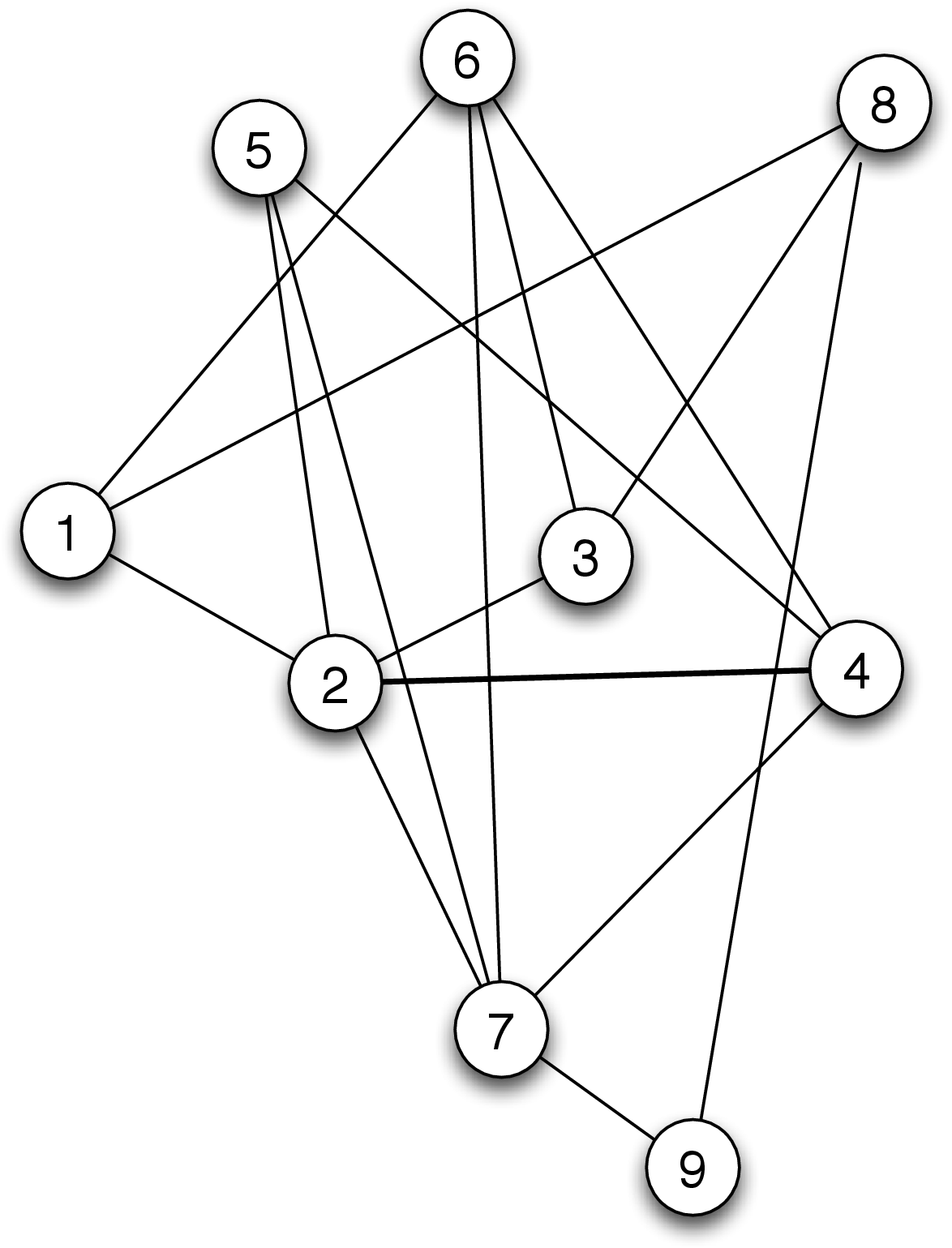}}
\caption{(a)(b)(c) and (d)(e)(f) are designed according to Case I and II, respectively, with QCD nodes $v_2,v_4,v_5,v_6.$}
\label{learepla2}
\end{center}
\end{figure}
In these graphs, $p=1,q=3;$ $s_1=2,s_2=4,t_1=5,t_2=6.$ In (a), $v_{s_1}=v_2$ is a neighbor of both $v_p=v_1$ and $v_q=v_3;$ and it is incident to $v_6,$ i.e., only one of $v_{t_1}=v_5$ and $v_{t_2}=v_6.$  This corresponds to case i) of Step 1. Similarly, $v_{s_2}$ corresponds to ii) of Step 1. These arguments exhibit the neighbor topology structure of $\{v_{s_1}, v_{s_2}\}$ to $\{v_p,v_q\}$  and $\{v_{t_1}, v_{t_2}\}.$ That of $\{v_{t_1}, v_{t_2}\}$ to $\{v_p,v_q\}$  and $\{v_{s_1}, v_{s_2}\}$ can be illustrated in the same manner. 
For graph (a), $\sigma=2$ since the number of neighbors of each $v_{s_k} (k=1,2)$ in $\mathcal{V}\setminus\{v_p,v_q,{v_{s_1}},{v_{s_2}},{v_{t_1}},{v_{t_2}}\}$ is 2.  The neighbor topology structures of $v_7,v_8,v_9$ are designed in accordance with Step 4. For $k=p,q,$ exactly one of $v_{s_1}=v_2,v_{s_2}=v_4$ ($v_2$ here) and one of $v_{t_1}=v_5,v_{t_2}=v_6$ ($v_5$ here) are included in the neighbor set of $v_k.$ This is consistent with Step 3. It can be verified that $\eta=[0,-0.5,0,-0.5,0.5,0.5,0,0,0]^T$ is an eigenvector of $\mathcal{L}$ of graph (a) associated with  eigenvalue $\sigma+4=6.$ For graph (b),  $\sigma=1$ and $\eta=[0,0.5,0,0.5,-0.5,-0.5,0,0,0]^T$ is an eigenvector of $\mathcal{L}$ of graph (b) associated with eigenvalue $\sigma+4=5.$ For graph (c), $\sigma=1$ as well, and $\eta=[0,-0.5,0,-0.5,0.5,0.5,0,0,0]^T$ is an eigenvector of $\mathcal{L}$ associated with eigenvalue 5.
Hence for graphs (a)(b)(c), the system is not controllable whenever leaders are selected from $\mathcal{V}\setminus \{{v_{s_1}},{v_{s_2}},{v_{t_1}},{v_{t_2}}\}.$ For graphs (d)(e)(f), there is a similar explanation. 
\end{example}

\subsubsection{QCD nodes of graphs of five vertices}

Consider an eigenvector $\bar y$ of $\mathcal{L}$ with $\bar y = [0, \ldots ,{y_{s_1}},\ldots,$ ${y_{s_2}}, \ldots ,{y_{t_1}},$ $\ldots,y_{t_2}, \ldots ,0]^T,$ $y_{s_1}, y_{s_2}, y_{t_1}, y_{t_2}\neq 0$ and all the other elements being zero. $\bar y$ does not necessarily meet  (\ref{eielem}) and each entry of it ought to satisfy the eigen-condition. For each $k\neq s_1, s_2, t_1, t_2;$ $\mathcal{N}_{kf}$ has five cases:  \\
a) $s_1, s_2, t_1, t_2\in\mathcal{N}_{kf};$\\
b) any three and only three of $s_1,s_2,t_1,t_2$ belong to $\mathcal{N}_{kf};$ \\
c) any two and only two of $s_1,s_2,t_1,t_2$ belong to $\mathcal{N}_{kf};$ \\
d) any one and only one of $s_1,s_2,t_1,t_2$ belongs to $\mathcal{N}_{kf};$ \\
e) none of $s_1,s_2,t_1,t_2$ belongs to $\mathcal{N}_{kf}.$ 

\begin{figure}[H]
\begin{center}
\subfigure[]{\includegraphics[width=1.84cm]{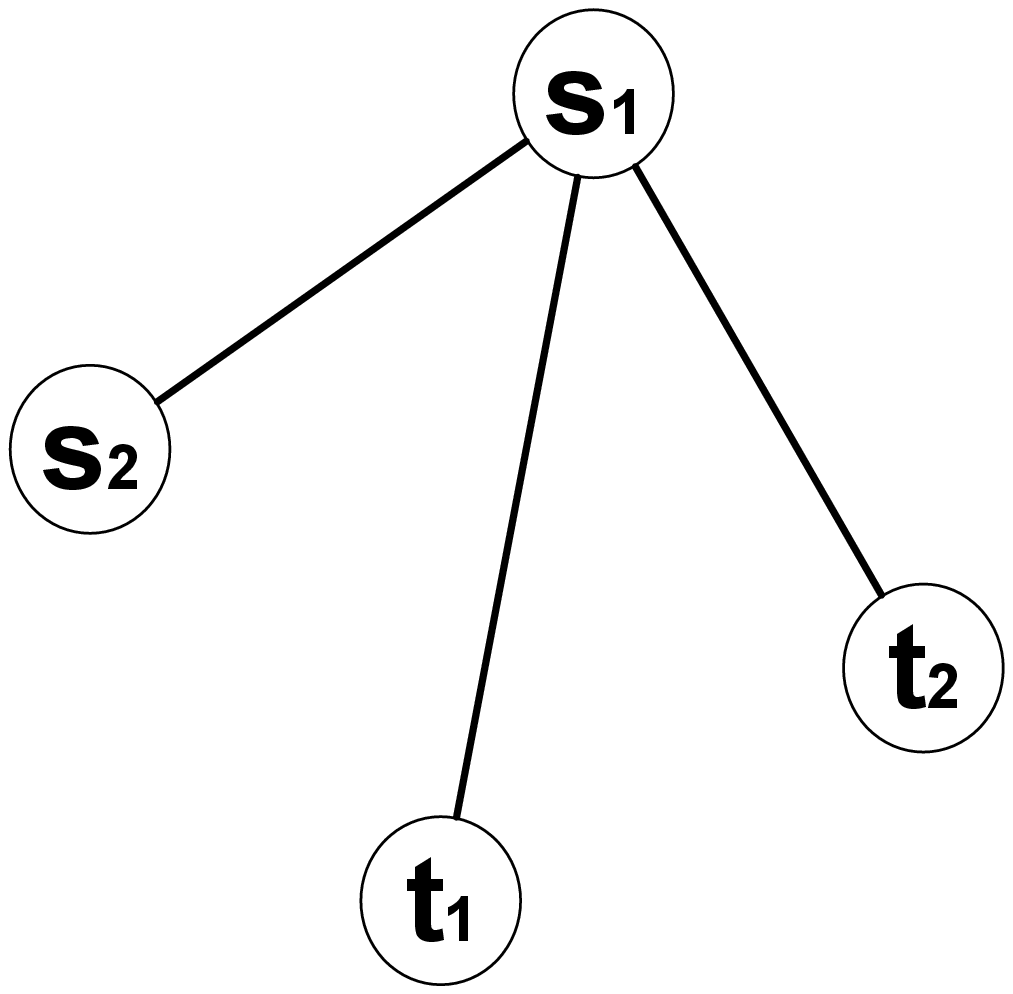}}
\subfigure[]{\includegraphics[width=1.84cm]{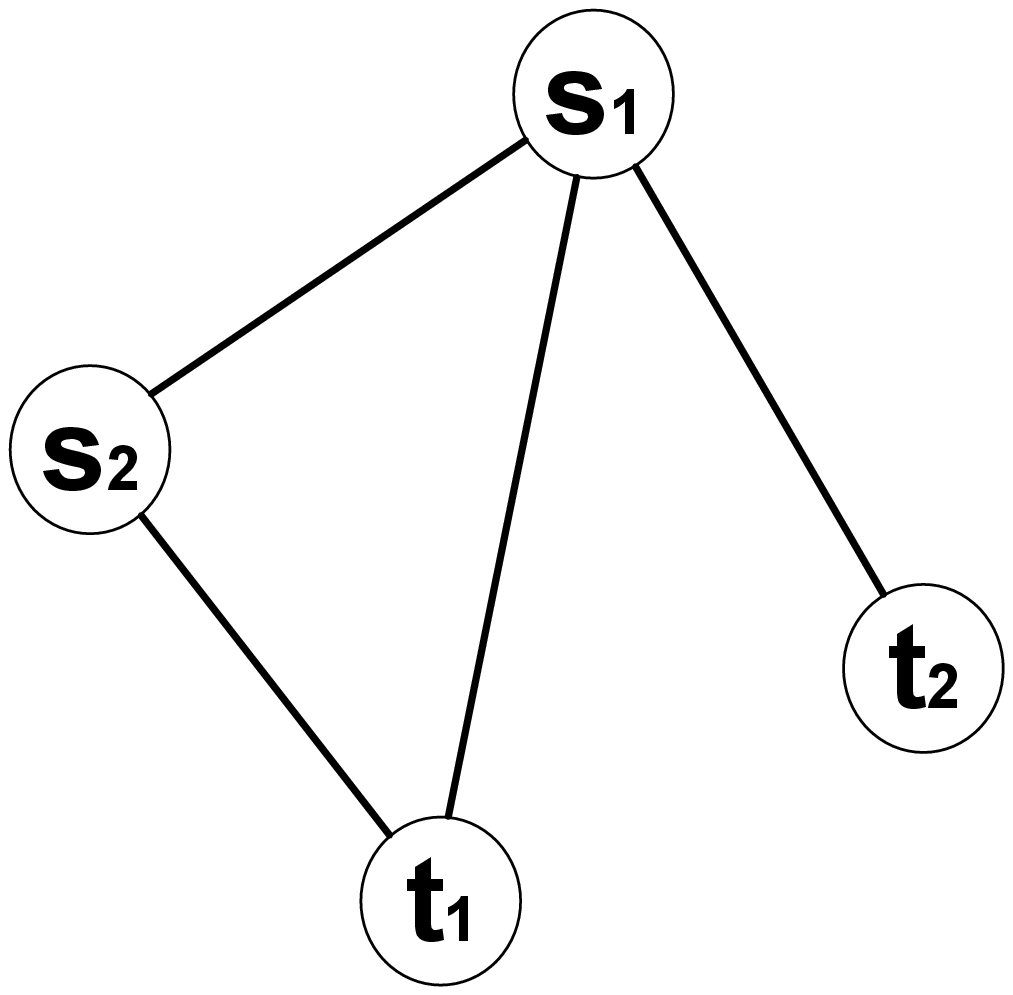}}
\subfigure[]{\includegraphics[width=1.84cm]{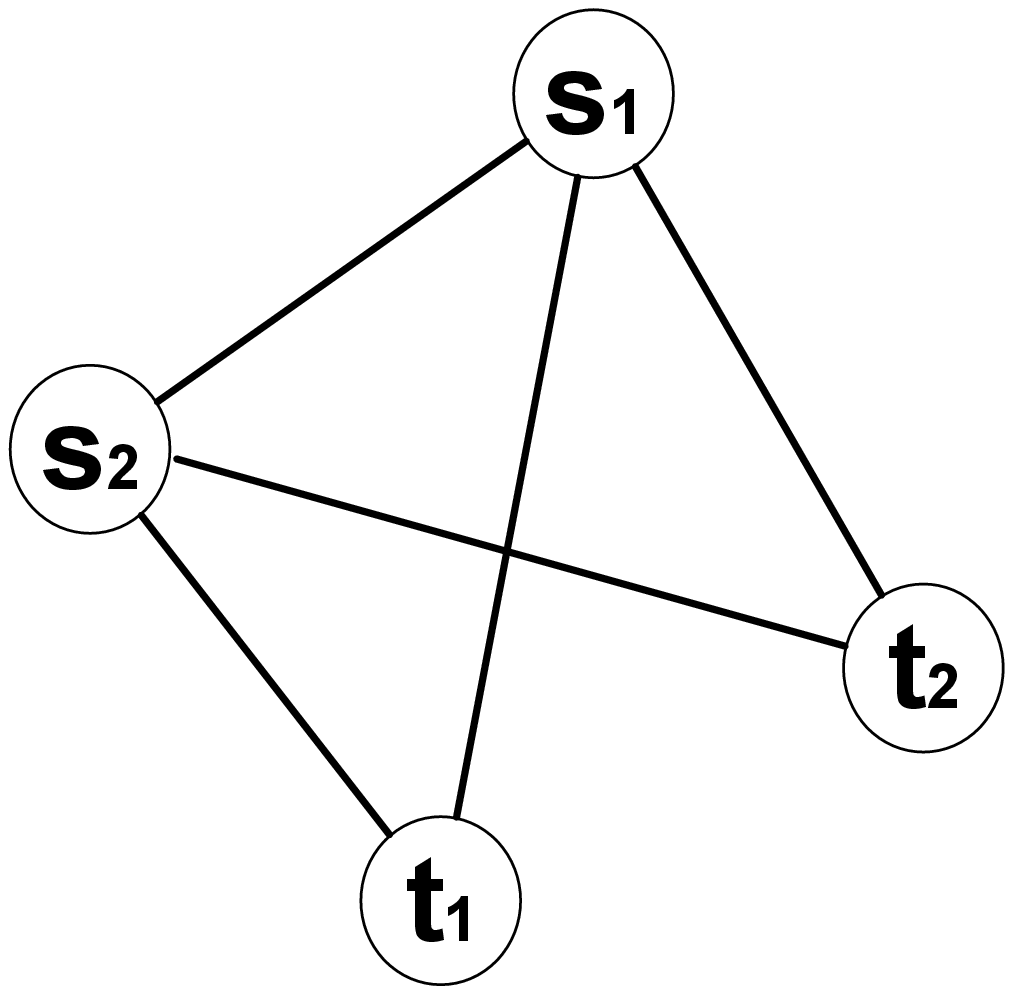}}
\subfigure[]{\includegraphics[width=1.84cm]{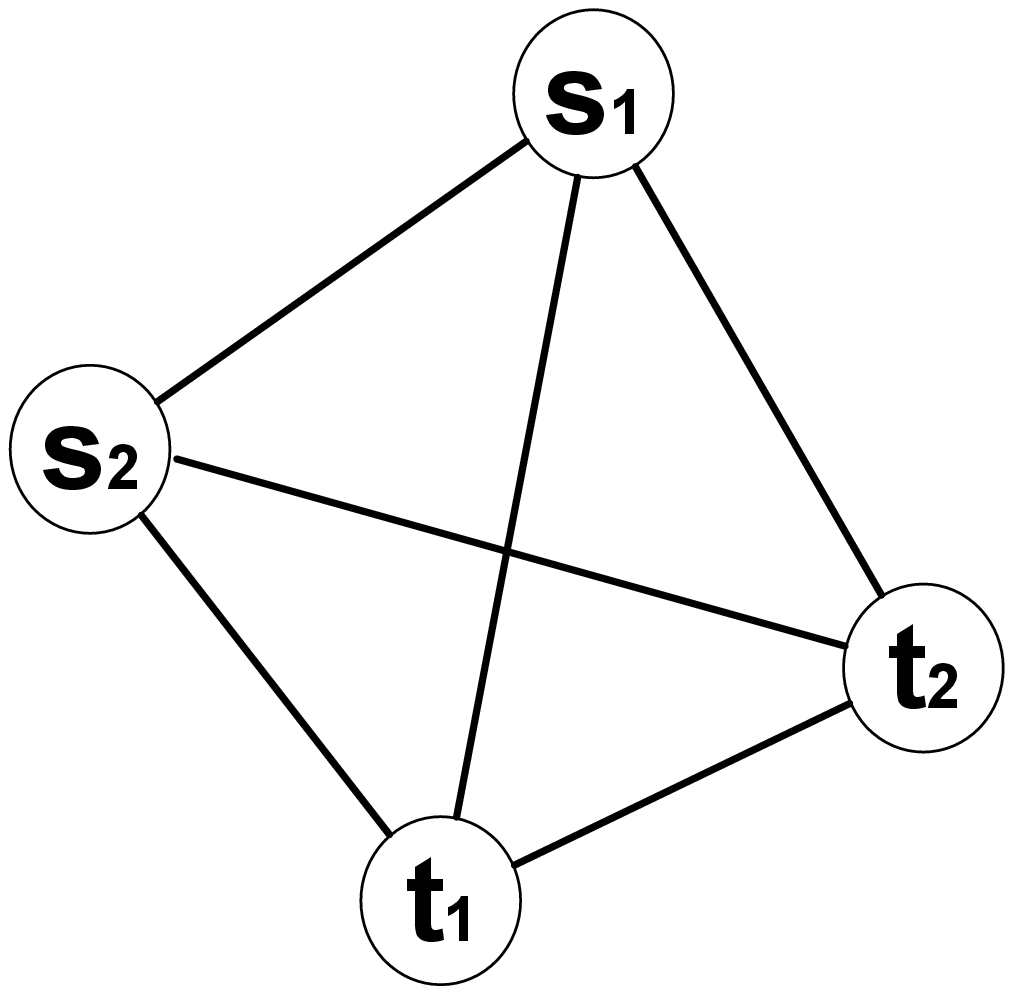}}
\subfigure[]{\includegraphics[width=1.84cm]{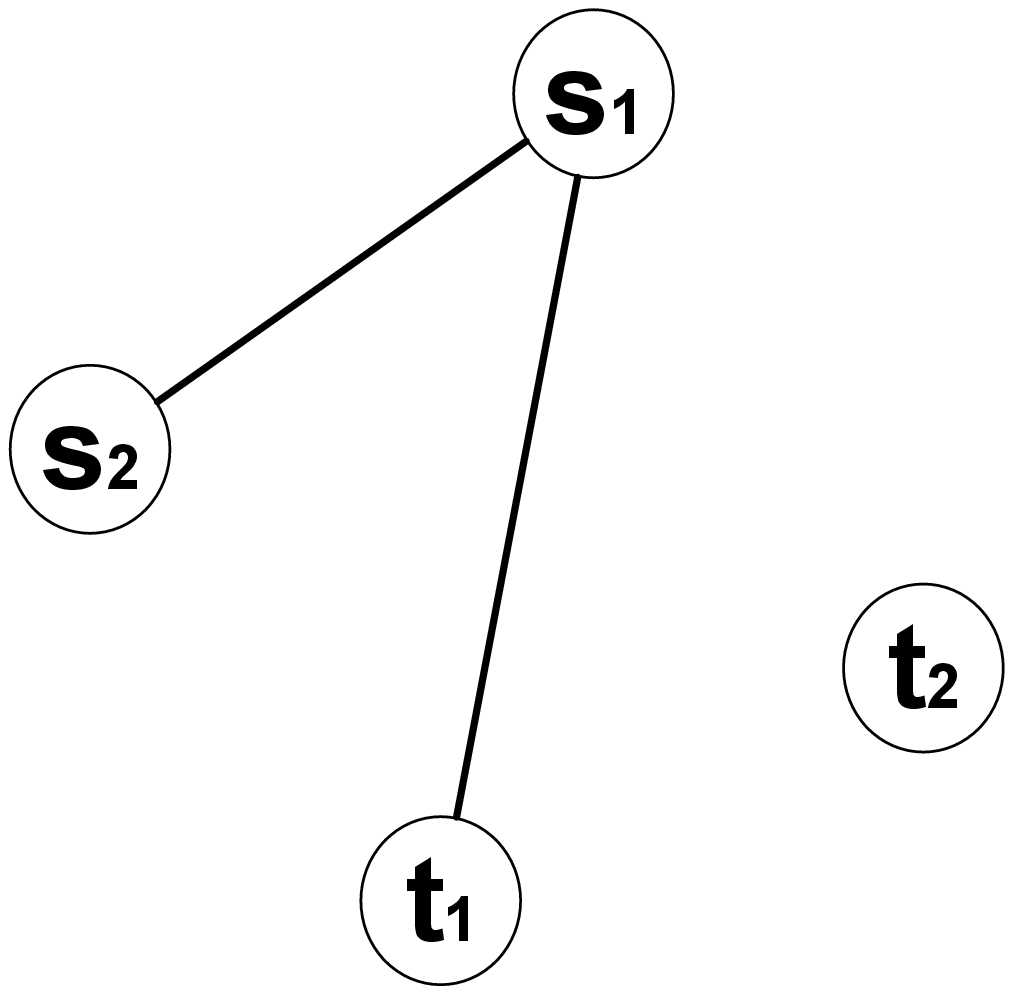}}\\
\subfigure[]{\includegraphics[width=1.832cm]{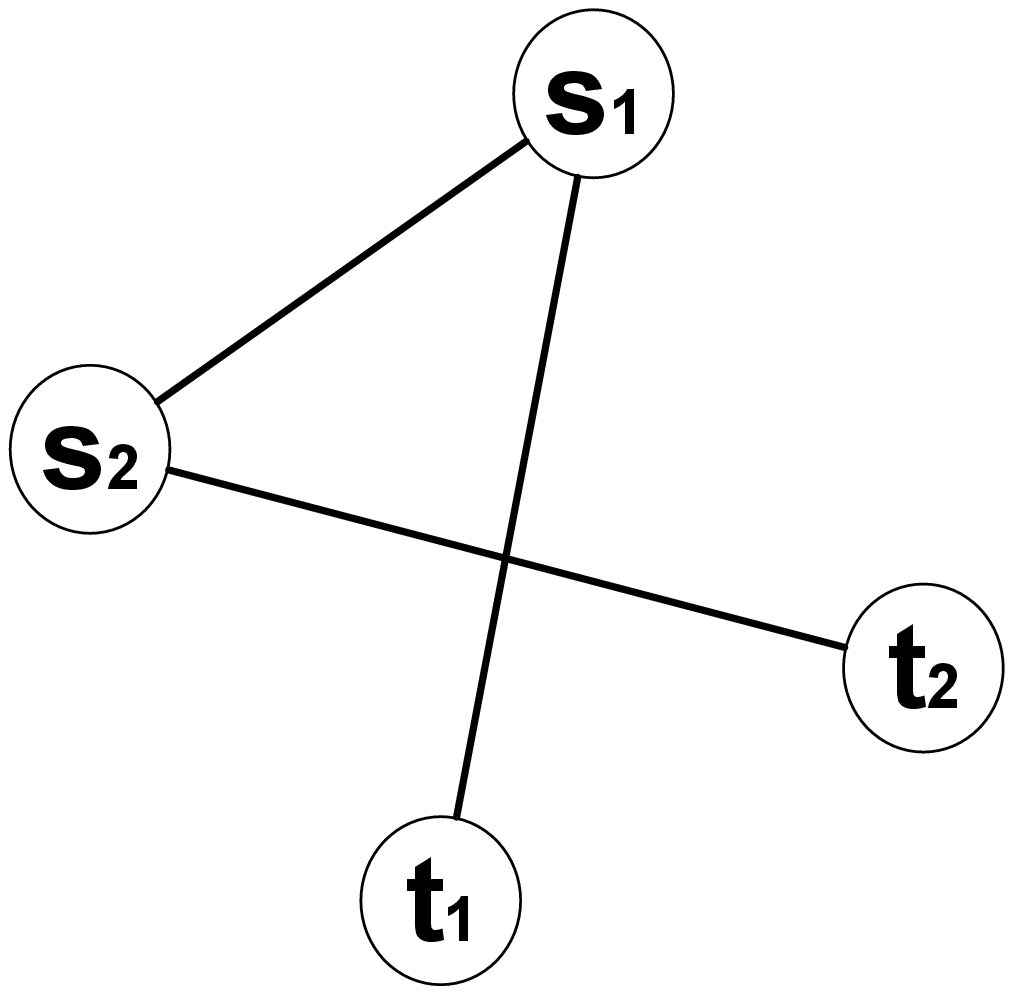}}
\subfigure[]{\includegraphics[width=1.832cm]{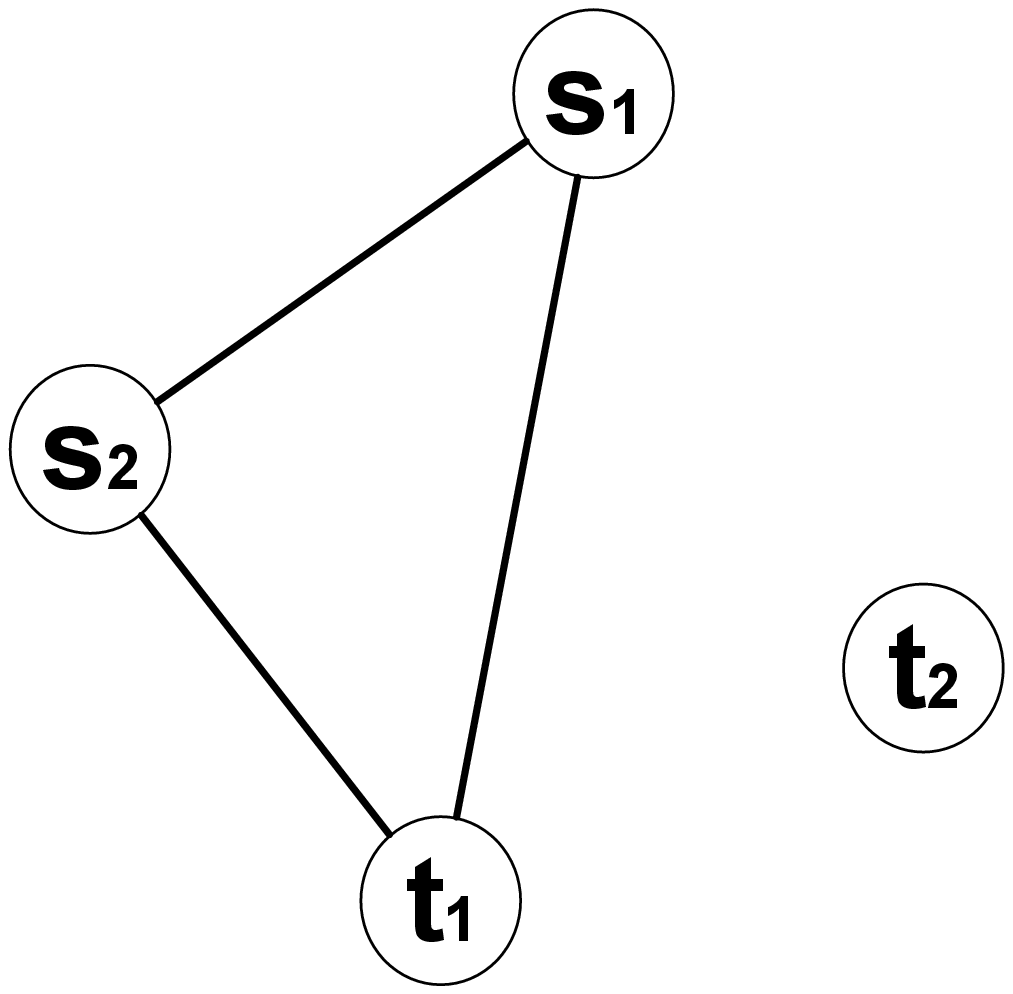}}
\subfigure[]{\includegraphics[width=1.832cm]{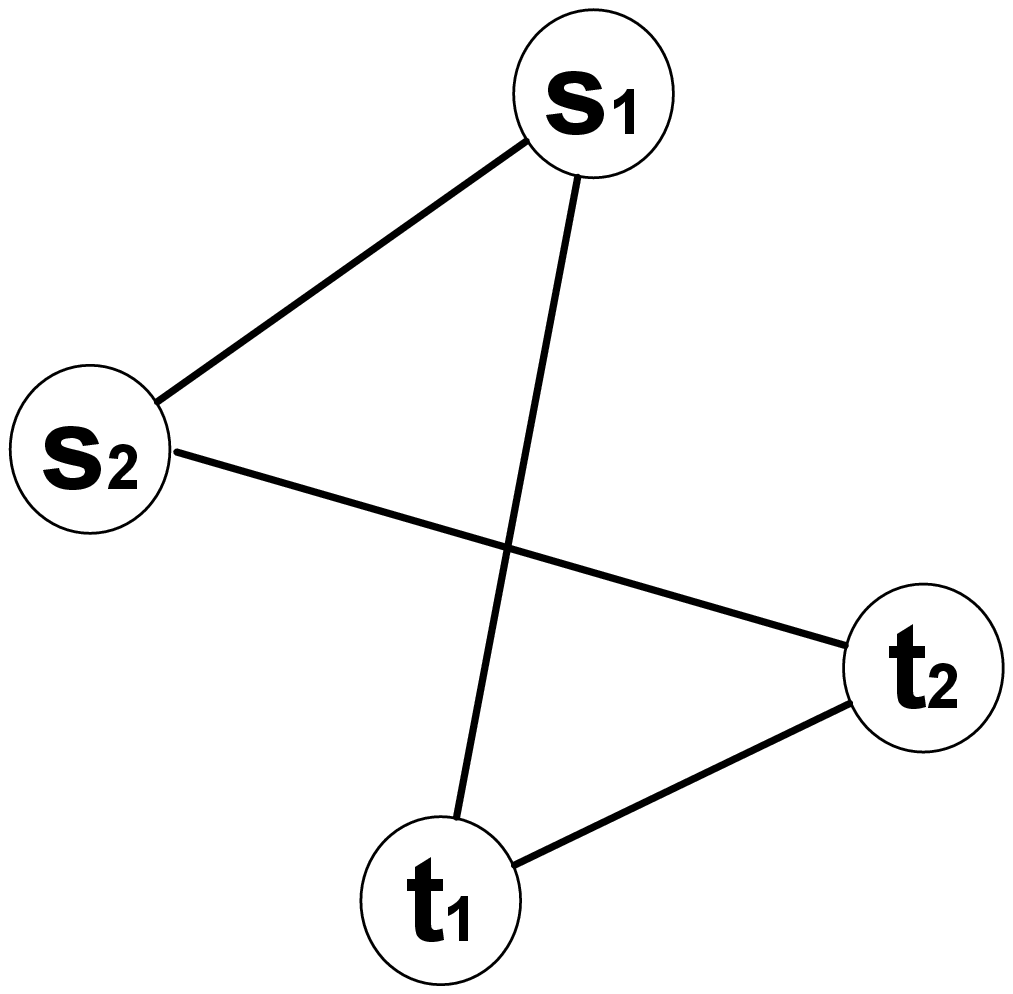}}
\subfigure[]{\includegraphics[width=1.832cm]{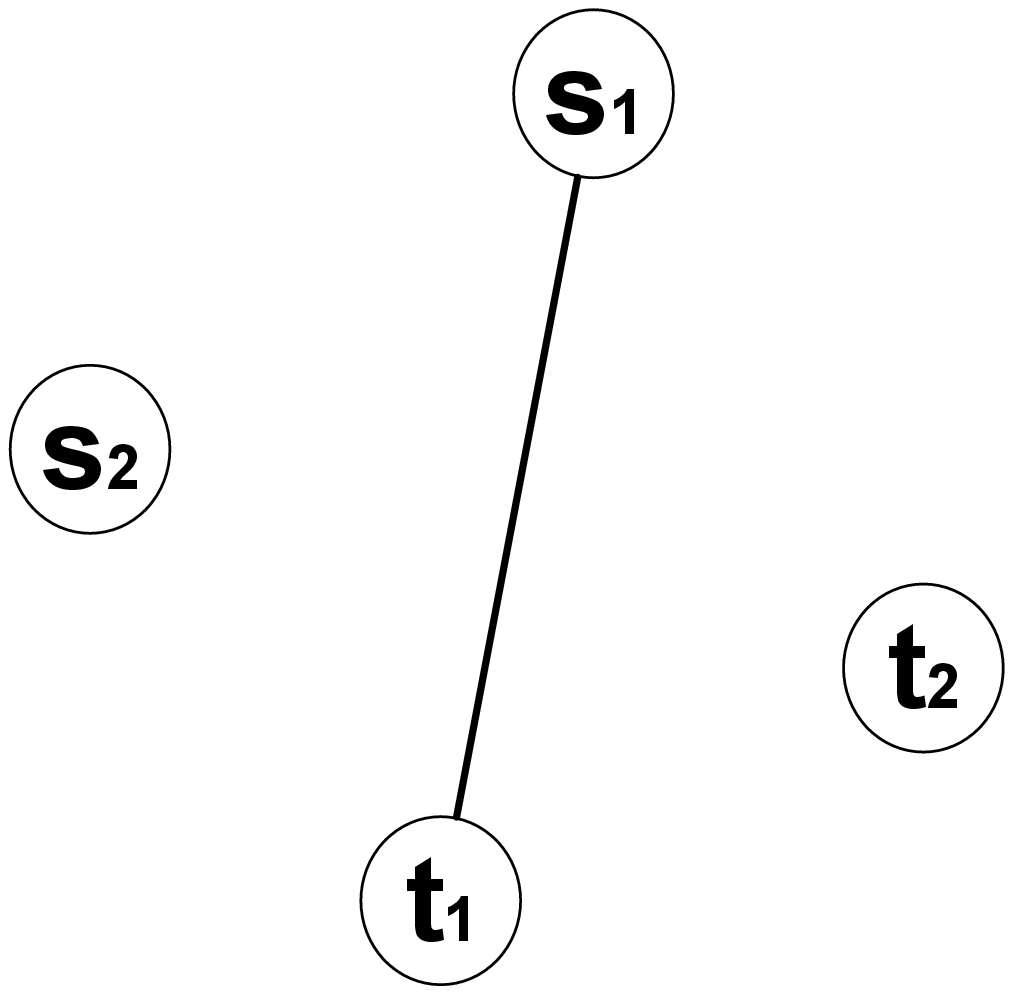}}
\subfigure[]{\includegraphics[width=1.832cm]{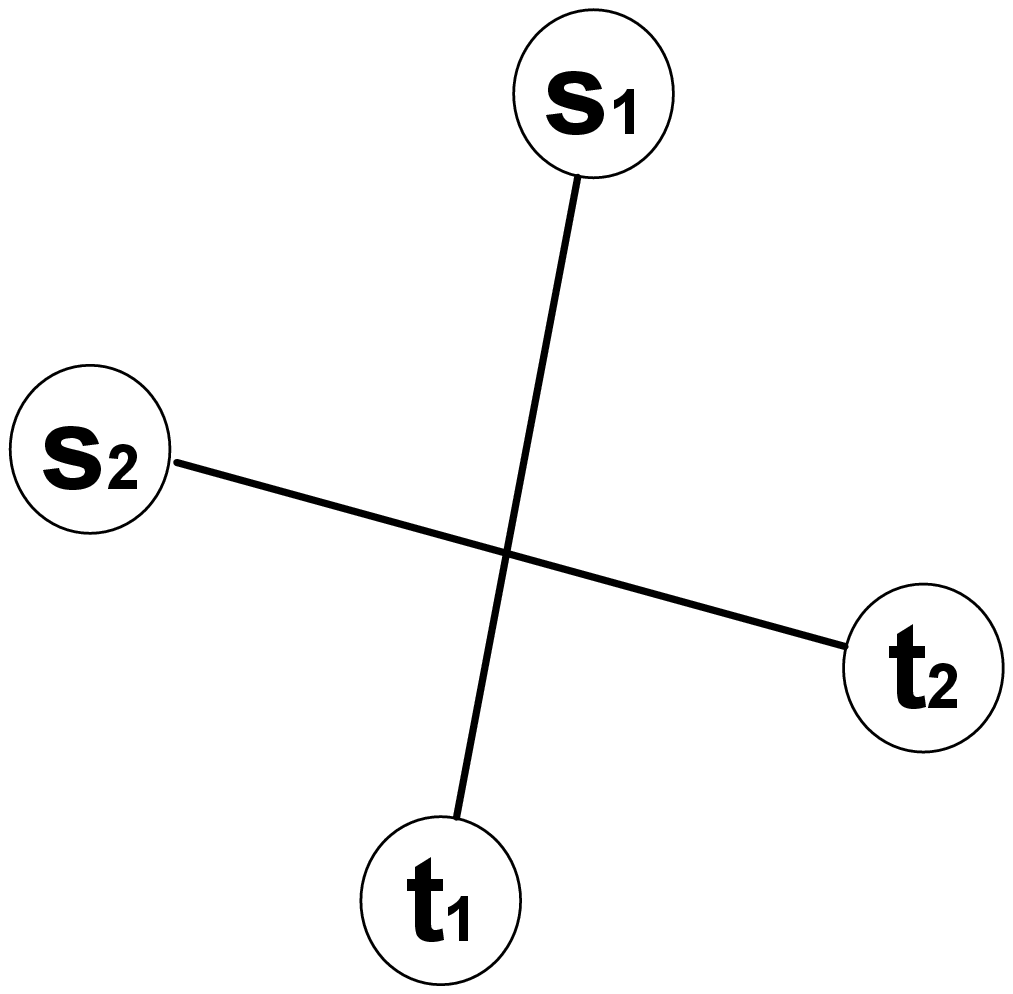}}
\subfigure[]{\includegraphics[width=1.832cm]{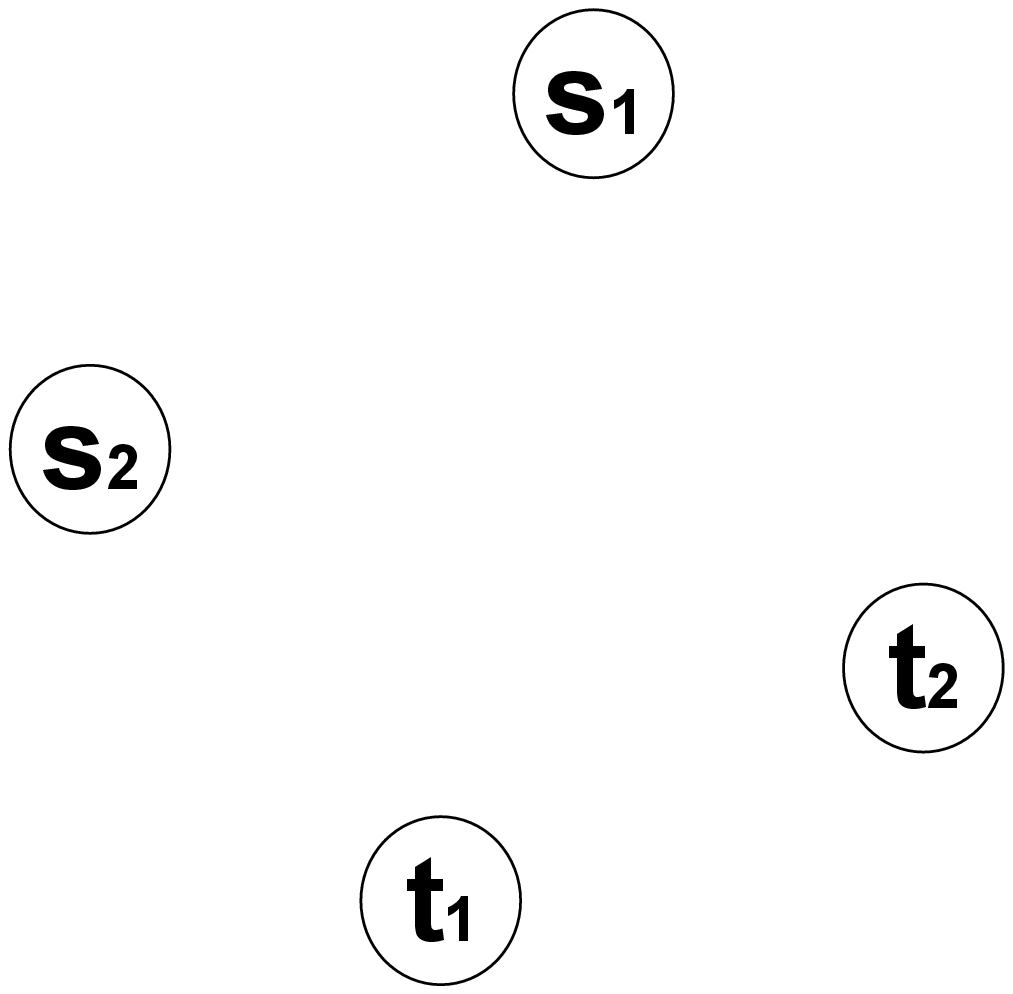}}
\caption{All topology structures consisting of $s_1, s_2, t_1, t_2.$}
\label{Fcdtopo4}
\end{center}
\end{figure}

\begin{proposition}\label{progeneve}
Suppose leaders are selected from $\mathcal{V}\setminus\{v_{s_1},v_{s_2},v_{t_1},v_{t_2}\}$ and $\bar y$ is an eigenvector of $\mathcal{L},$ then 
\begin{itemize}
\item for any given $k\neq s_1,s_2,t_1,t_2;$ $\mathcal{N}_{kf}$ conforms to one and only one of the following two situations:\\
 {\emph{i)}} at least one of cases a)  c)  e) occurs;\\
{\emph{ii)}} at least one of cases b)  c)  e) occurs.

Moreover, if b) arises, there are at most three different $k\neq s_1,s_2,t_1,t_2$ with each $\mathcal{N}_{{k}f}$ containing a different set of three indices of $\{s_1,s_2,t_1,t_2\};$ and so is to c) with each set containing two indices of $\{s_1,s_2,t_1,t_2\}.$
\item for $k= s_1, s_2, t_1, t_2;$ all possible topologies consisting of $v_{s_1}, v_{s_2}, v_{t_1}, v_{t_2}$ are depicted in Fig. \ref{Fcdtopo4}.
\end{itemize}
\end{proposition}
\begin{proof}
Consider $k\neq s_1, s_2, t_1, t_2$ and $k=s_1, s_2, t_1, t_2.$ 
In case $k\neq s_1, s_2, t_1, t_2,$ 
$
\sum_{i \in {\mathcal{N}_{kf}}} {{y_i}}  = 0
$ 
which can be shown in the same way as (\ref{trireequaio}). 
If circumstance a) arises, the same arguments as (\ref{tripleequ}) yield 
\begin{equation}
y_{s_1}+y_{s_2}+y_{t_1}+y_{t_2}=0.\label{Fcdeqn1}
\end{equation}
If circumstance b) arises and $s_1,s_2,t_1\in\mathcal{N}_{kf},$ it follows from 
$
\sum_{i \in {\mathcal{N}_{kf}}} {{y_i}}  = 0
$
that 
\begin{equation}
y_{s_1}+y_{s_2}+y_{t_1}=0. \label{Fcdeqn2}
\end{equation}
Situations (\ref{Fcdeqn1}), (\ref{Fcdeqn2}) cannot occur simultaneously, or else, $y_{t_2}=0.$ Similarly, if another  $\mathcal{N}_{kf} (k\neq s_1, s_2, t_1, t_2)$ contains, say  
$s_2,t_1,t_2,$ one has
\begin{equation}
y_{s_2}+y_{t_1}+y_{t_2}=0. \label{Fcdeqn3}
\end{equation}
(\ref{Fcdeqn2}) and (\ref{Fcdeqn3}) lead to $y_{s_2}+y_{t_1}=-y_{s_1}=-y_{t_2}.$ If there is the third $k\neq s_1, s_2, t_1, t_2$ with its $\mathcal{N}_{kf}$ containing, say $s_1,s_2,t_2,$ one has  
$
y_{s_1}+y_{s_2}+y_{t_2}=0. 
$
Combining this equation with (\ref{Fcdeqn2}) yields $y_{s_1}+y_{s_2}=-y_{t_1}=-y_{t_2}.$
If there is the fourth $k\neq s_1, s_2, t_1, t_2$ with $s_1,t_1,t_2\in\mathcal{N}_{kf},$ then
$
y_{s_1}+y_{t_1}+y_{t_2}=0. 
$
This together with (\ref{Fcdeqn3}) yields $y_{s_1}=y_{s_2}.$ Thus, if the above four situations arise at the same time, then $y_{s_1}=y_{s_2}=y_{t_1}=y_{t_2}=0,$ which contradicts to the assumption. Therefore, at most three of the above four situations occur. 

If circumstance c) arises, there are totally $C_4^2=6$ situations, i.e., $s_1,s_2\in\mathcal{N}_{kf}; s_1,t_1\in\mathcal{N}_{kf}; s_1,t_2\in\mathcal{N}_{kf};$ $s_2,t_1\in\mathcal{N}_{kf};  s_2,t_2\in\mathcal{N}_{kf}; t_1, t_2\in\mathcal{N}_{kf}.$ 
The same discussion as circumstance b) shows that the eigen-condition allows at most three of the above  situations occur. The circumstance d) cannot occur. This follows from the same discussion as c) of the Case I of TCD nodes. For circumstance e), the special form of $\bar y$ means that the condition $\sum_{i \in {\mathcal{N}_{kf}}} {{y_i}}  = 0$ is always satisfied. Thus for any given $k\neq s_1,s_2,t_1,t_2,$ $\mathcal{N}_{kf}$ conforms to one and only one of the above two cases i) and ii). 

In case 
$k= s_1, s_2, t_1, t_2,$ all possible topologies consisting of $s_1, s_2, t_1, t_2$ are generated by following the same discussion as Case II in the proof of Lemma \ref{trilem} , which are depicted in Fig.\ref{Fcdtopo4}. \hfill
\end{proof}

%The topology II mentioned at the beginning of this section corresponds to case i) when circumstances a), e) arise and c) does not. 

\begin{remark}
Proposition \ref{progeneve} greatly reduces the number of graphs required in the identification of QCD nodes. In particular, it contributes to a complete characterization of QCD nodes for graphs consisting of five nodes. To this end, the following definition and lemma are also needed. 
\end{remark}

\begin{definition}
A graph is said to be designed from (a) of Fig. \ref{Fcdtopo4} if the topology structure of $v_{s_1}, v_{s_2}, v_{t_1}, v_{t_2}$ accords with (a) and the graph is obtained by adding edges between $\{v_{s_1}, v_{s_2}, v_{t_1}, v_{t_2}\}$ and $\mathcal{V}\setminus\{v_{s_1},v_{s_2},v_{t_1},v_{t_2}\}.$ The definition applies to other topologies of Fig. \ref{Fcdtopo4}.  \end{definition}

\begin{lemma}\label{fivenolem}
Suppose $\bar y$ is an eigenvector of a graph designed from (a) of Fig. \ref{Fcdtopo4}. The following assertions hold:
\begin{itemize}
\item if the situation a) of Proposition \ref{progeneve} arises, then 
\begin{equation}
\frac{1}{{{d_{{t_2}}} - {d_{{s_1}}} - 1}} + \frac{1}{{{d_{{t_1}}} - {d_{{s_1}}} - 1}} + \frac{1}{{{d_{{s_2}}} - {d_{{s_1}}} - 1}} = -1. \label{situaeqn}
\end{equation}

\item if situation b) arises with a $v_k\in\mathcal{V}\setminus\{v_{s_1}, v_{s_2}, v_{t_1}, v_{t_2}\} $ incident to only three of $v_{s_1}, v_{s_2}, v_{t_1},v_{t_2},$ say $v_{s_1}, v_{s_2}, v_{t_1},$ then one of the following four equations must occur:
\begin{equation}
\lambda _1 = \tilde \lambda _1, \lambda _1 = \tilde \lambda _2, \lambda _2 = \tilde \lambda _1, \lambda _2 = \tilde \lambda _2, \label{foucases}
\end{equation}
where 
\begin{align}
{\lambda _{1,2}} =& \frac{{{d_{{t_1}}} + {d_{{s_2}}} + 2 \pm \sqrt {{{({d_{{s_2}}} - {d_{{t_1}}})}^2} + 4} }}{2} \label{root1}\\
{\tilde\lambda _{1,2}} =& \frac{{{d_{{s_1}}} + {d_{{t_2}}} + 1 \pm \sqrt {{{[({d_{{s_1}}} - {d_{{t_2}}}) + 1]}^2} + 4} }}{2}\label{root2}
\end{align}

\item if c) arises with a $v_k\in\mathcal{V}\setminus\{v_{s_1}, v_{s_2}, v_{t_1}, v_{t_2}\} $ incident to only two of $v_{s_1}, v_{s_2}, v_{t_1},$ $v_{t_2},$ say, $v_{s_1}, v_{s_2},$ then 
\begin{equation}
{d_{{s_1}}} - {d_{{s_2}}} = \frac{1}{{{d_{{t_1}}} - {d_{{s_2}}} - 1}} + \frac{1}{{{d_{{t_2}}} - {d_{{s_2}}} - 1}}.\label{sitaceqn}
\end{equation}
\end{itemize}
\end{lemma}
\begin{proof}
Suppose any of situations a) b) c) of Proposition \ref{progeneve} arises and the graph is designed from topology (a) of Fig. \ref{Fcdtopo4}. The eigen-condition is to be computed for $v_{s_1}, v_{s_2}, v_{t_1}, v_{t_2},$ respectively. First, for node $v_{t_2},$ since $y_k=0$ for any $k\neq s_1,s_2,t_1,t_2,$ it follows that
$
\sum_{i \in {\mathcal{N}_{{t_2}l}}} {{y_i}}  = 0,\sum_{i \in {\mathcal{N}_{{t_2}f}}} {{y_i}}  = {y_{{s_1}}}.
$
Accordingly
$
{d_{{t_2}}}{y_{{t_2}}} - \sum_{i \in {\mathcal{N}_{{t_2}}}} {{y_i}}  = {d_{{t_2}}}{y_{{t_2}}} - {y_{{s_1}}}.
$
So the eigen-condition 
requires 
\begin{equation}
({d_{{t_2}}} - \lambda){y_{{t_2}}} = {y_{{s_1}}}.\label{fifcda}
\end{equation}
Similarly, the eigen-conditions of $v_{t_1}$ and $v_{s_2}$ require that 
\begin{equation}
({d_{{t_1}}} - \lambda ){y_{{t_1}}} = {y_{{s_1}}} ~ ~%\label{fcda2eqn}~
\mbox{and} ~~
({d_{{s_2}}} - \lambda ){y_{{s_2}}} = {y_{{s_1}}}.\label{fcda3eqn}
\end{equation}
For $v_{s_1},$ since 
$
\sum_{i \in {\mathcal{N}_{{s_1}l}}} {{y_i}}  = 0,\sum_{i \in {\mathcal{N}_{{s_1}f}}} {{y_i}}  = {y_{{s_2}}} + {y_{{t_1}}} + {y_{{t_2}}},
$
one has 
$
{d_{{s_1}}}{y_{{s_1}}} - \sum_{i \in {\mathcal{N}_{{s_1}}}} {{y_i}}  = {d_{{s_1}}}{y_{{s_1}}} - ({y_{{s_2}}} + {y_{{t_1}}} + {y_{{t_2}}}).
$
Then the eigen-condition associated with $v_{s_1}$ requires
\begin{equation}
({d_{{s_1}}} - \lambda ){y_{{s_1}}} = {y_{{s_2}}} + {y_{{t_1}}} + {y_{{t_2}}}.\label{s1eigcon}
\end{equation}
Since $y_{s_1}\neq 0$ and $\bar y$ is an eigenvector, it can be assumed that $y_{s_1}=1.$ 
Consider the following circumstances.
\begin{itemize}
\item Situation a) of Proposition \ref{progeneve} arises with a $v_k\in\mathcal{V}\setminus\{v_{s_1}, v_{s_2}, v_{t_1}, v_{t_2}\} $ incident to all $v_{s_1}, v_{s_2}, v_{t_1}, v_{t_2}.$ In this situation, (\ref{Fcdeqn1}) holds. By (\ref{s1eigcon}), 
$
({d_{{s_1}}} - \lambda +1){y_{{s_1}}} =0.
$
Since $y_{s_1}\neq 0,$ 
$
\lambda=d_{s_1}+1.
$
Substituting $\lambda$, (\ref{fifcda}) and (\ref{fcda3eqn}) into (\ref{Fcdeqn1}) yields (\ref{situaeqn}).
Thus, if $\bar y$ is an eigenvector, condition (\ref{situaeqn}) ought to be satisfied.

\item Situation b) arises with a $v_k\in\mathcal{V}\setminus\{v_{s_1}, v_{s_2}, v_{t_1}, v_{t_2}\} $ incident to only three of $v_{s_1}, v_{s_2}, v_{t_1},$ $v_{t_2},$ say $v_{s_1}, v_{s_2}, v_{t_1}.$
In this situation, (\ref{Fcdeqn2}) holds. Substituting (\ref{Fcdeqn2}) into (\ref{fcda3eqn}) yields  
$
({d_{{t_1}}} - \lambda  + 1){y_{{t_1}}} =  - {y_{{s_2}}}
$
and 
$
({d_{{s_2}}} - \lambda  + 1){y_{{s_2}}} =  - {y_{{t_1}}}.
$
Thus
$
({d_{{s_2}}} - \lambda  + 1)({d_{{t_1}}} - \lambda  + 1){y_{{t_1}}} = {y_{{t_1}}}.
$
Since $y_{{t_1}}\neq 0,$ 
$
({d_{{s_2}}} - \lambda  + 1)({d_{{t_1}}} - \lambda  + 1) = 1
$
whose roots are (\ref{root1}). On the other hand, combining (\ref{s1eigcon}) with (\ref{Fcdeqn2}) yields
$
y_{{t_2}}={d_{{s_1}}} - \lambda  + 1.
$
By (\ref{fifcda}), 
$
{y_{{t_2}}} = \frac{1}{{{d_{{t_2}}} - \lambda }}.
$
Thus
$
{d_{{s_1}}} - \lambda  + 1 = \frac{1}{{{d_{{t_2}}} - \lambda }}, 
$
i.e., 
\begin{equation}
{\lambda ^2} - ({d_{{s_1}}} + {d_{{t_2}}} + 1)\lambda  + {d_{{t_2}}}{d_{{s_1}}} + {d_{{t_2}}} - 1 = 0.\label{eigenvalue2}
\end{equation}
The two roots of (\ref{eigenvalue2}) are (\ref{root2}). Because the eigen-condition of each node holds for the same eigenvalue $\lambda,$ it follows from (\ref{root1}) and (\ref{root2}) that one of the four cases of (\ref{foucases}) must occur. 

\item Situation c) arises with a $v_k\in\mathcal{V}\setminus\{v_{s_1}, v_{s_2}, v_{t_1}, v_{t_2}\} $ incident to only two of $v_{s_1}, v_{s_2}, v_{t_1},$ $v_{t_2},$ say $v_{s_1}, v_{s_2}.$ 
Similar arguments as (\ref{Fcdeqn2}) yields
$
y_{s_1}+y_{s_2}=0. 
$
Substituting this with $y_{s_1}=1$ into (\ref{fifcda}) (\ref{fcda3eqn}) and (\ref{s1eigcon}) results in $\lambda  = {d_{{s_2}}} + 1$ and accordingly (\ref{sitaceqn}) should be met.
\end{itemize}
\end{proof}

 \begin{remark}
 Lemma \ref{fivenolem} serves to check whether $\bar y$ is an eigenvector of a graph designed from (a) of Fig. \ref{Fcdtopo4} and accordingly contributes to the discrimination of topologies of QCD nodes. Graphs designed from other topologies of Fig. \ref{Fcdtopo4} can be analyzed in the same manner. This provides a method of identifying topologies of QCD nodes by which all topology structures of QCD nodes are to be revealed for graphs composed of five vertices. 
 \end{remark}
By Proposition \ref{progeneve}, the following candidate graphs consisting of five vertices are designed to discriminate topologies of QCD nodes. 
\begin{figure}[H]
\begin{center}
\subfigure[]{\includegraphics[width=1.4cm]{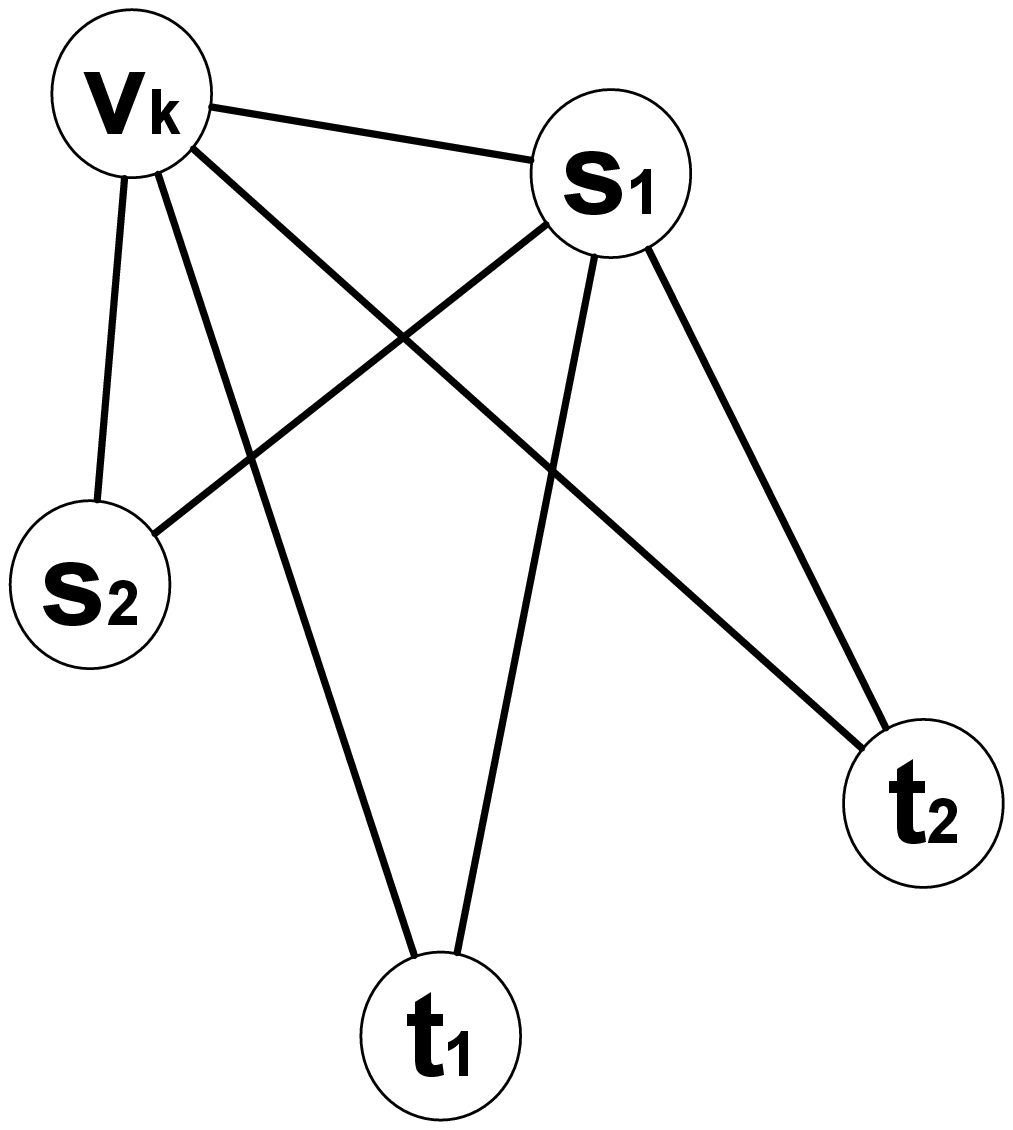}}
\subfigure[]{\includegraphics[width=1.4cm]{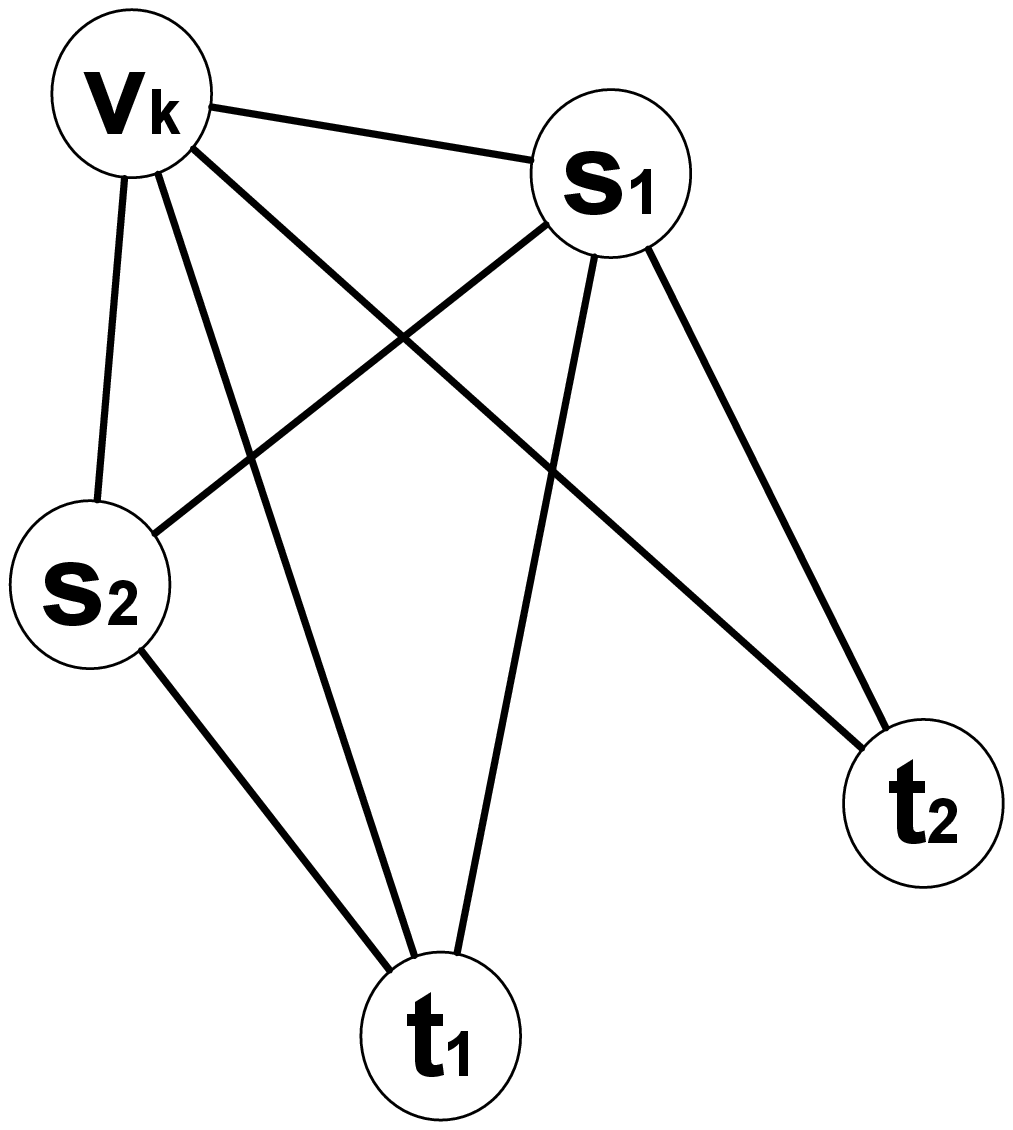}}
\subfigure[]{\includegraphics[width=1.4cm]{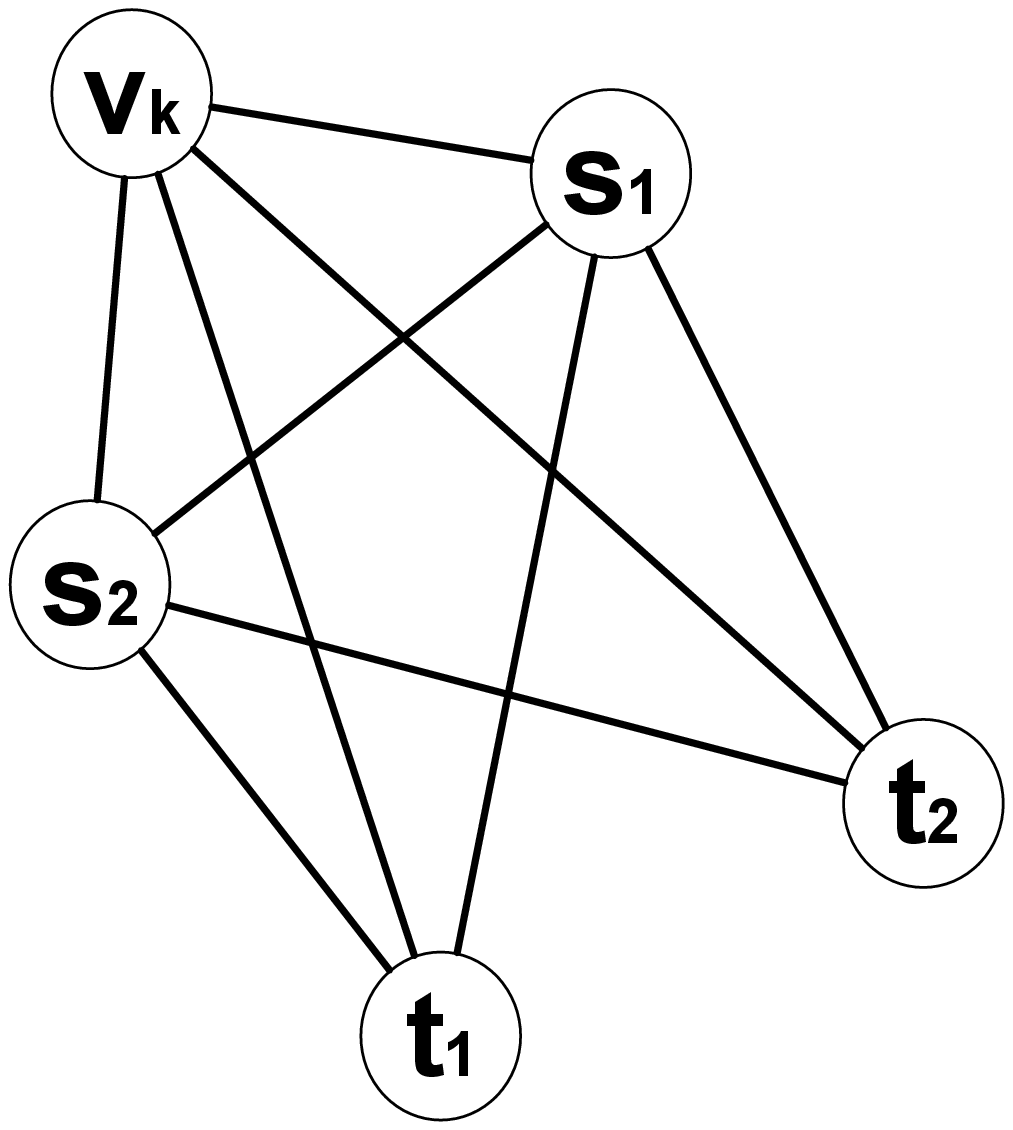}}
\subfigure[]{\includegraphics[width=1.4cm]{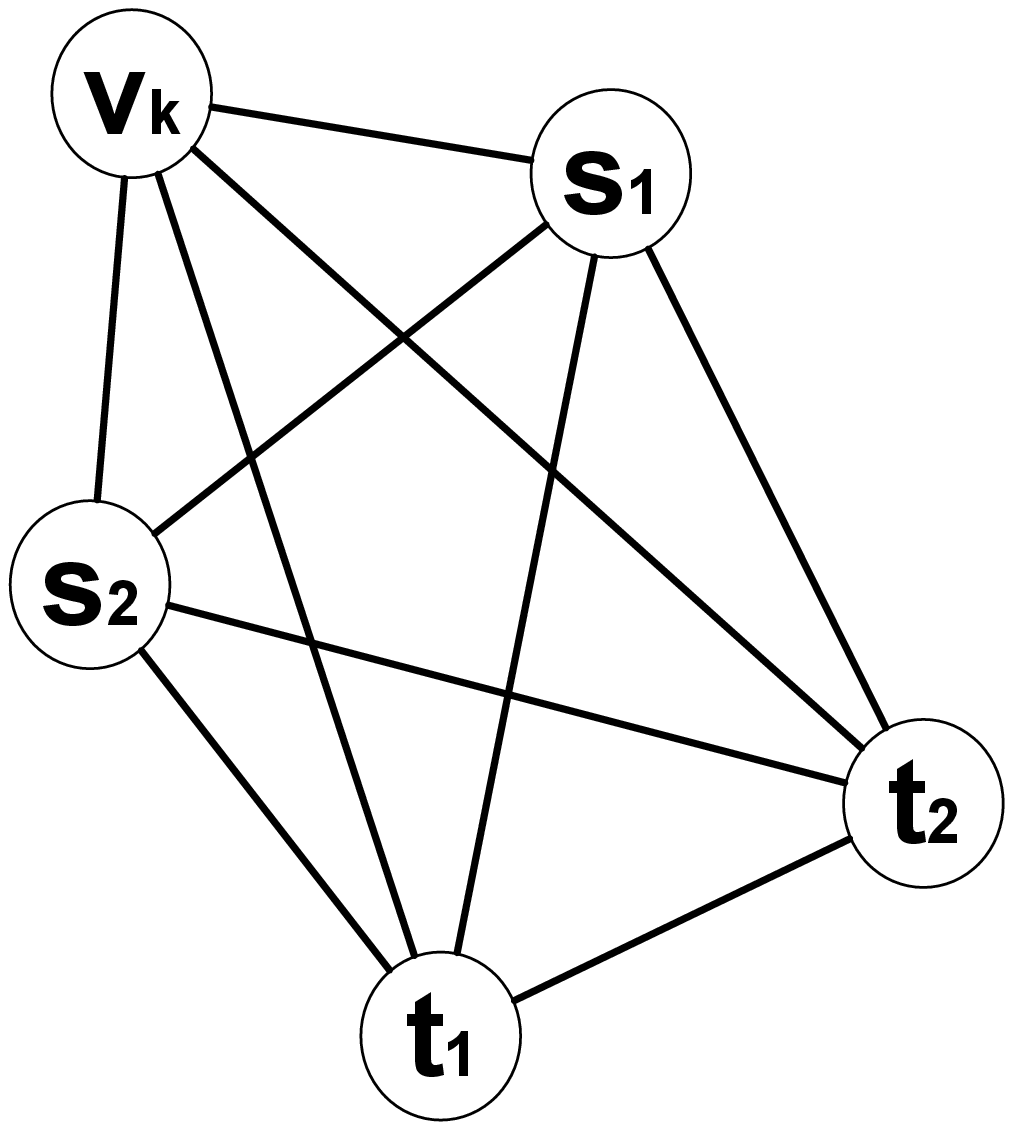}}
\subfigure[]{\includegraphics[width=1.4cm]{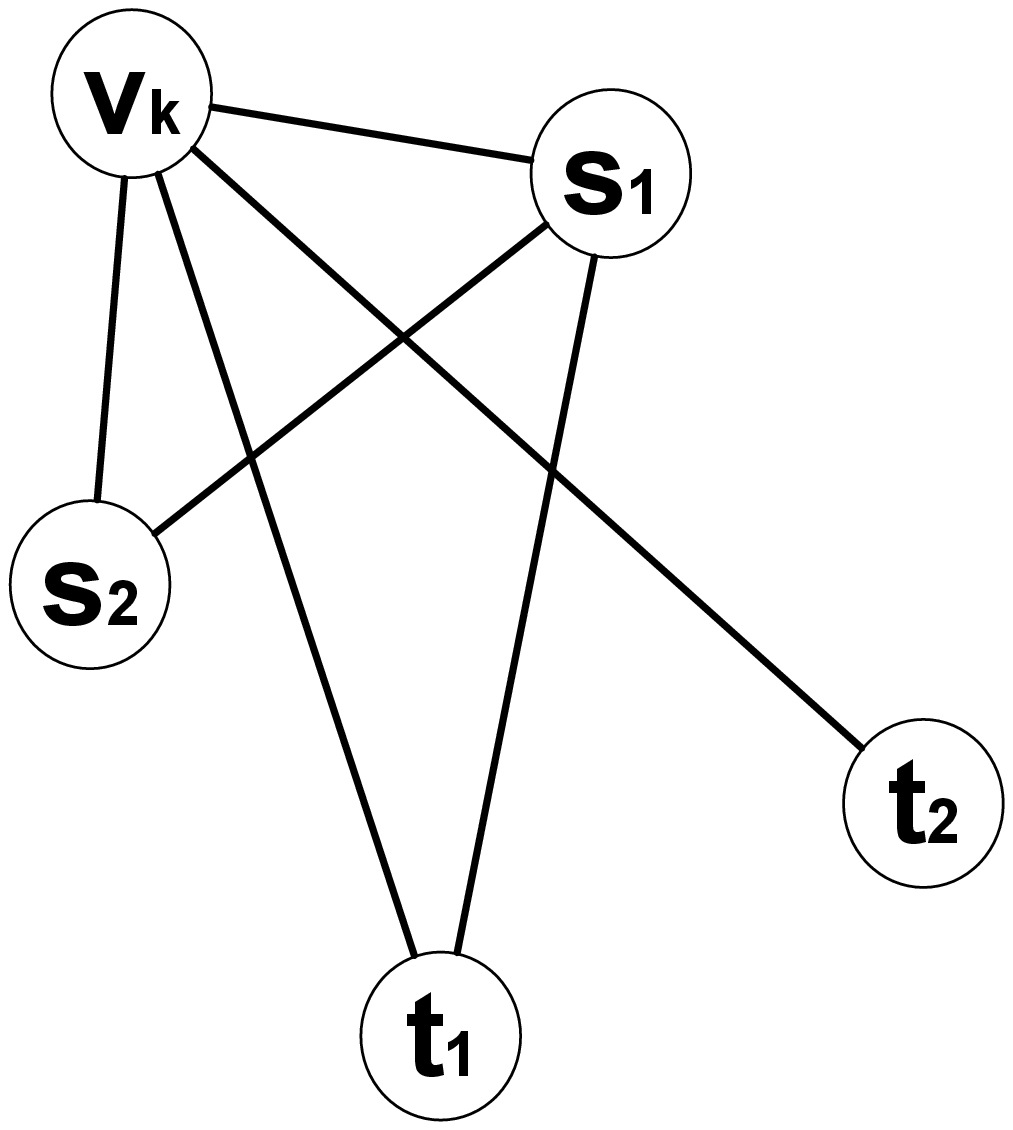}}
\subfigure[]{\includegraphics[width=1.43cm]{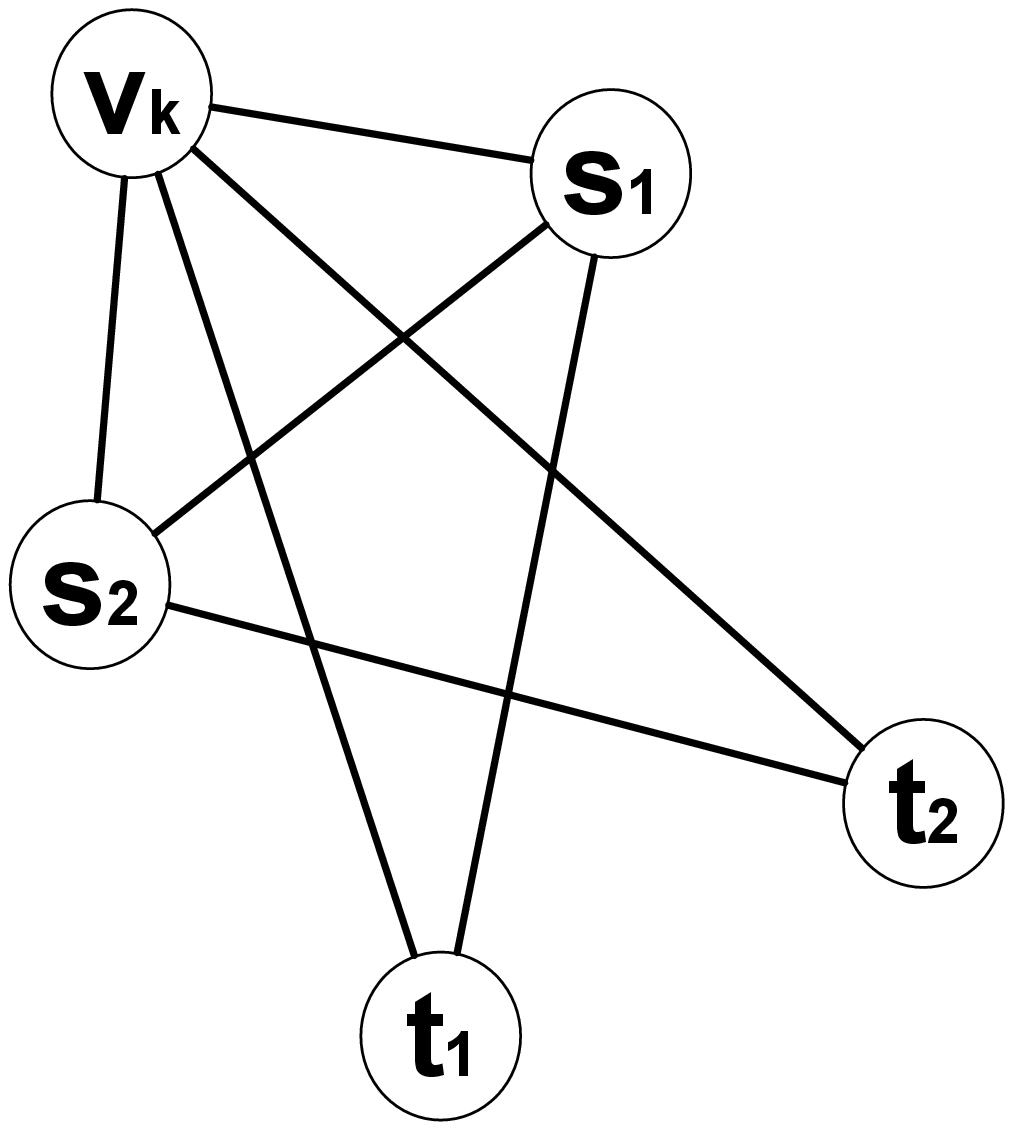}}
\subfigure[]{\includegraphics[width=1.43cm]{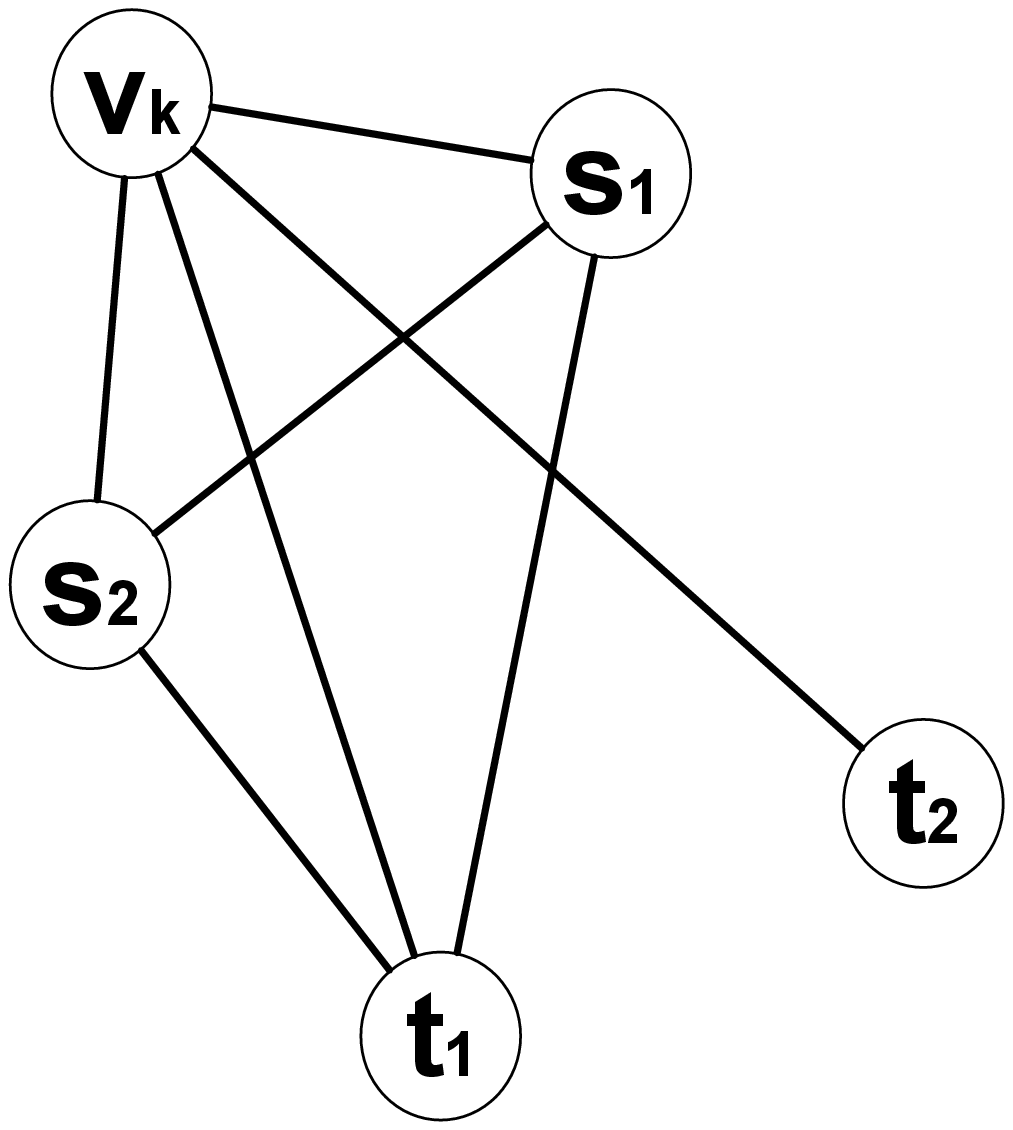}}
\subfigure[]{\includegraphics[width=1.43cm]{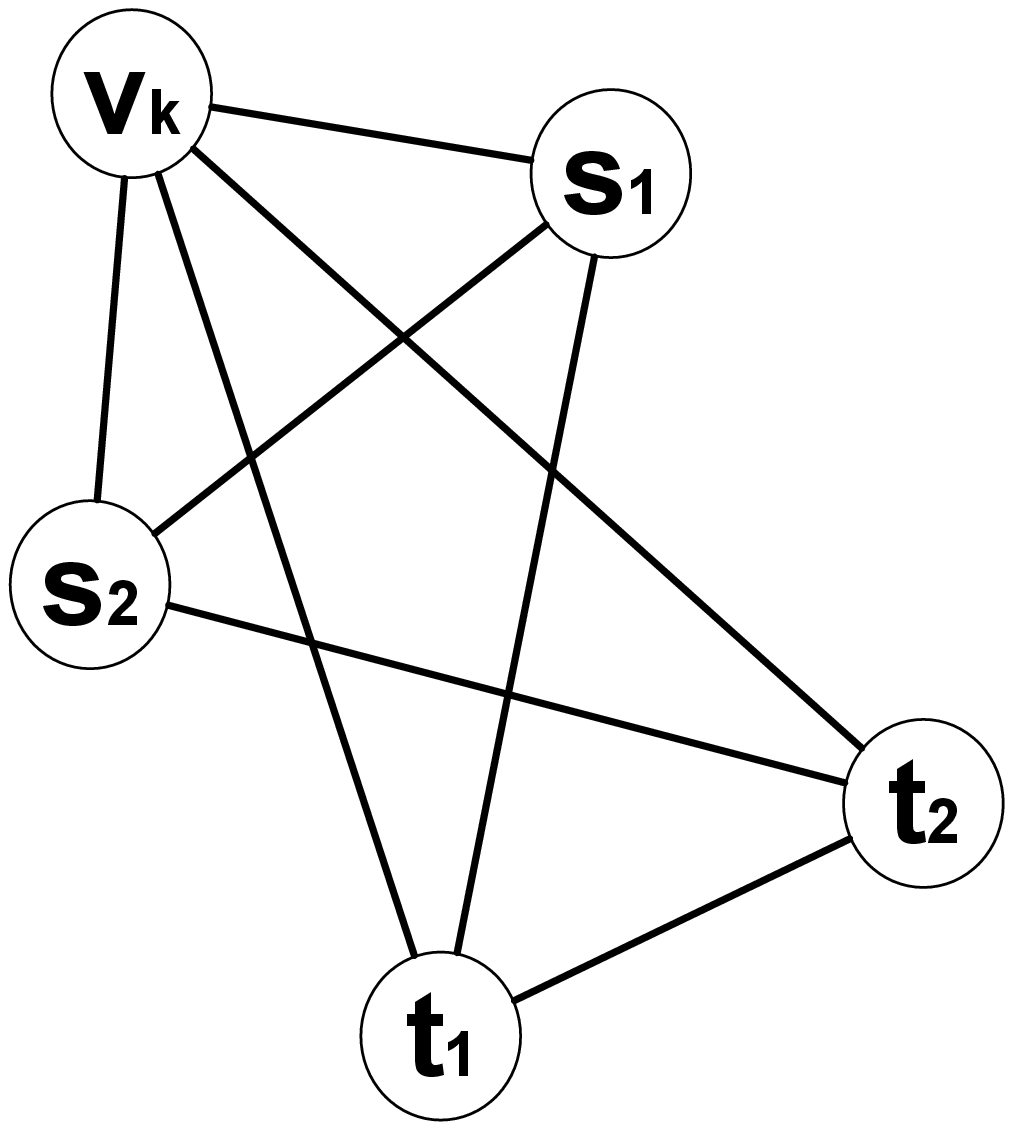}}
\subfigure[]{\includegraphics[width=1.43cm]{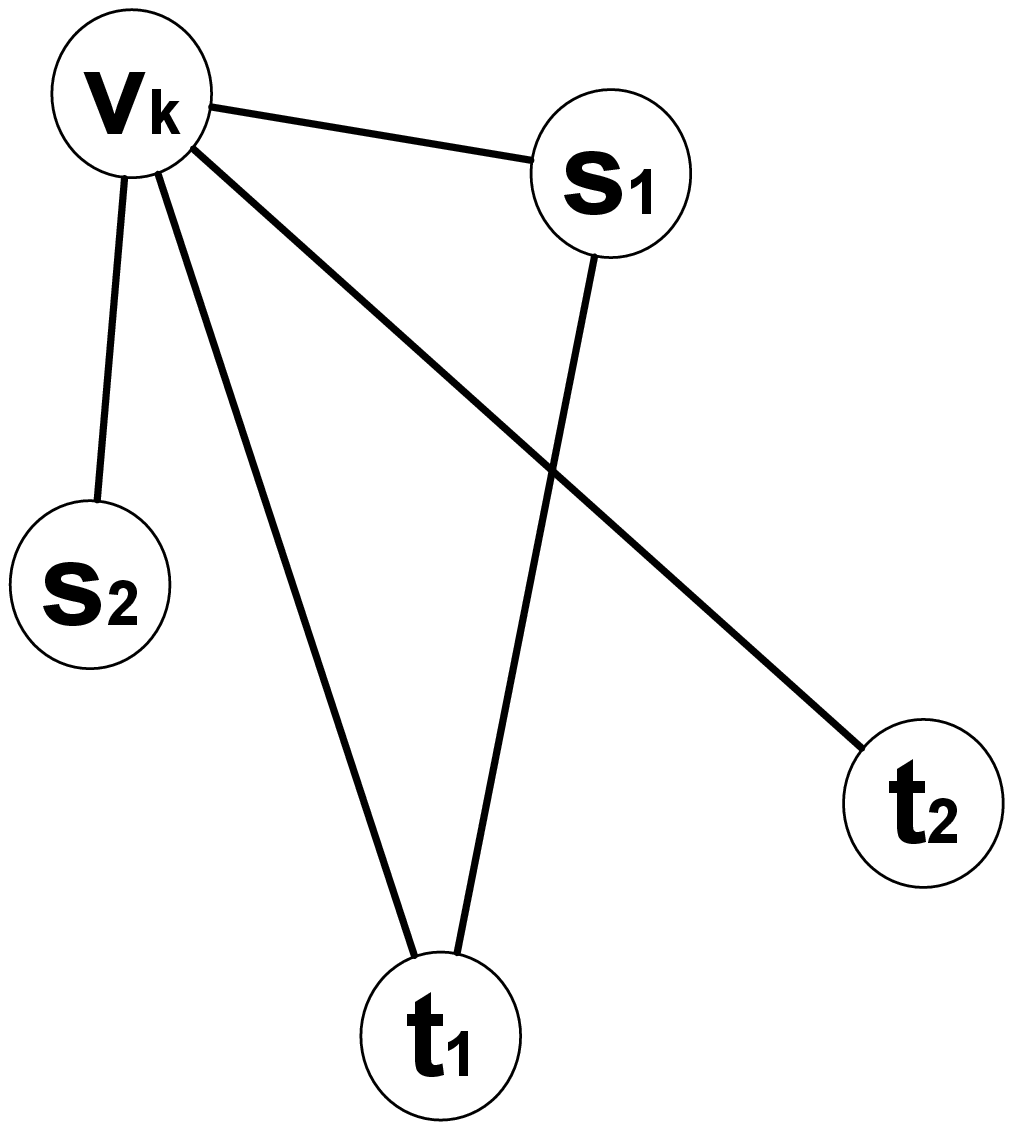}}
\subfigure[]{\includegraphics[width=1.43cm]{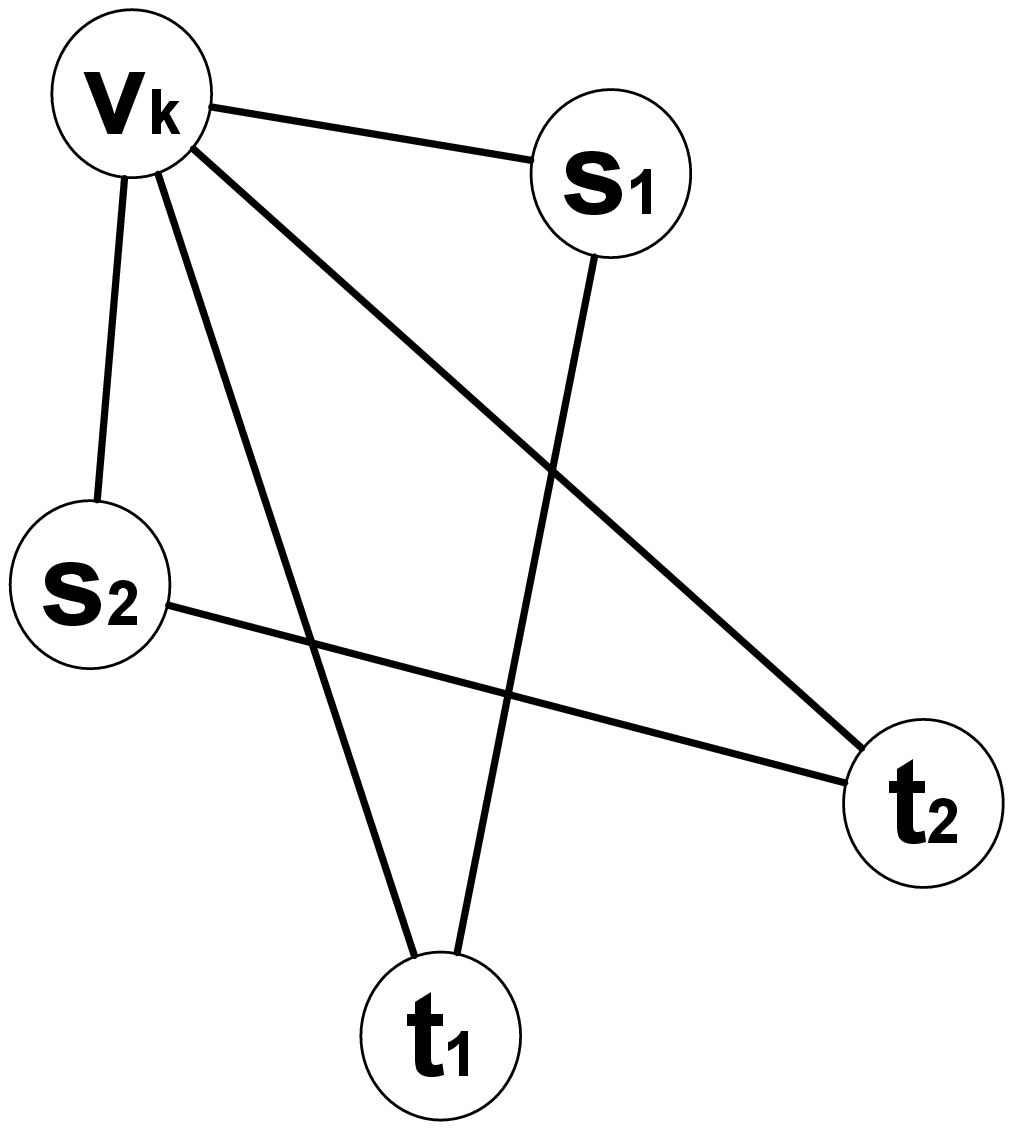}}
\subfigure[]{\includegraphics[width=1.43cm]{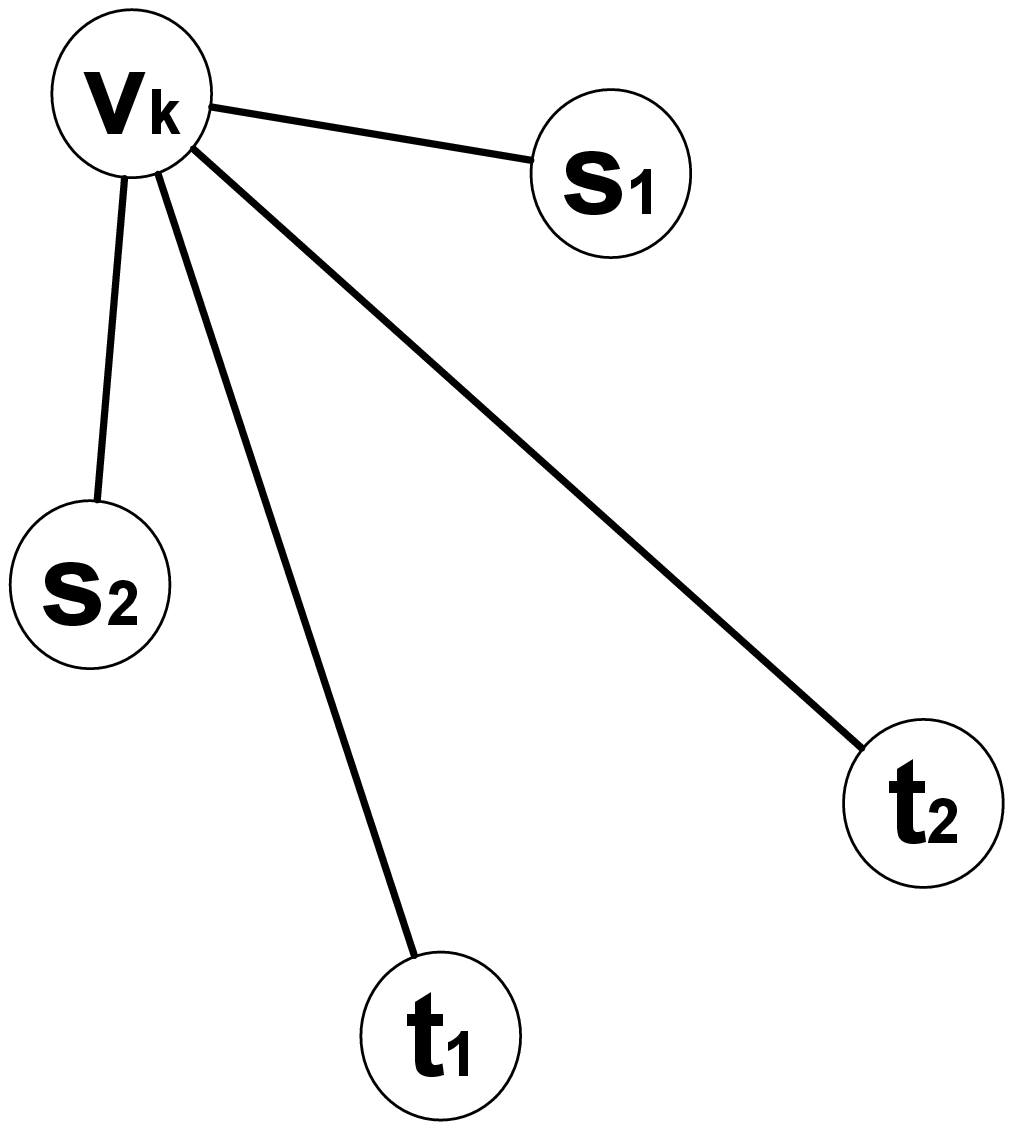}}
\caption{Graphs abiding by situation a) of Proposition \ref{progeneve}, where (a)-(k) are designed, respectively, from the topology structures (a)-(k) of Fig. \ref{Fcdtopo4}.}
\label{Fivnodeelv}
\end{center}
\end{figure}
\begin{figure}[H]
\begin{center}
\subfigure[]{\includegraphics[width=1.42cm]{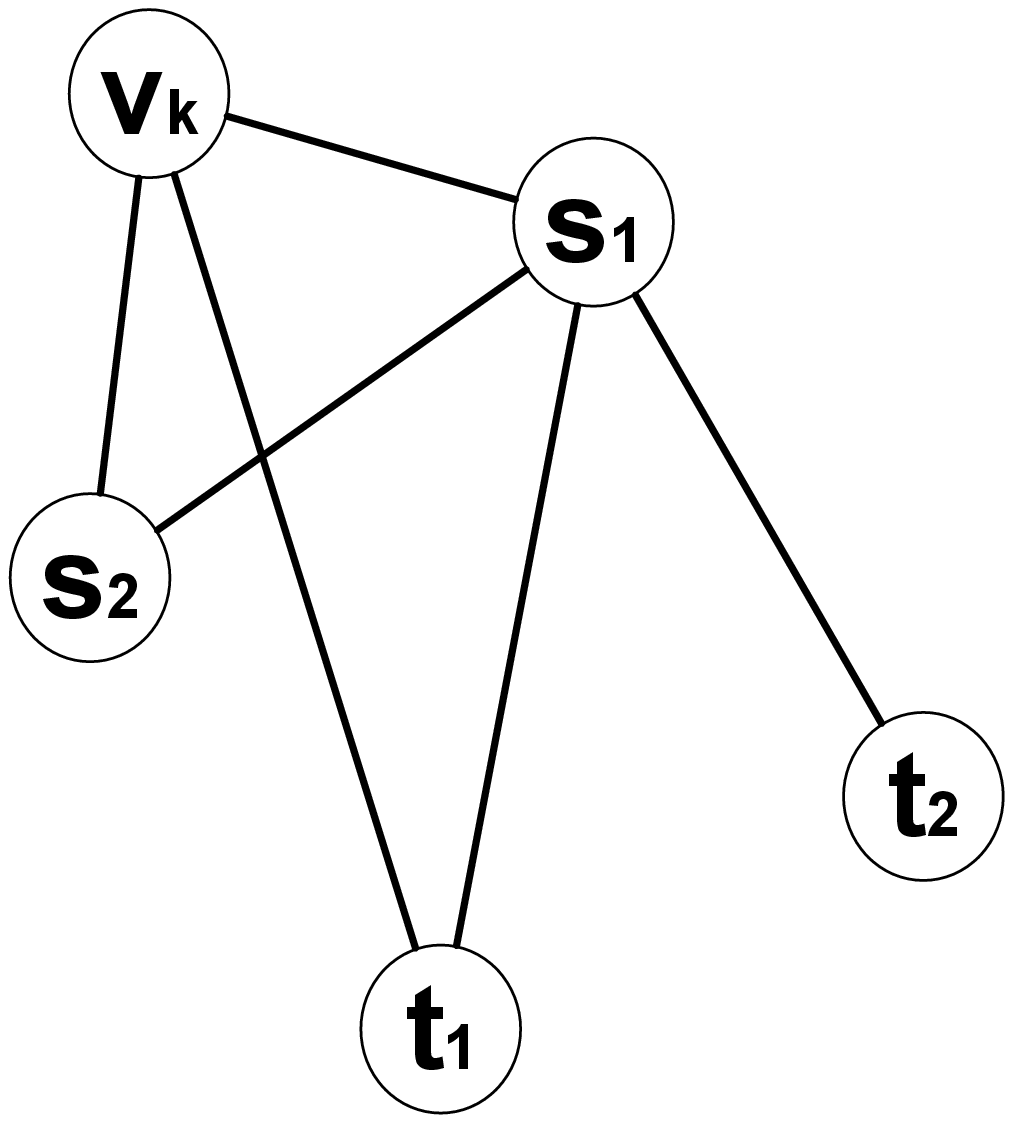}}
\subfigure[]{\includegraphics[width=1.42cm]{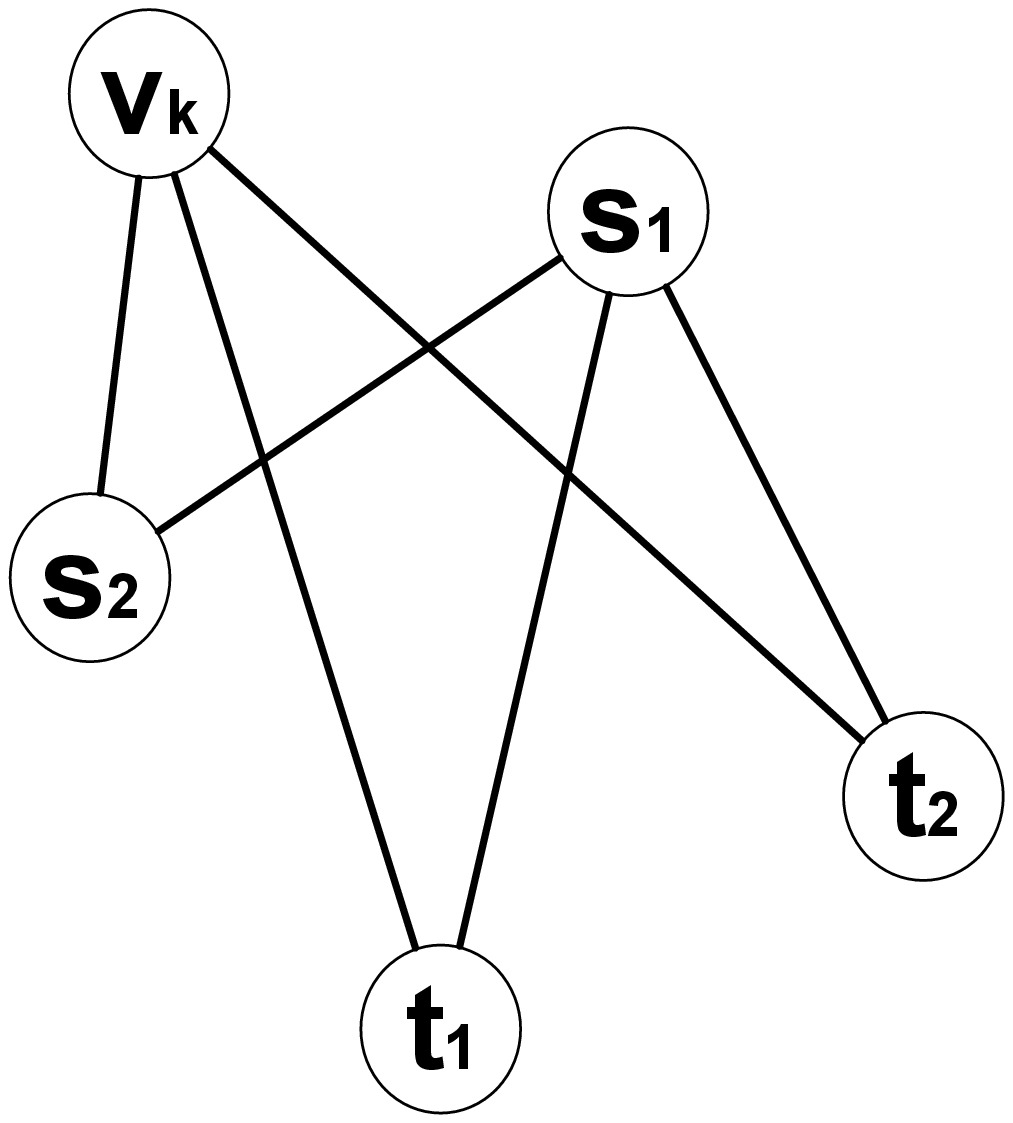}}
\subfigure[]{\includegraphics[width=1.42cm]{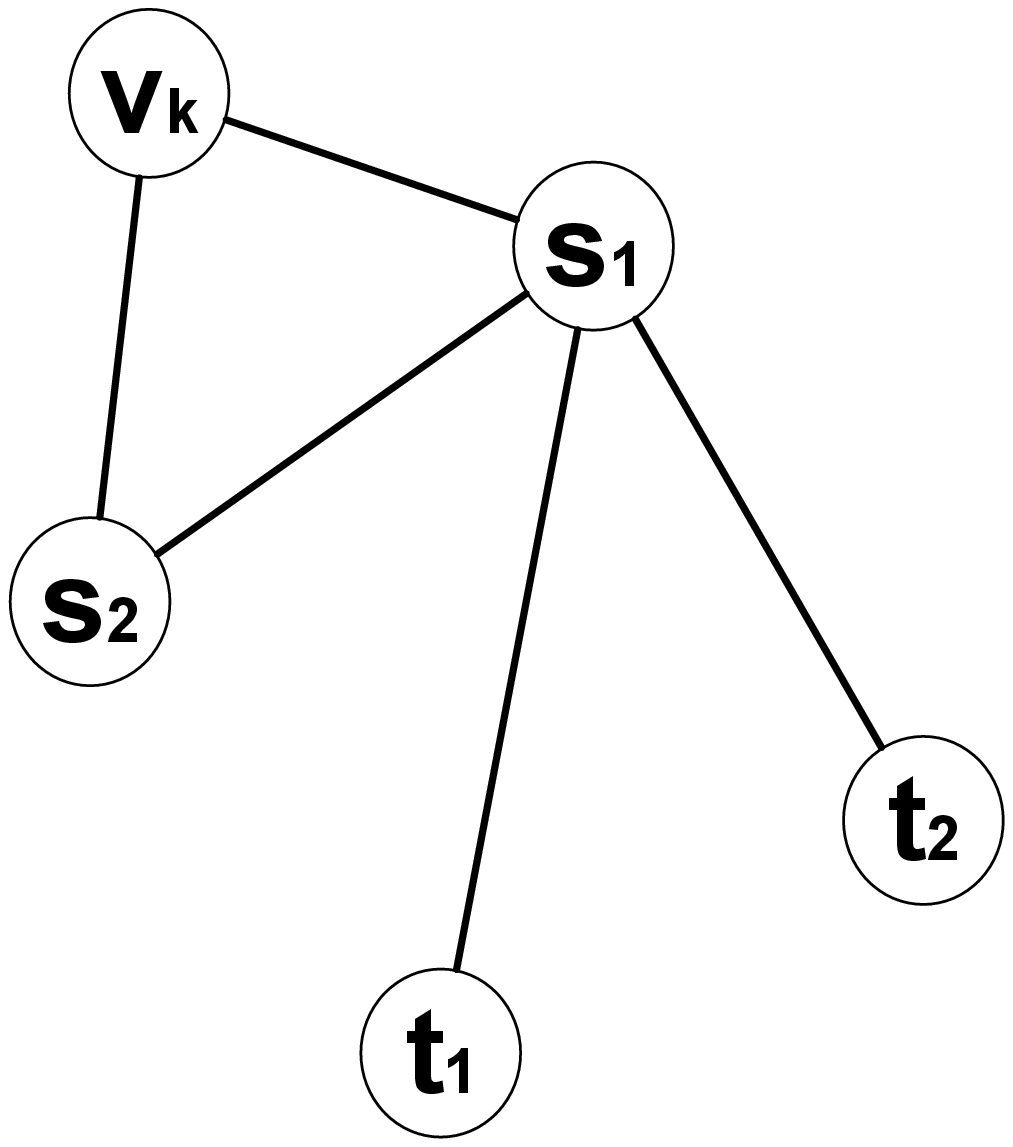}}
\subfigure[]{\includegraphics[width=1.42cm]{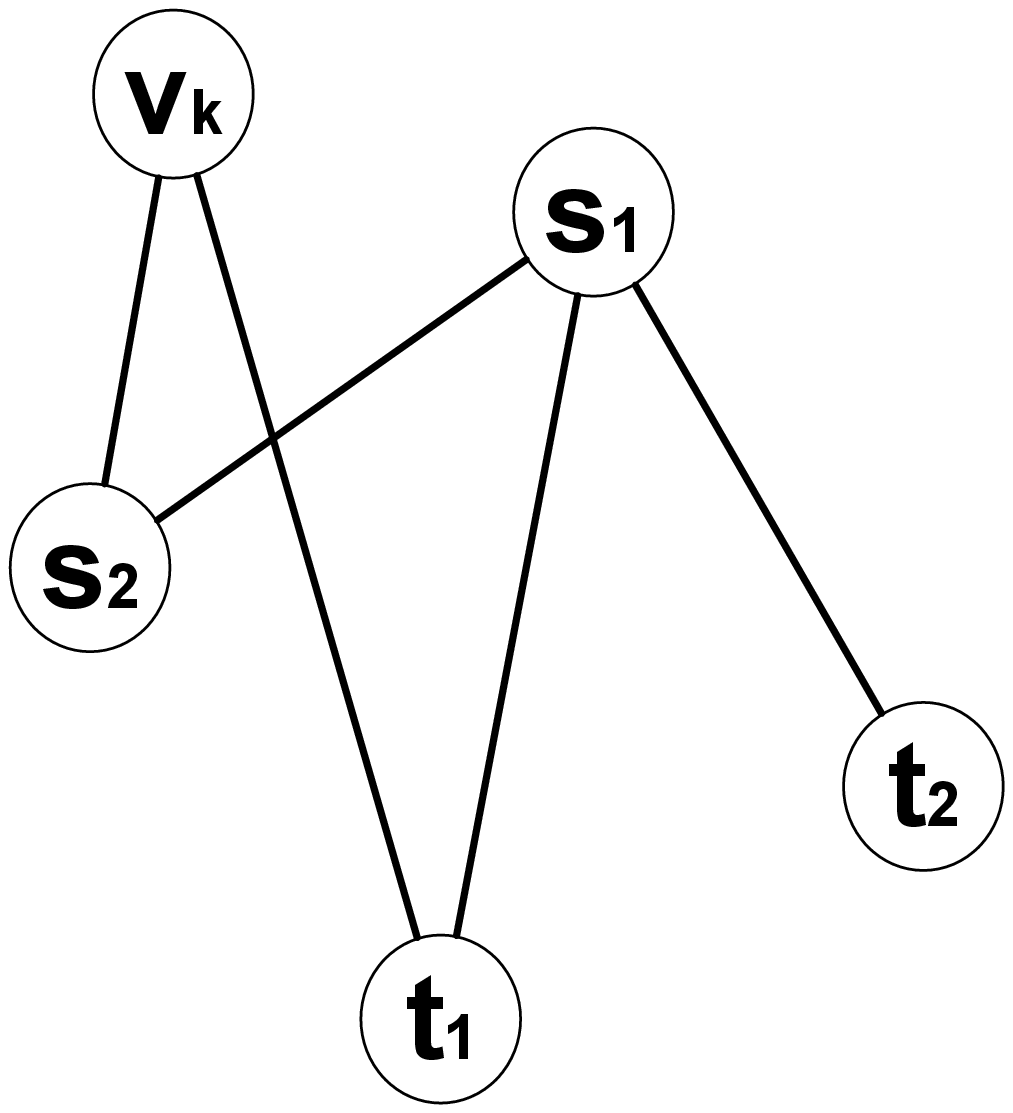}}
\subfigure[]{\includegraphics[width=1.42cm]{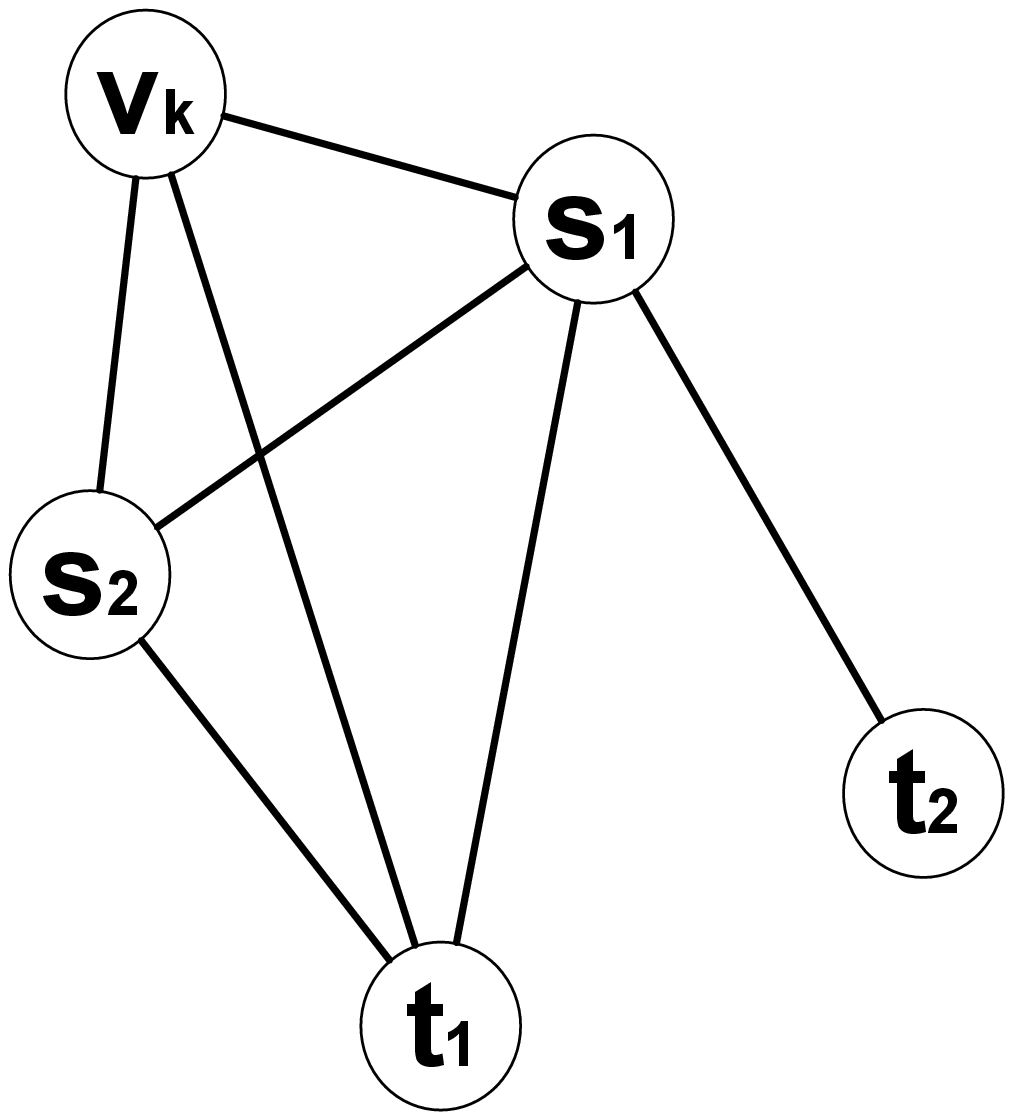}}
\subfigure[]{\includegraphics[width=1.43cm]{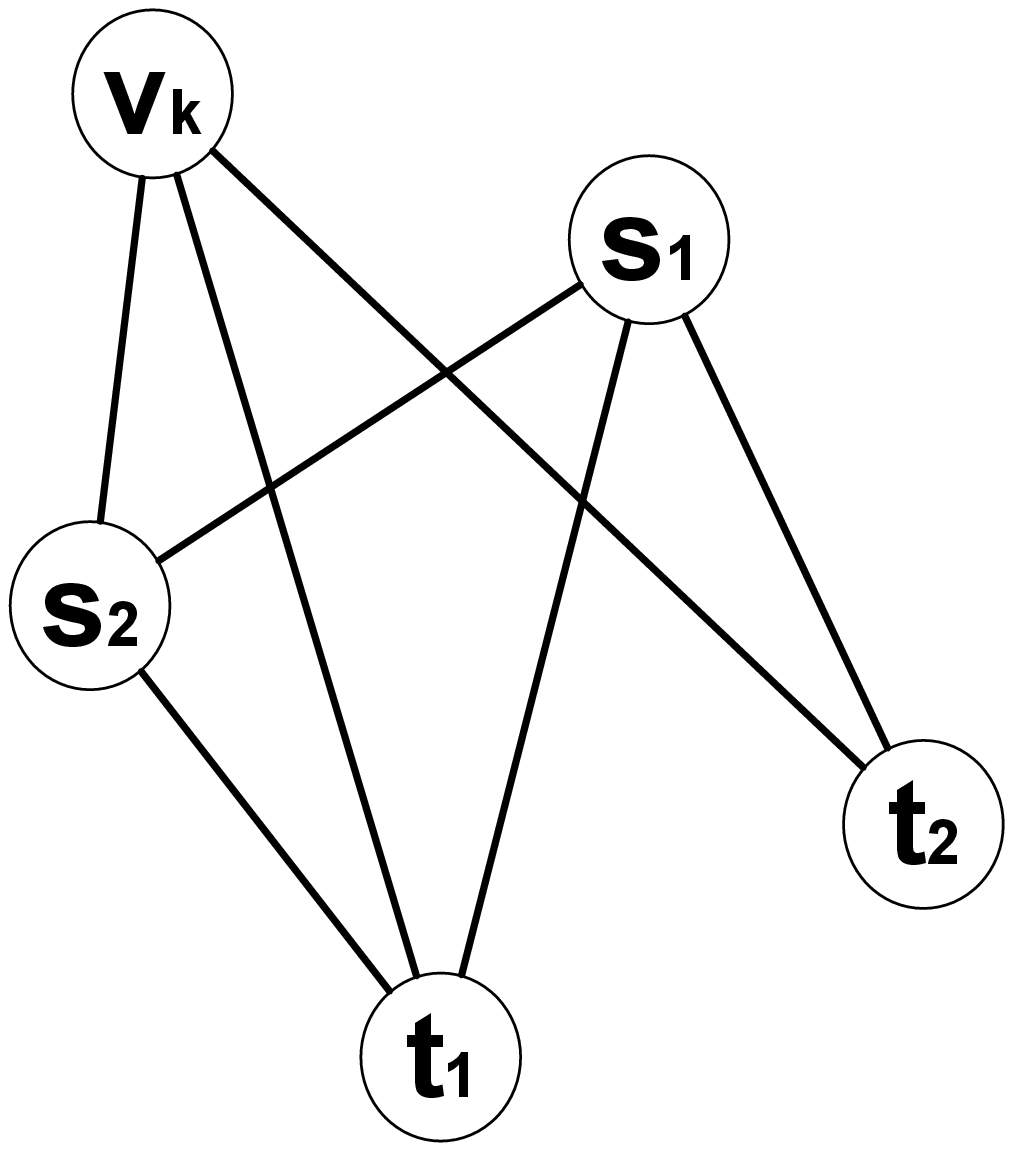}}
\subfigure[]{\includegraphics[width=1.43cm]{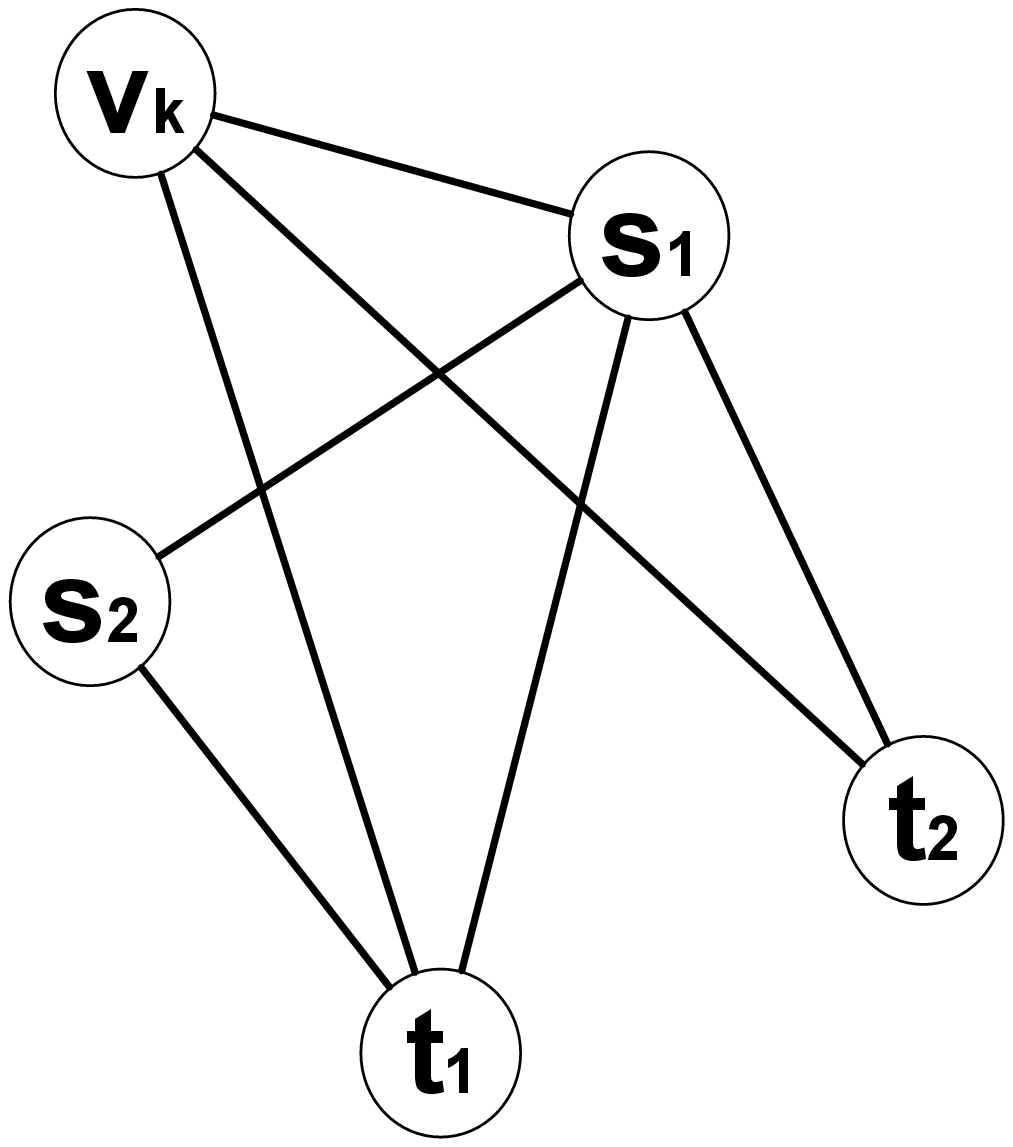}}
\subfigure[]{\includegraphics[width=1.43cm]{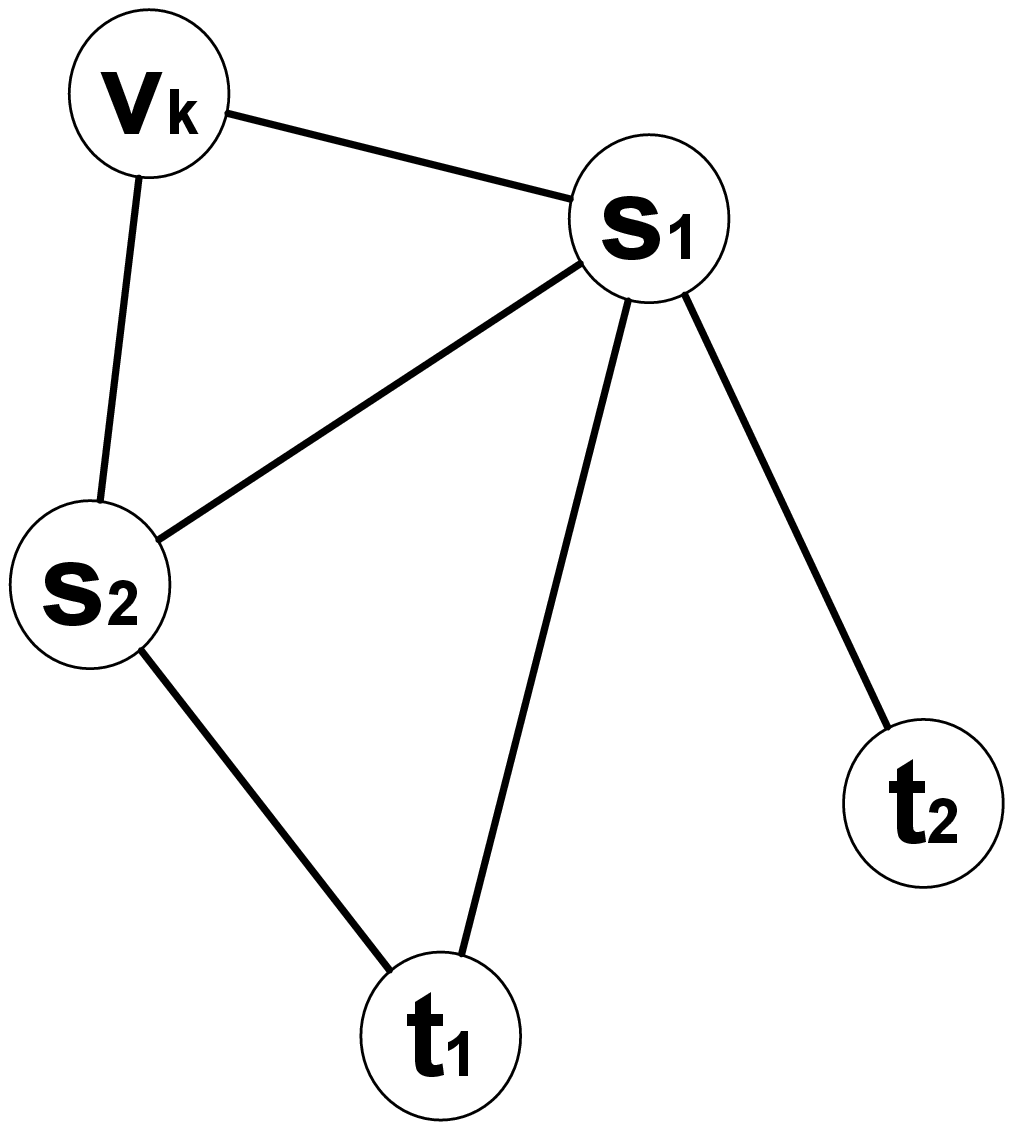}}
\subfigure[]{\includegraphics[width=1.43cm]{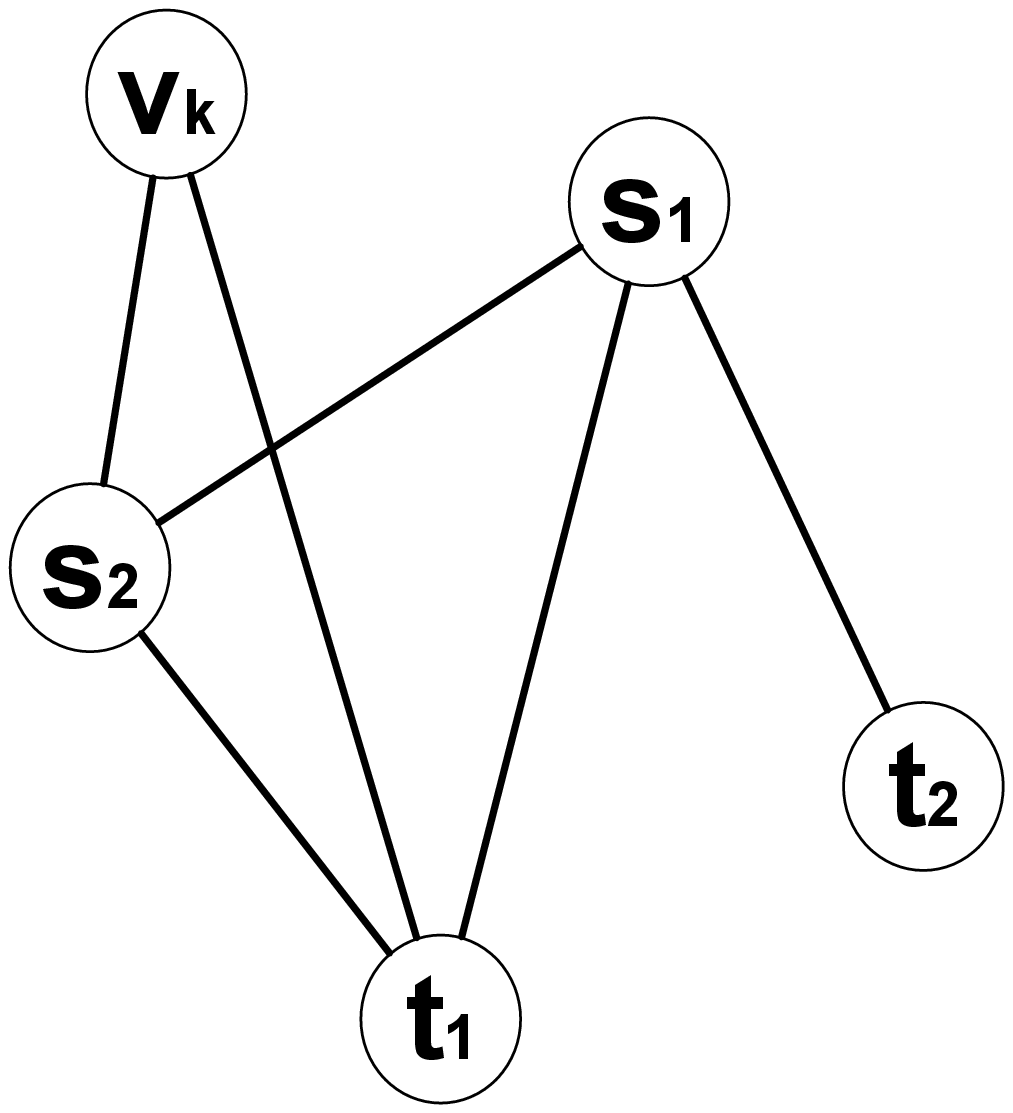}}
\subfigure[]{\includegraphics[width=1.43cm]{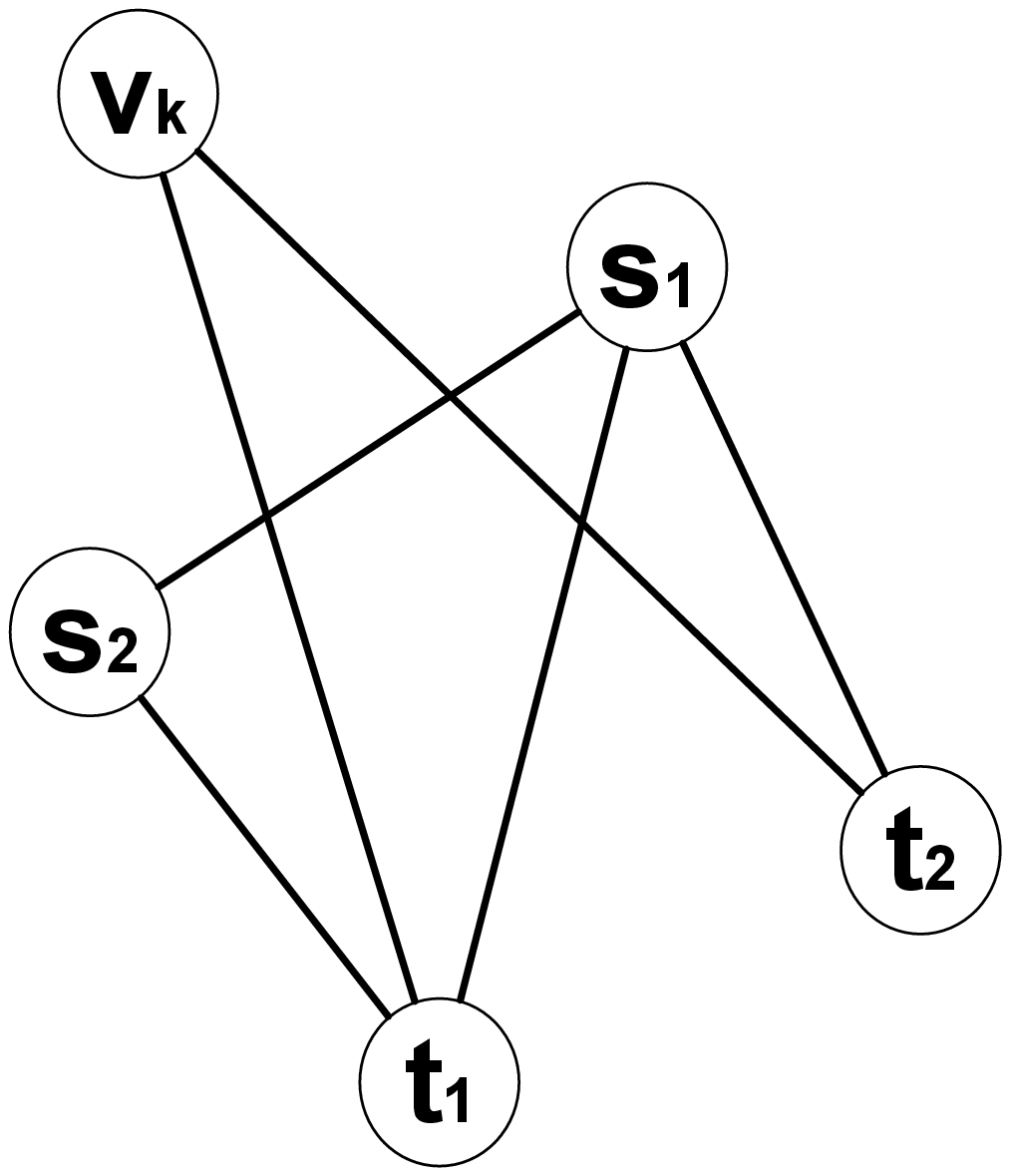}}
\subfigure[]{\includegraphics[width=1.43cm]{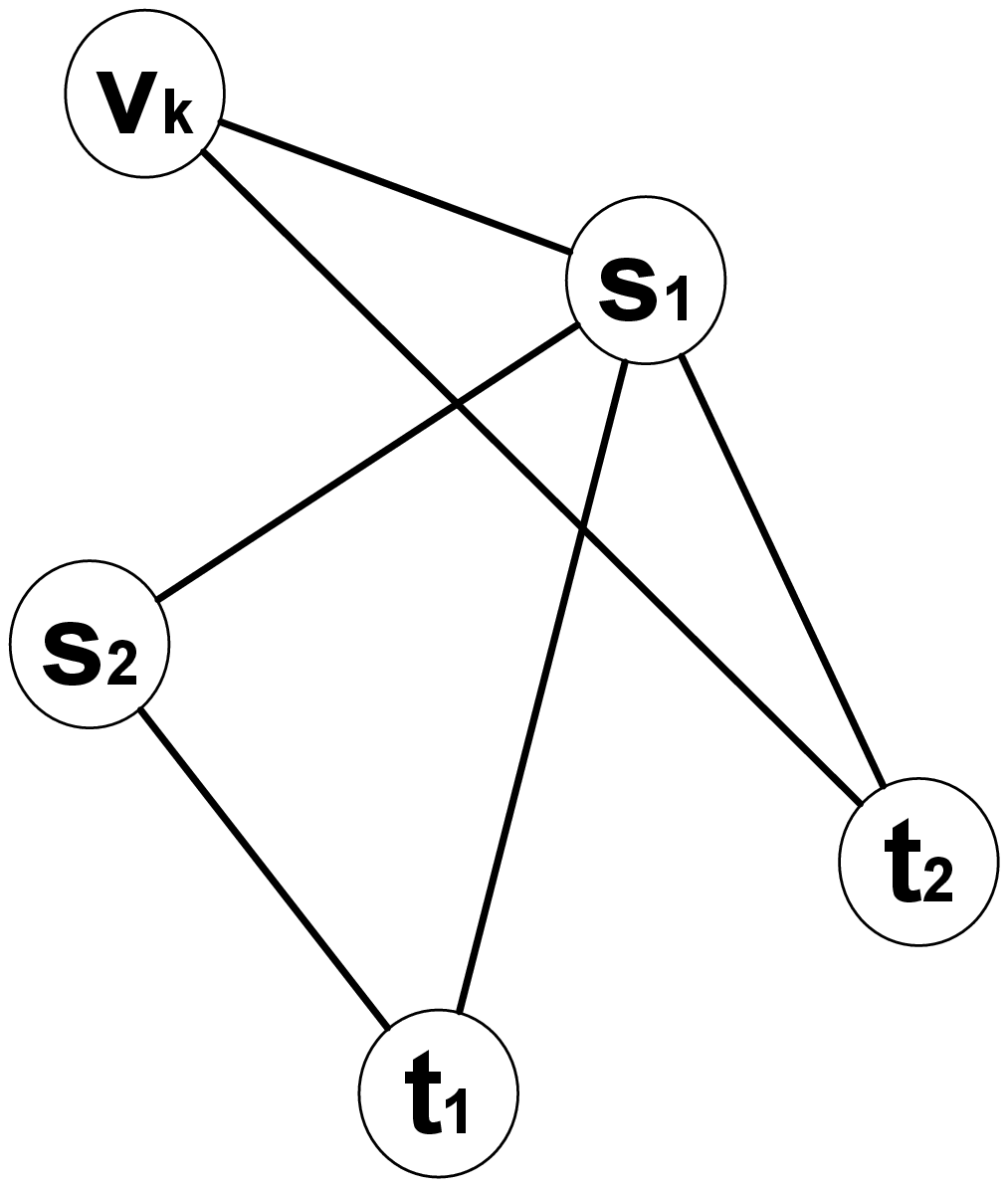}}
\caption{Graphs abiding by situation b) or c) of Proposition \ref{progeneve}, where (a)-(d) and (e)-(k) are designed, respectively, from the topology structures (a) and (b) of Fig. \ref{Fcdtopo4}.}
\label{Fivnodesb}
\end{center}
\end{figure}
\begin{figure}[H]
\begin{center}
\subfigure[]{\includegraphics[width=1.56cm]{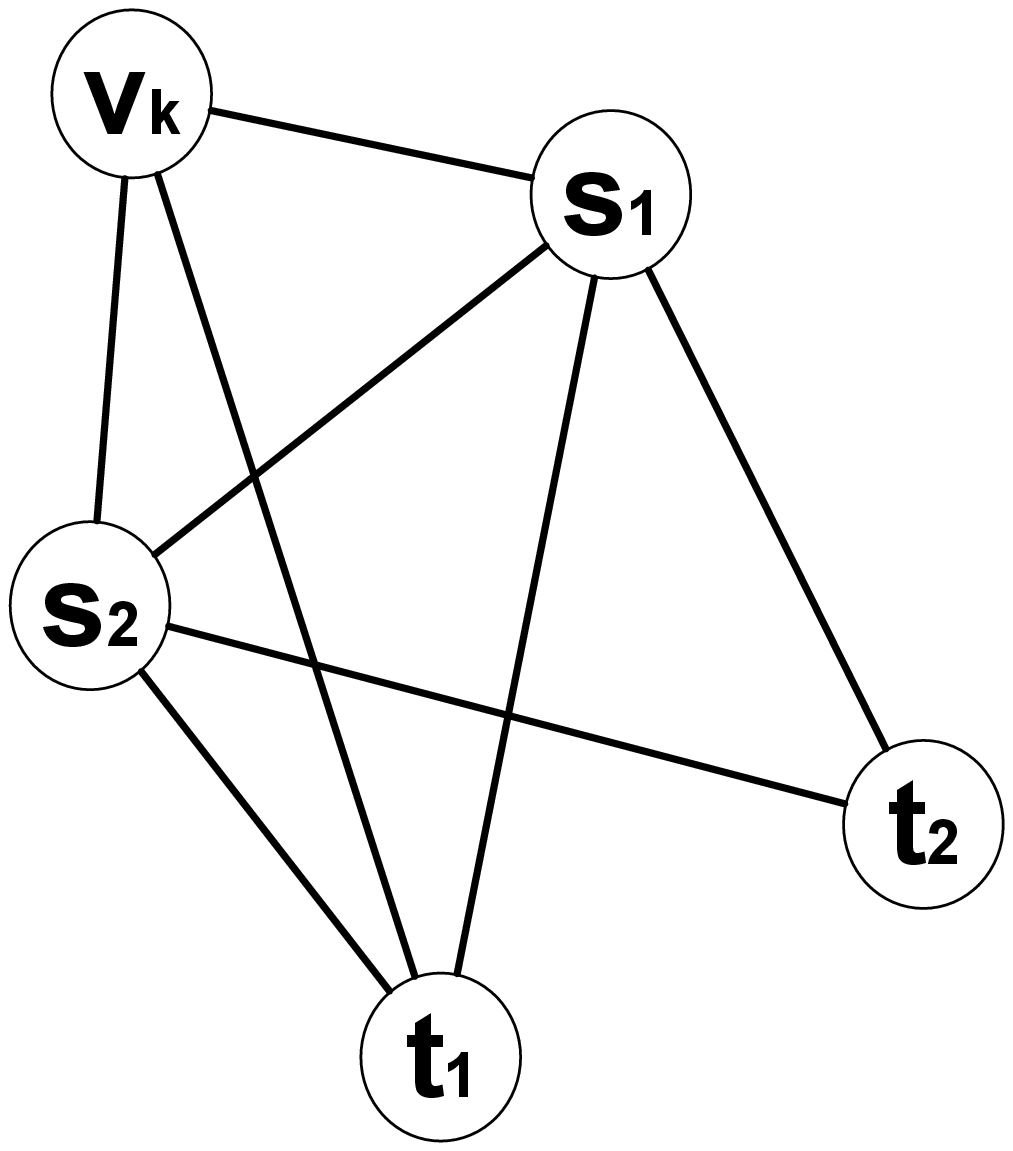}}
\subfigure[]{\includegraphics[width=1.56cm]{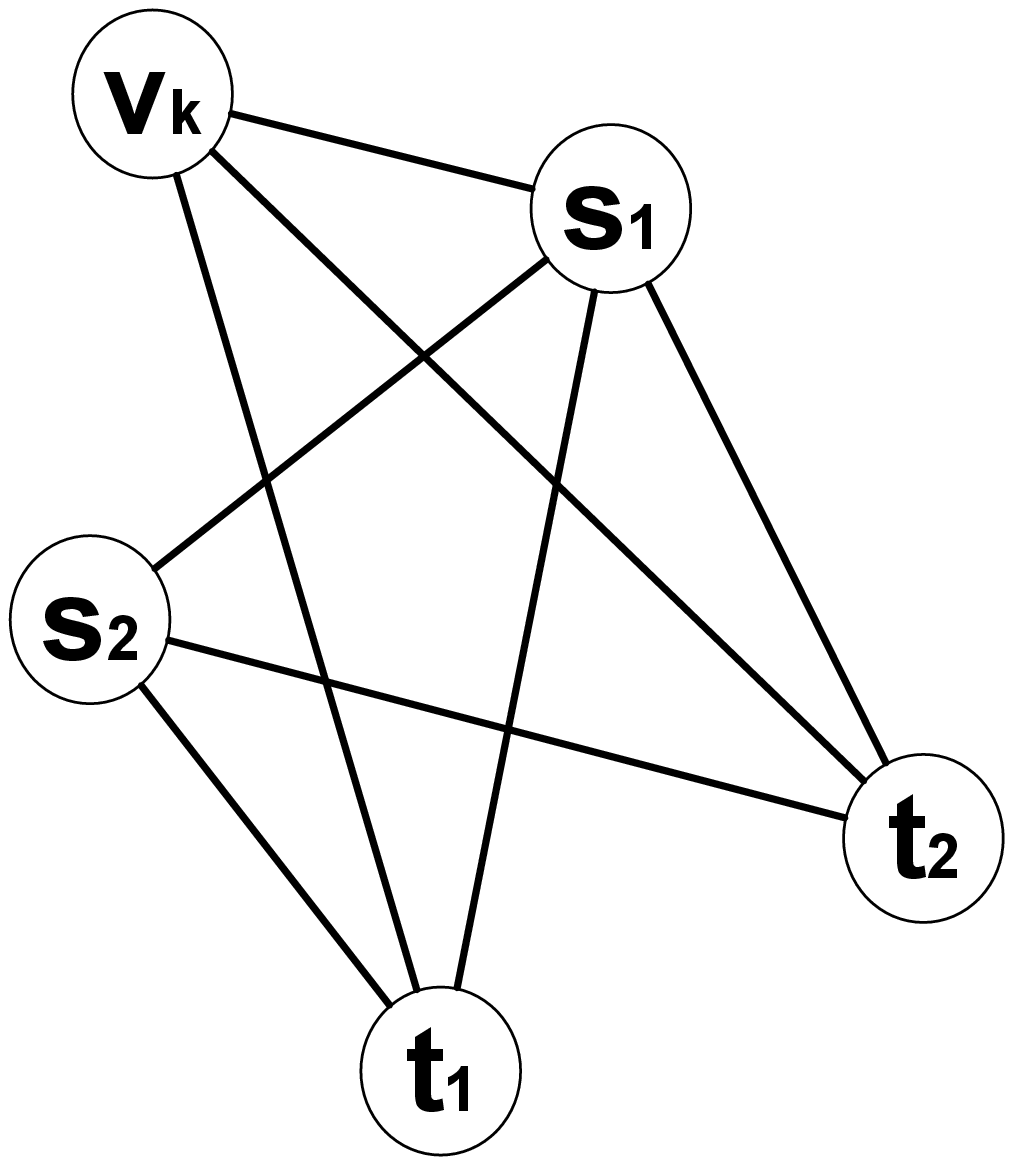}}
\subfigure[]{\includegraphics[width=1.56cm]{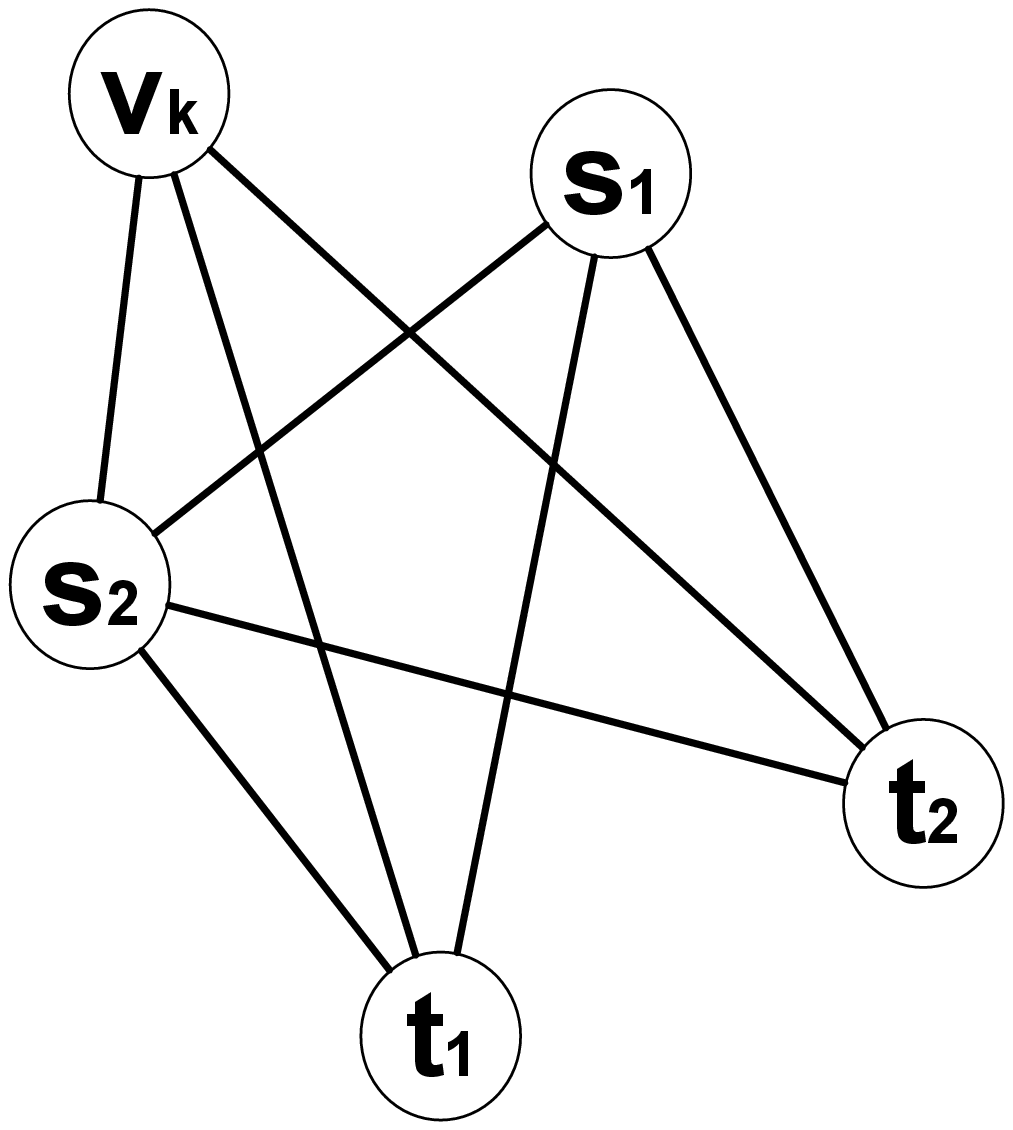}}
\subfigure[]{\includegraphics[width=1.56cm]{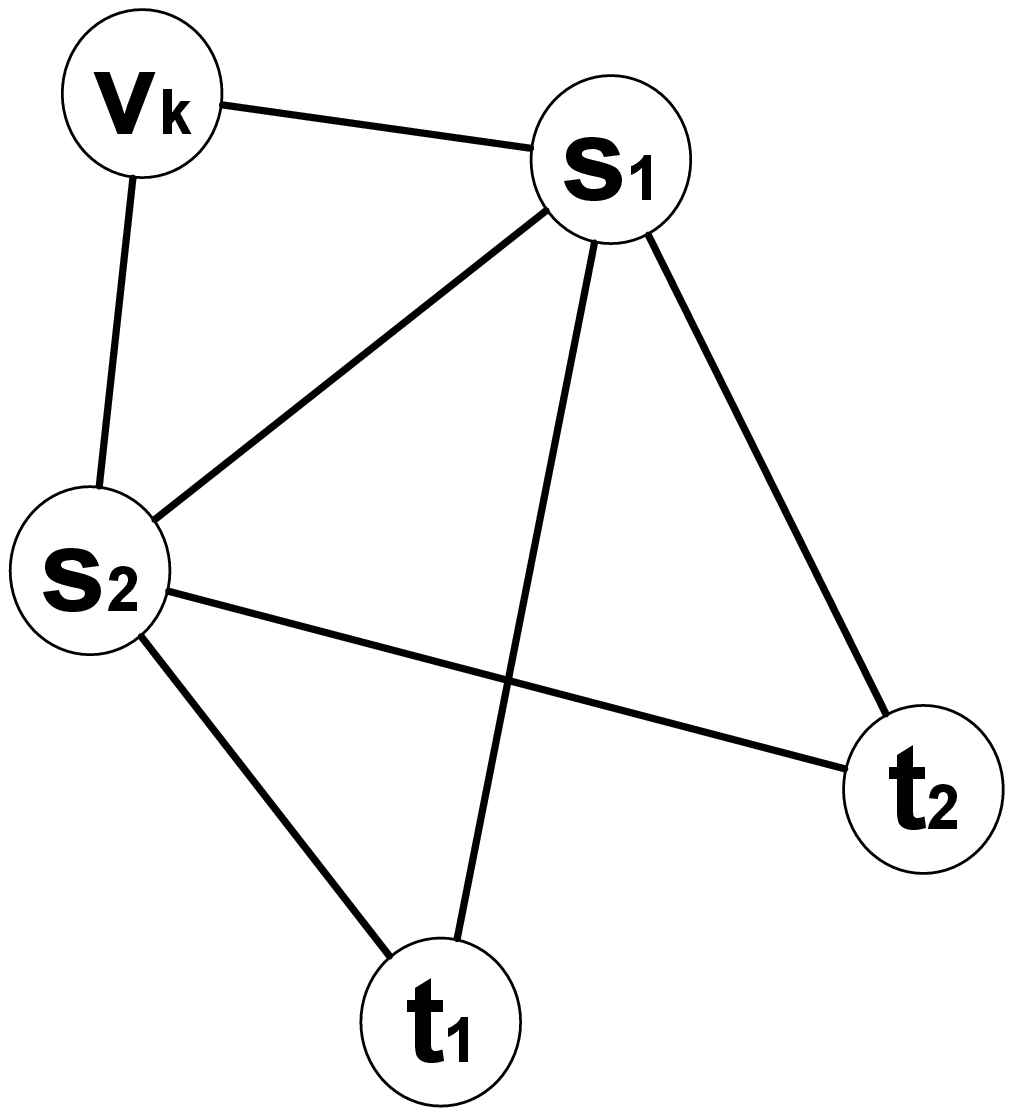}}
\subfigure[]{\includegraphics[width=1.56cm]{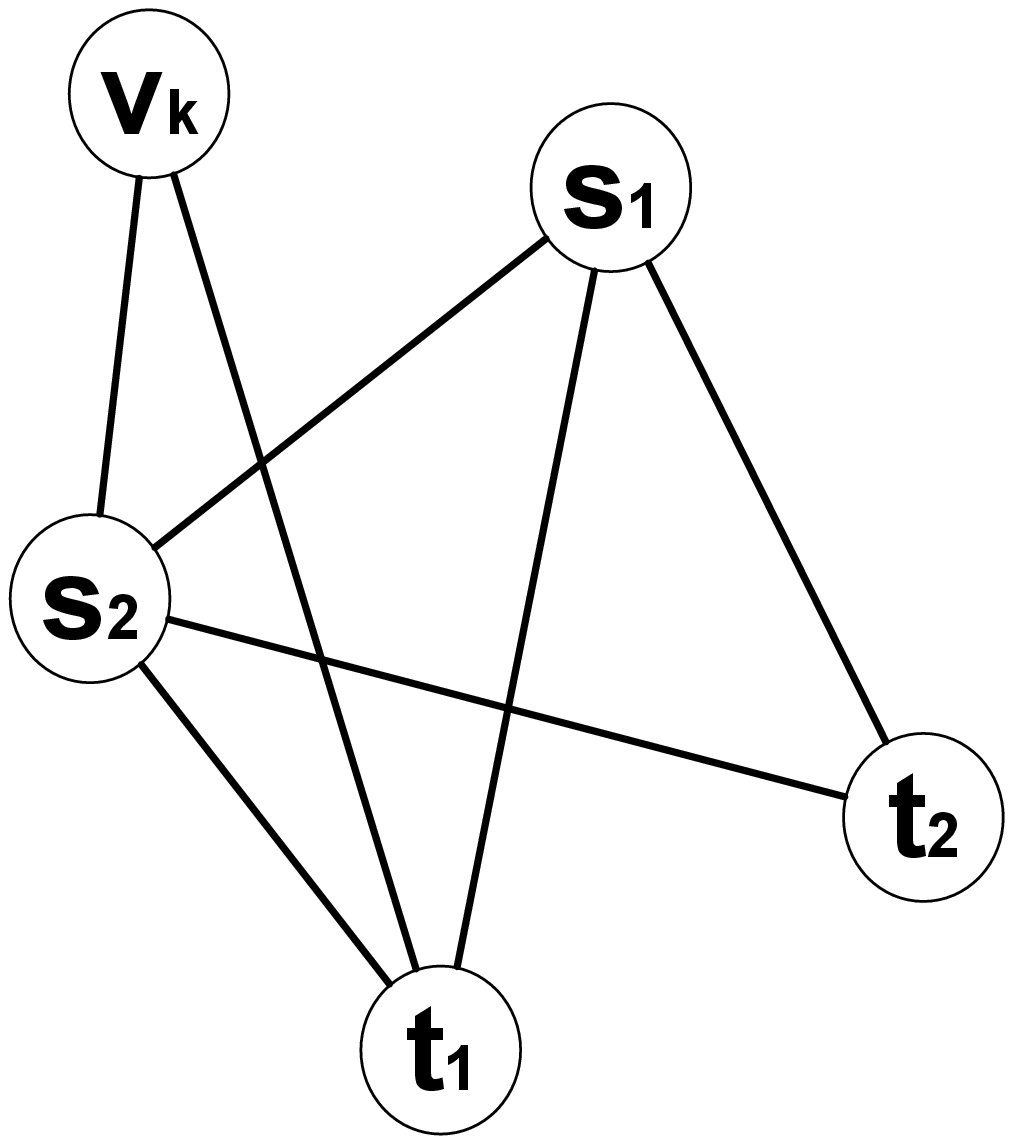}}
\subfigure[]{\includegraphics[width=1.56cm]{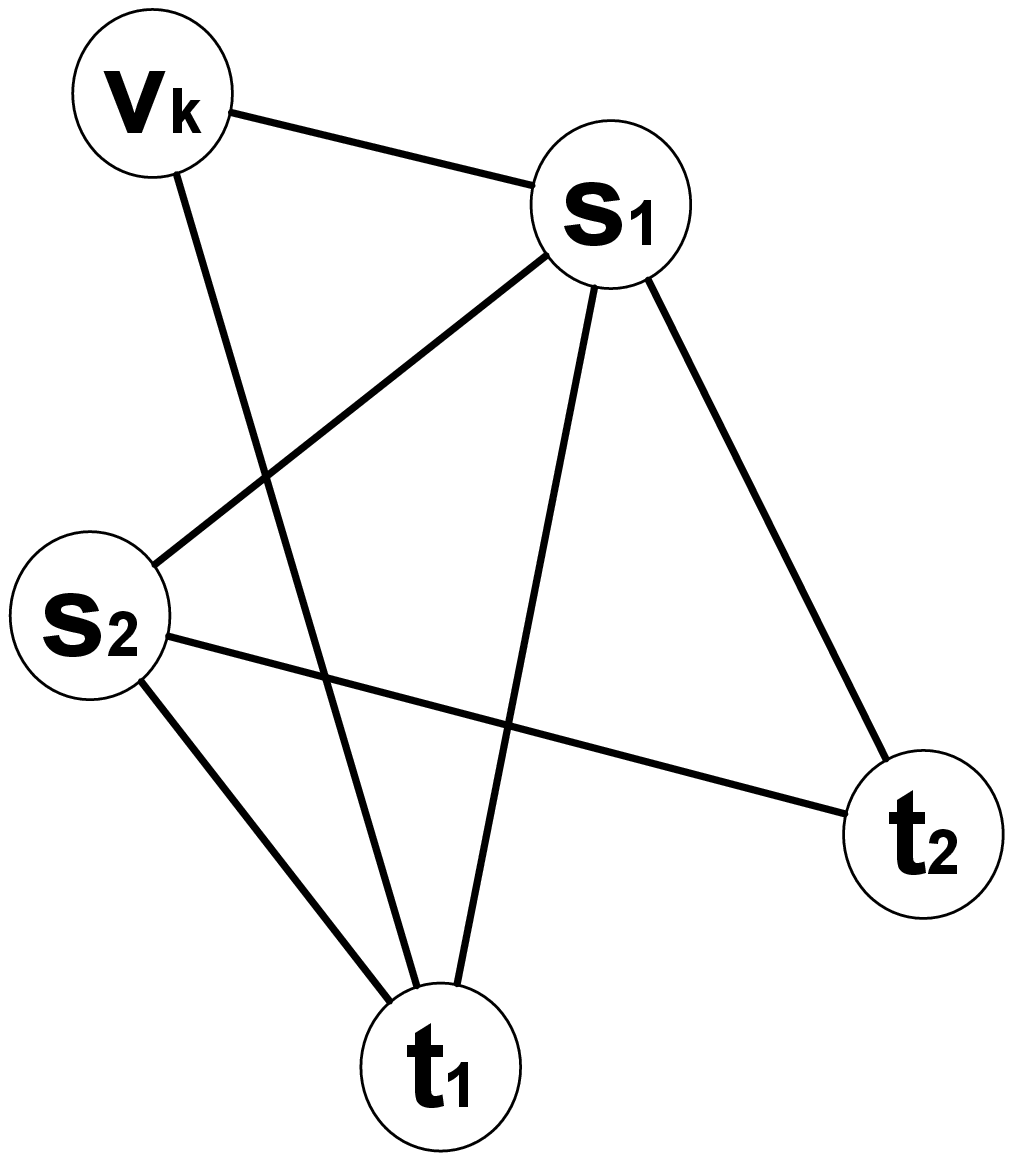}}
\subfigure[]{\includegraphics[width=1.56cm]{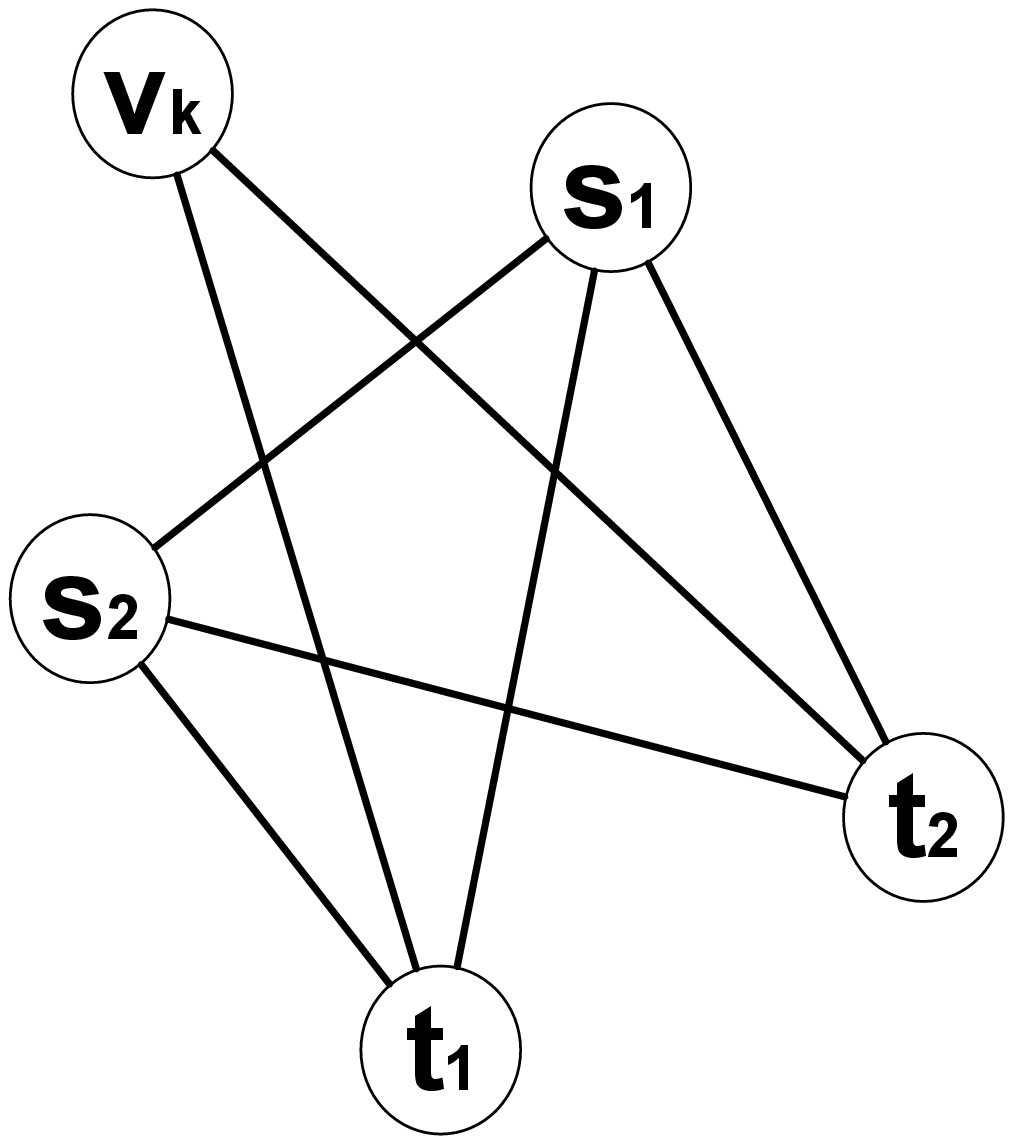}}
\subfigure[]{\includegraphics[width=1.56cm]{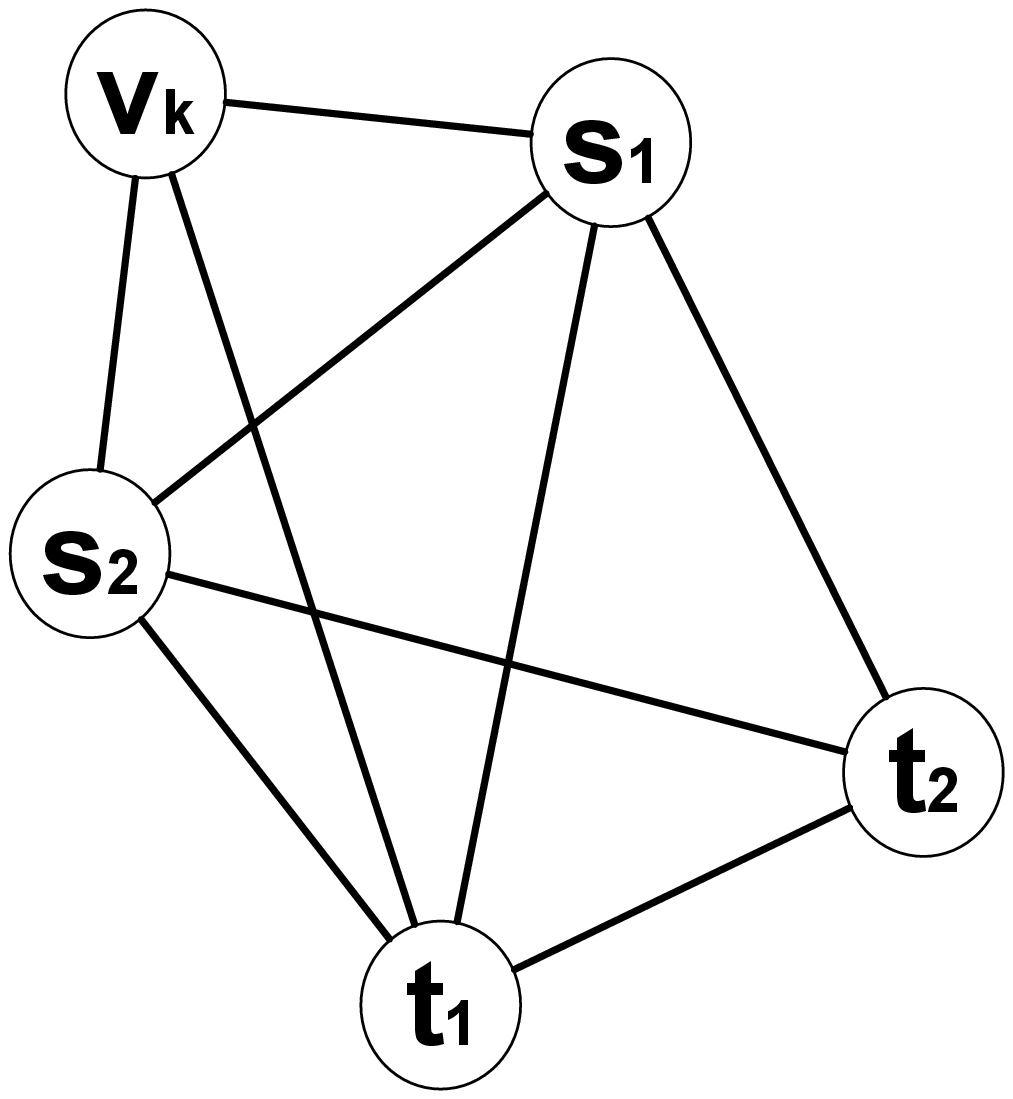}}
\subfigure[]{\includegraphics[width=1.56cm]{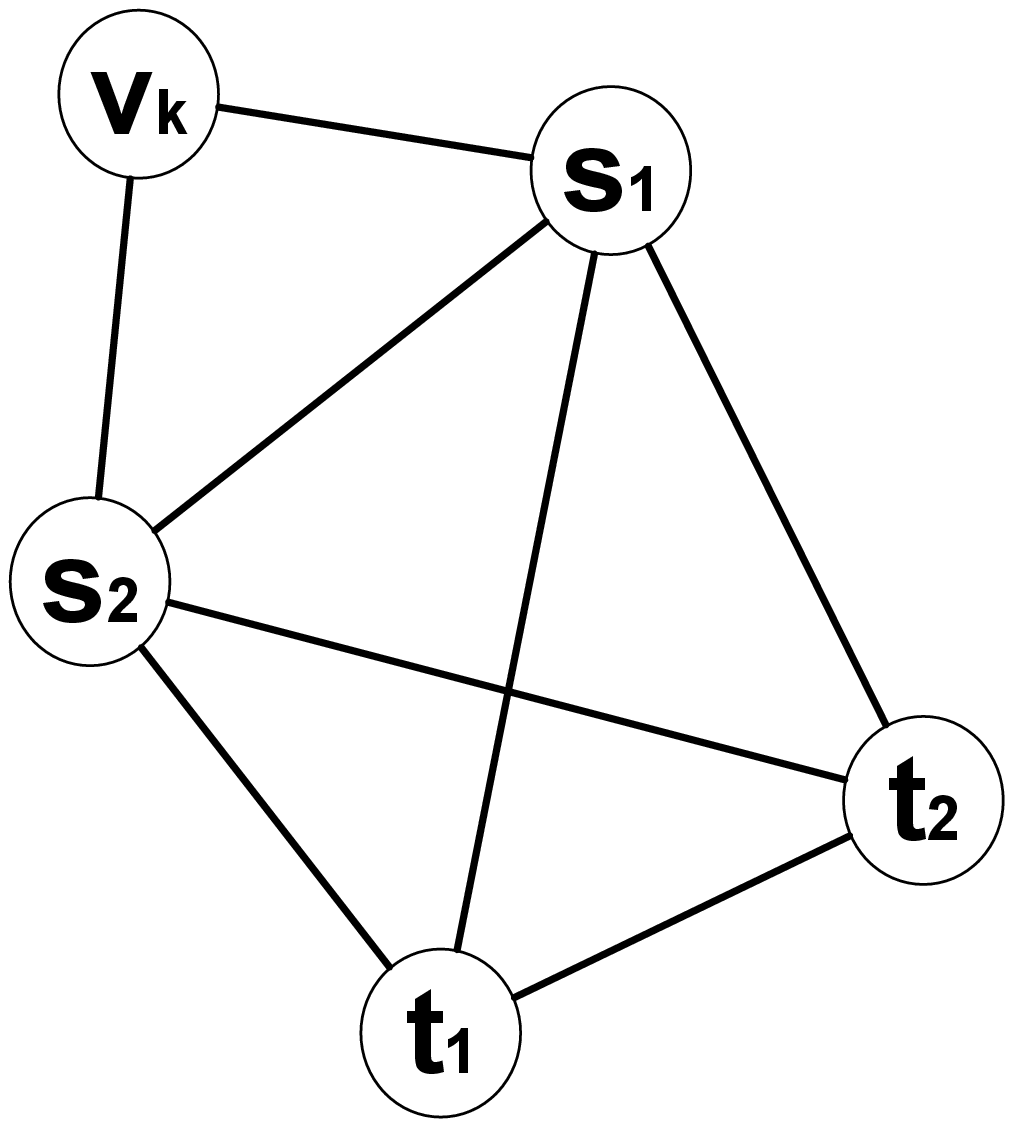}}
\caption{Graphs abiding by situation b) or c) of Proposition \ref{progeneve}, where (a)-(g) and (h)(i) are designed, respectively, from the topology structures (c) and (d) of Fig. \ref{Fcdtopo4}.}
\label{Fivnodesc}
\end{center}
\end{figure}
\begin{figure}[H]
\begin{center}
\subfigure[]{\includegraphics[width=1.56cm]{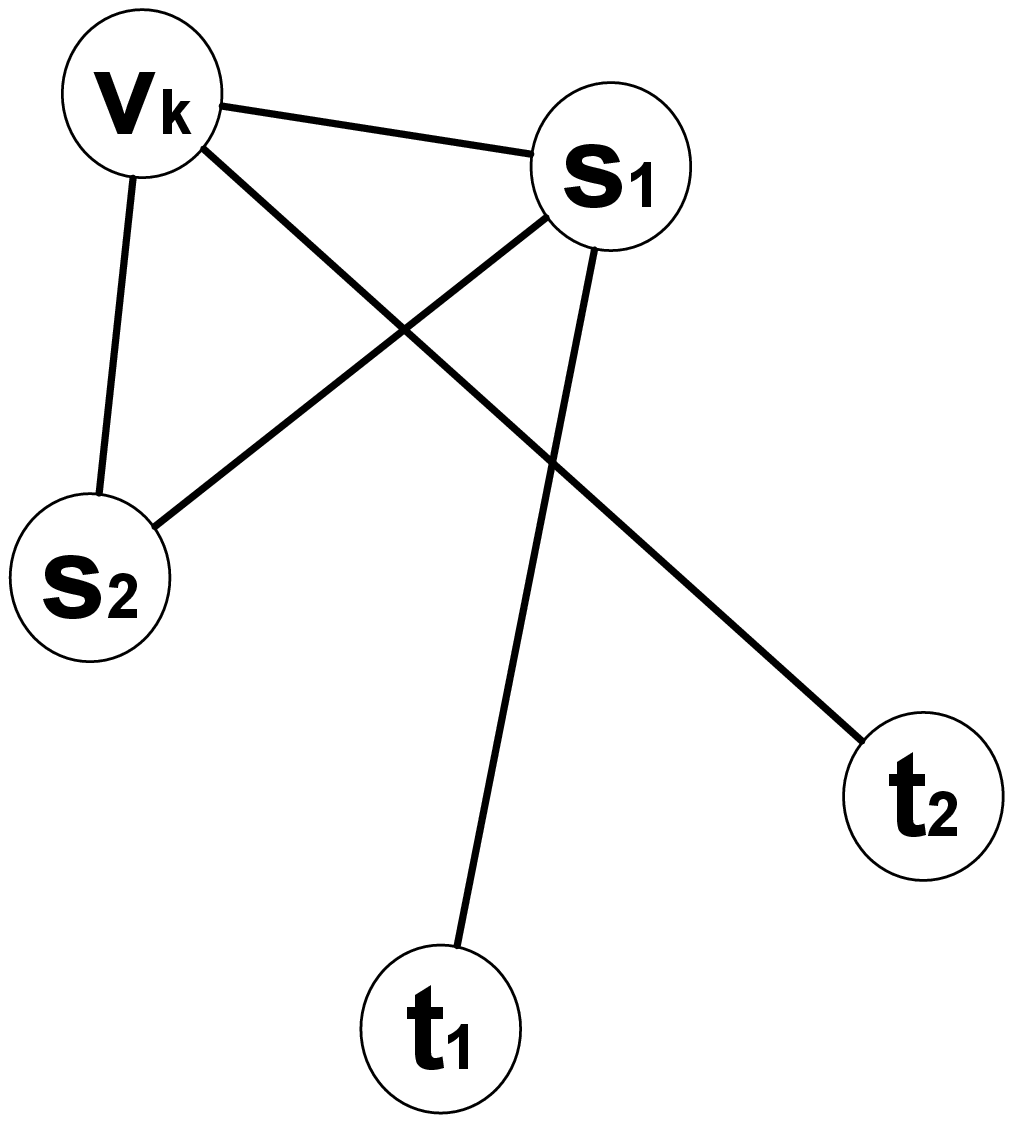}}
\subfigure[]{\includegraphics[width=1.56cm]{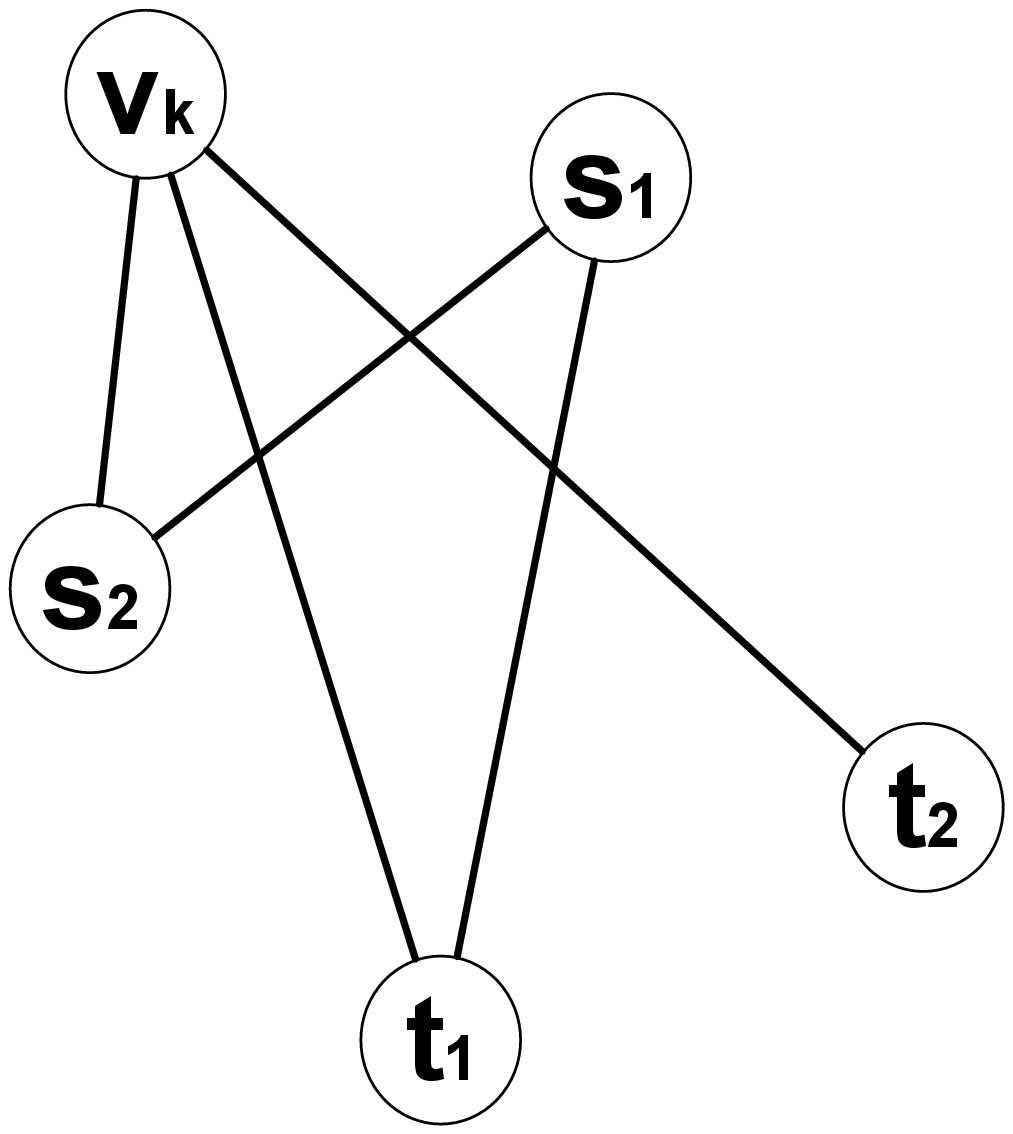}}
\subfigure[]{\includegraphics[width=1.56cm]{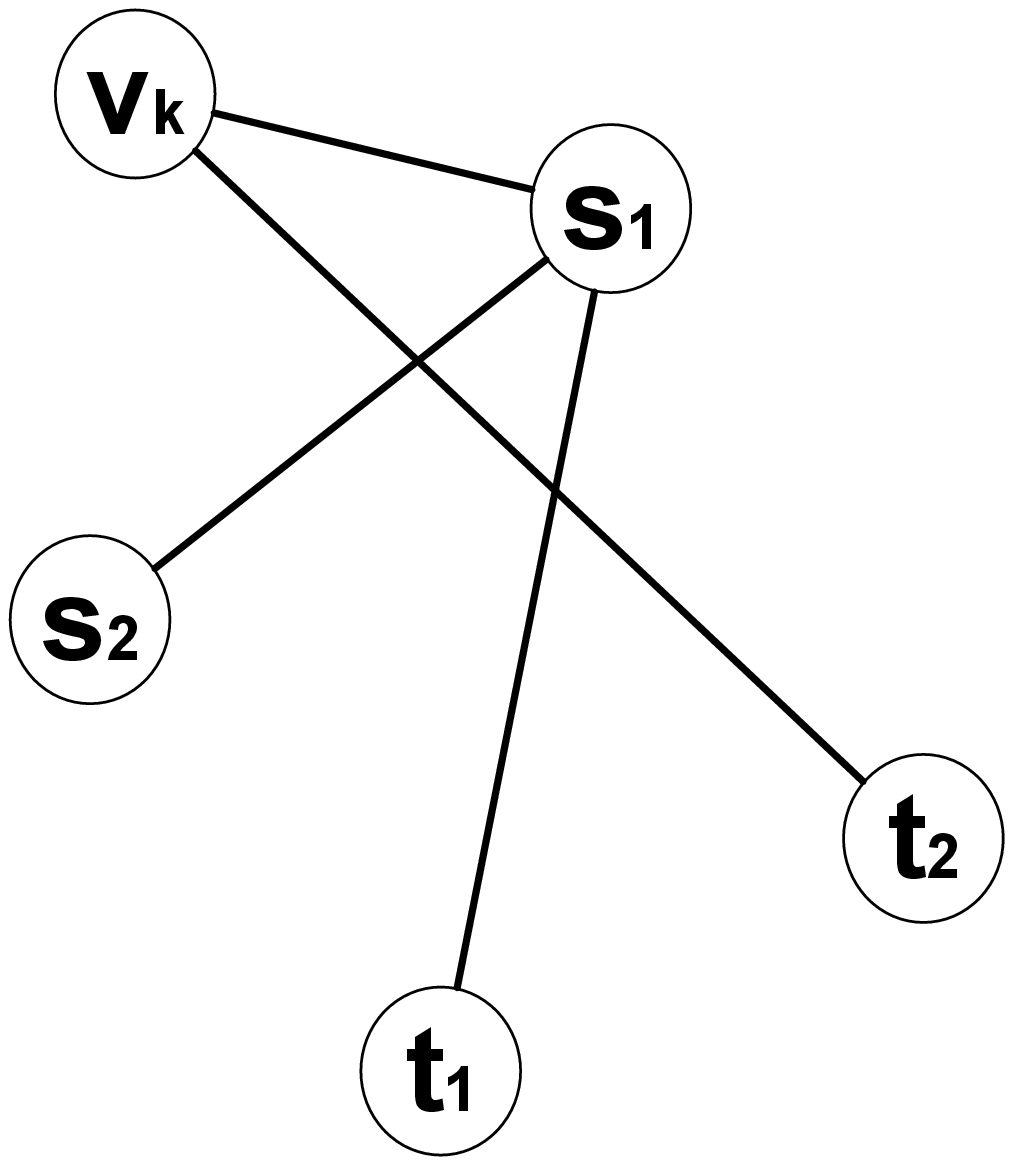}}
\subfigure[]{\includegraphics[width=1.56cm]{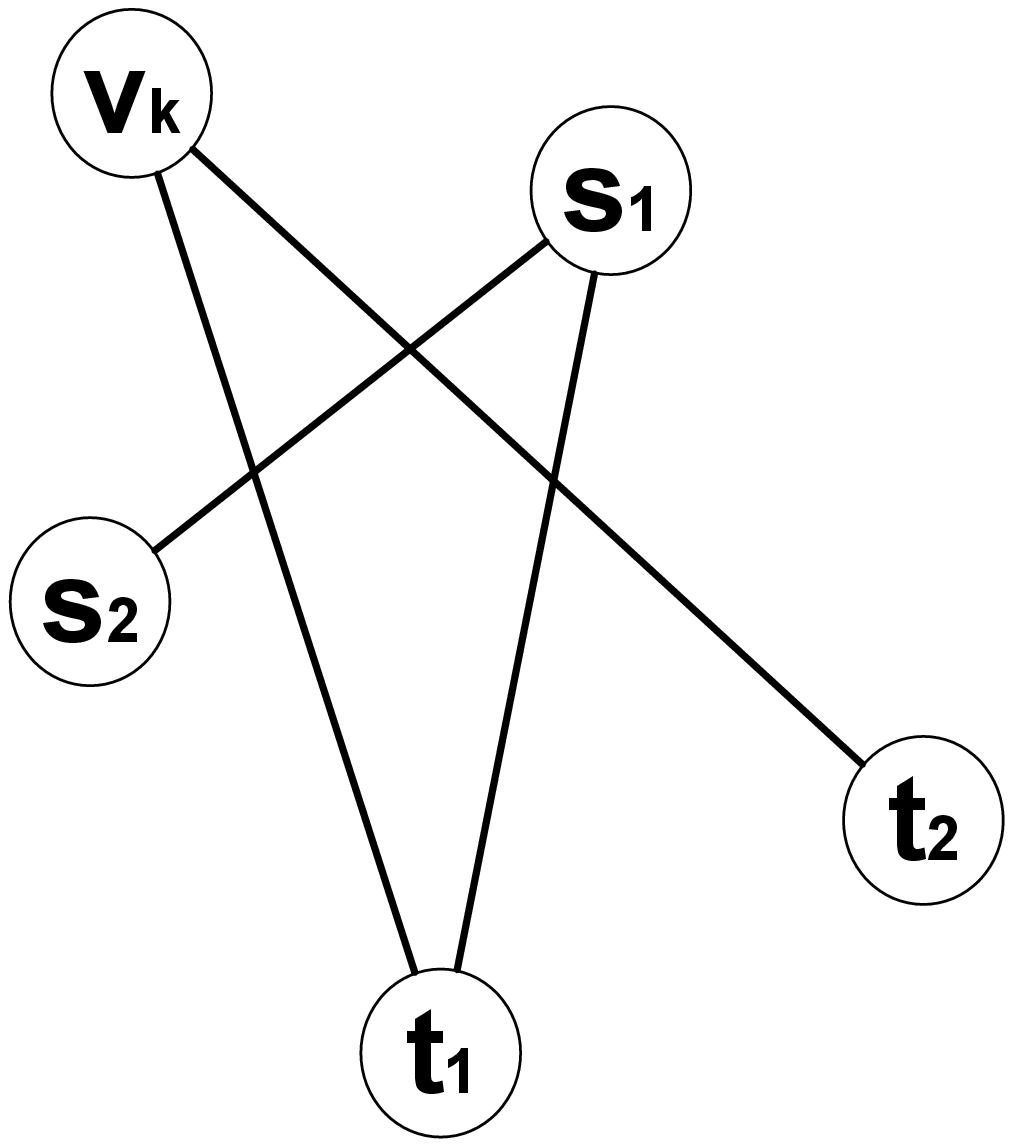}}
\subfigure[]{\includegraphics[width=1.56cm]{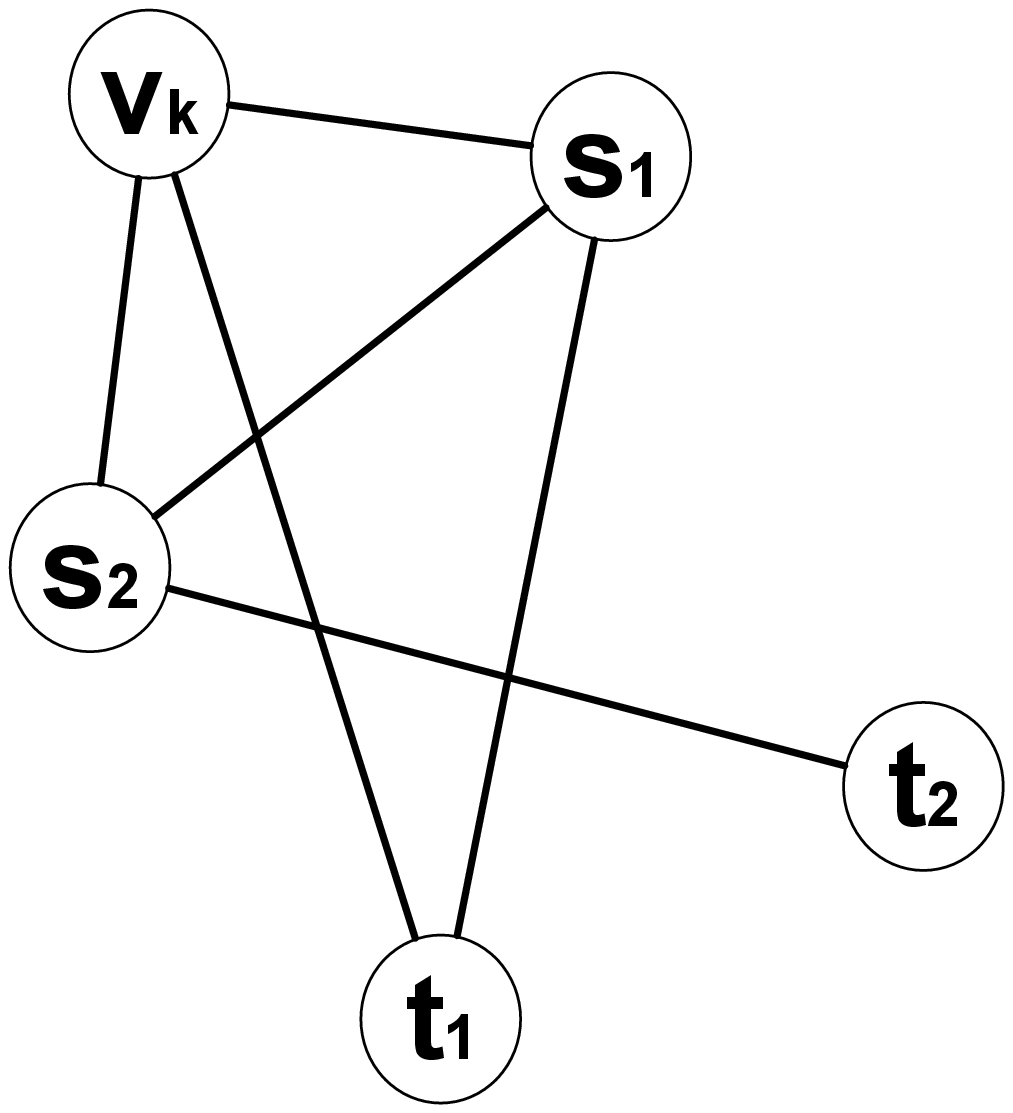}}
\subfigure[]{\includegraphics[width=1.56cm]{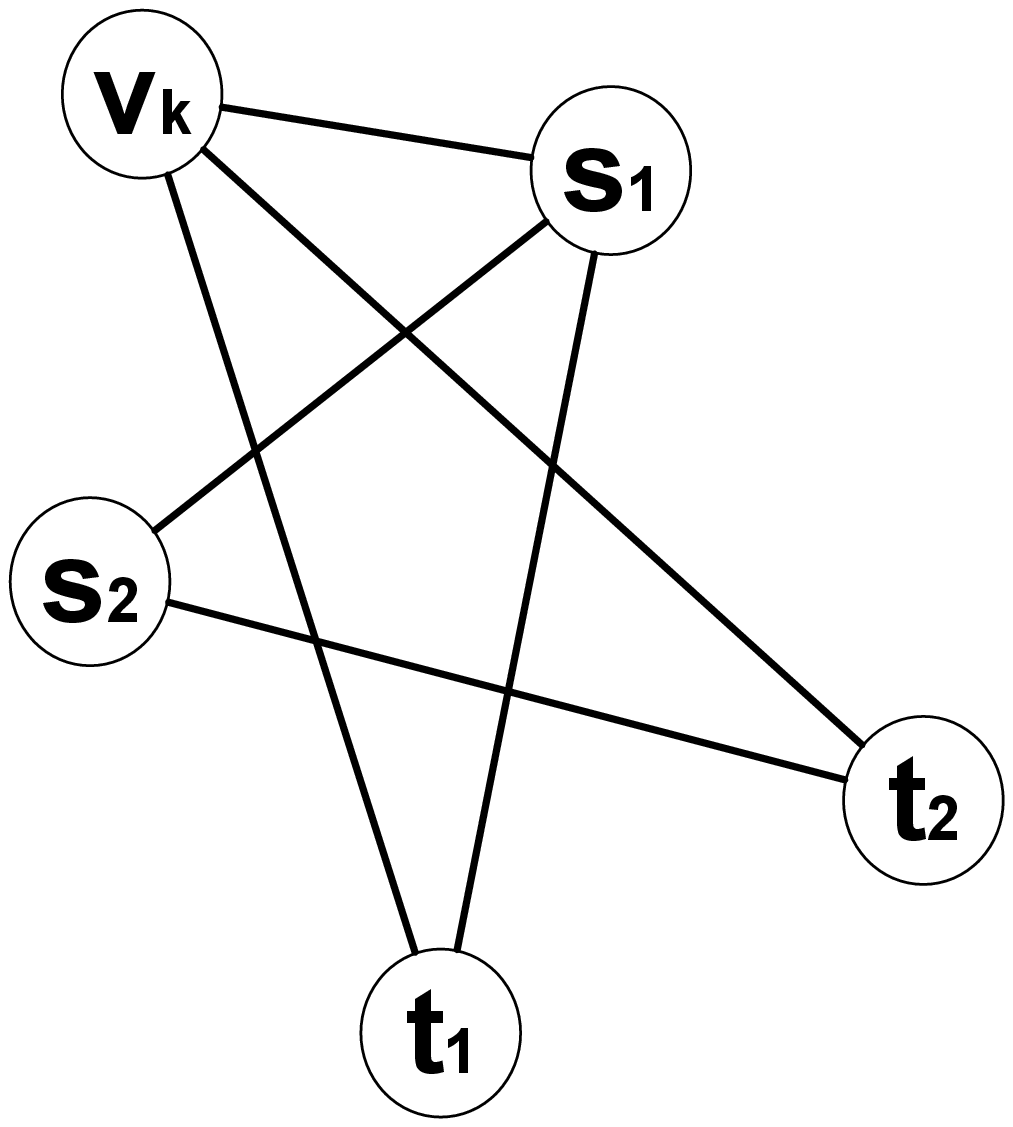}}
\subfigure[]{\includegraphics[width=1.56cm]{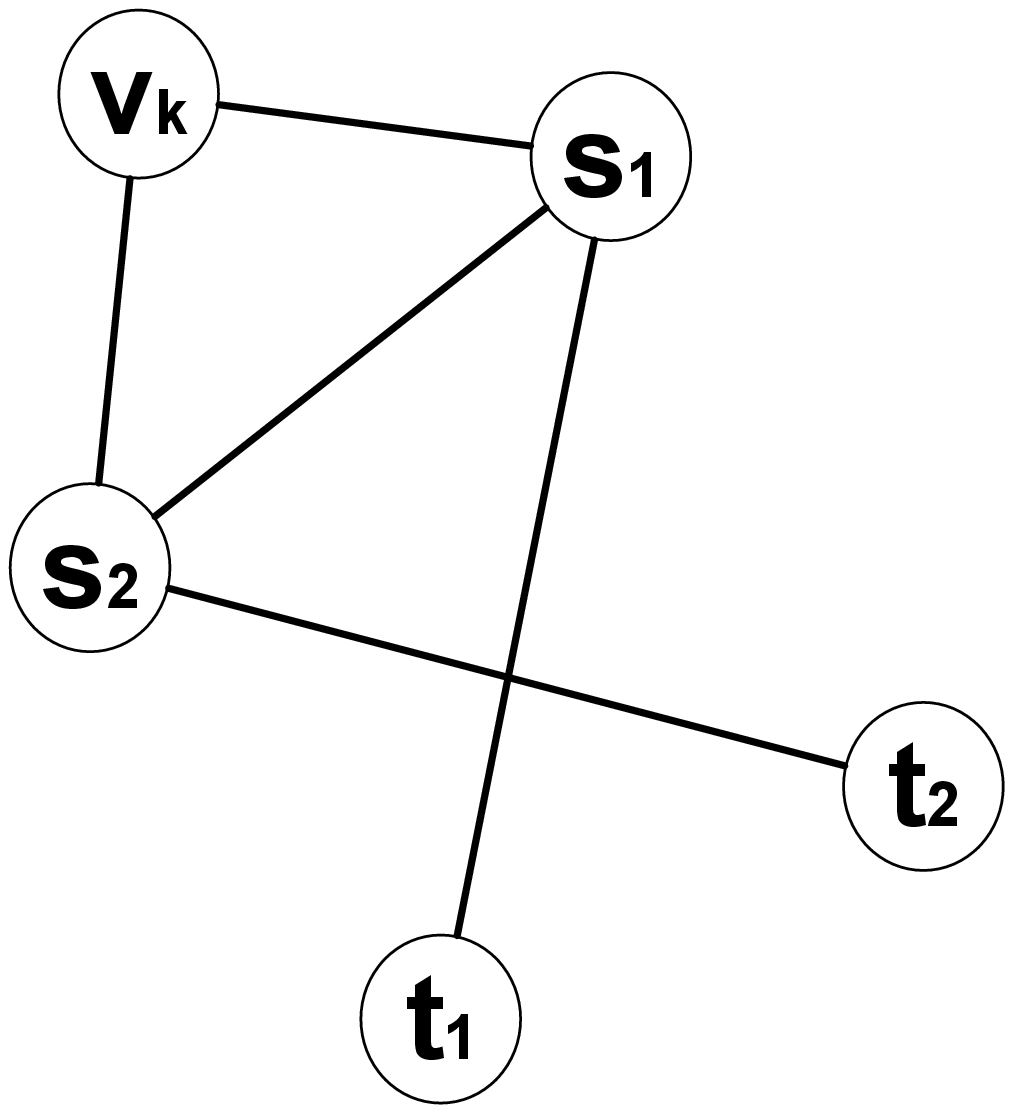}}
\subfigure[]{\includegraphics[width=1.56cm]{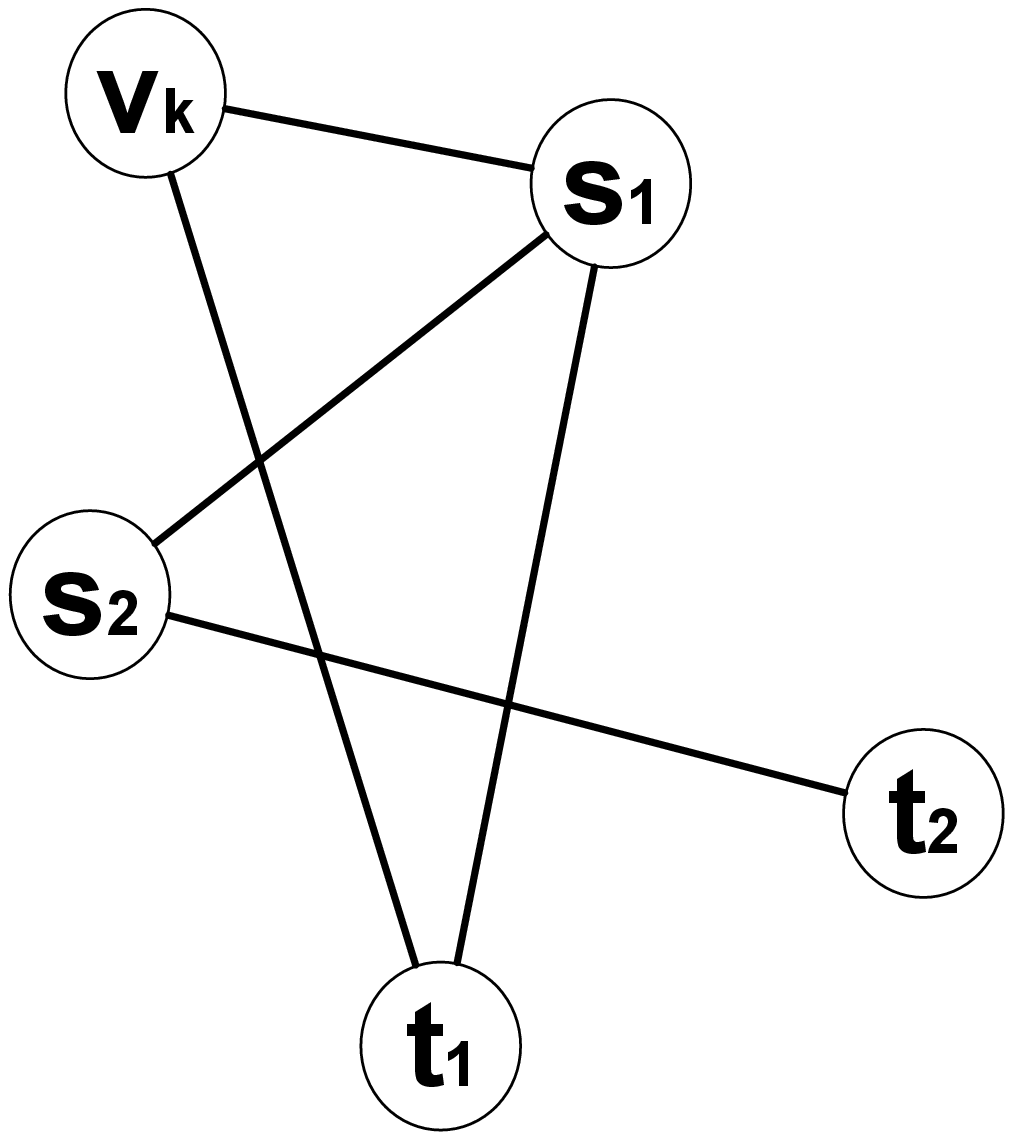}}
\subfigure[]{\includegraphics[width=1.56cm]{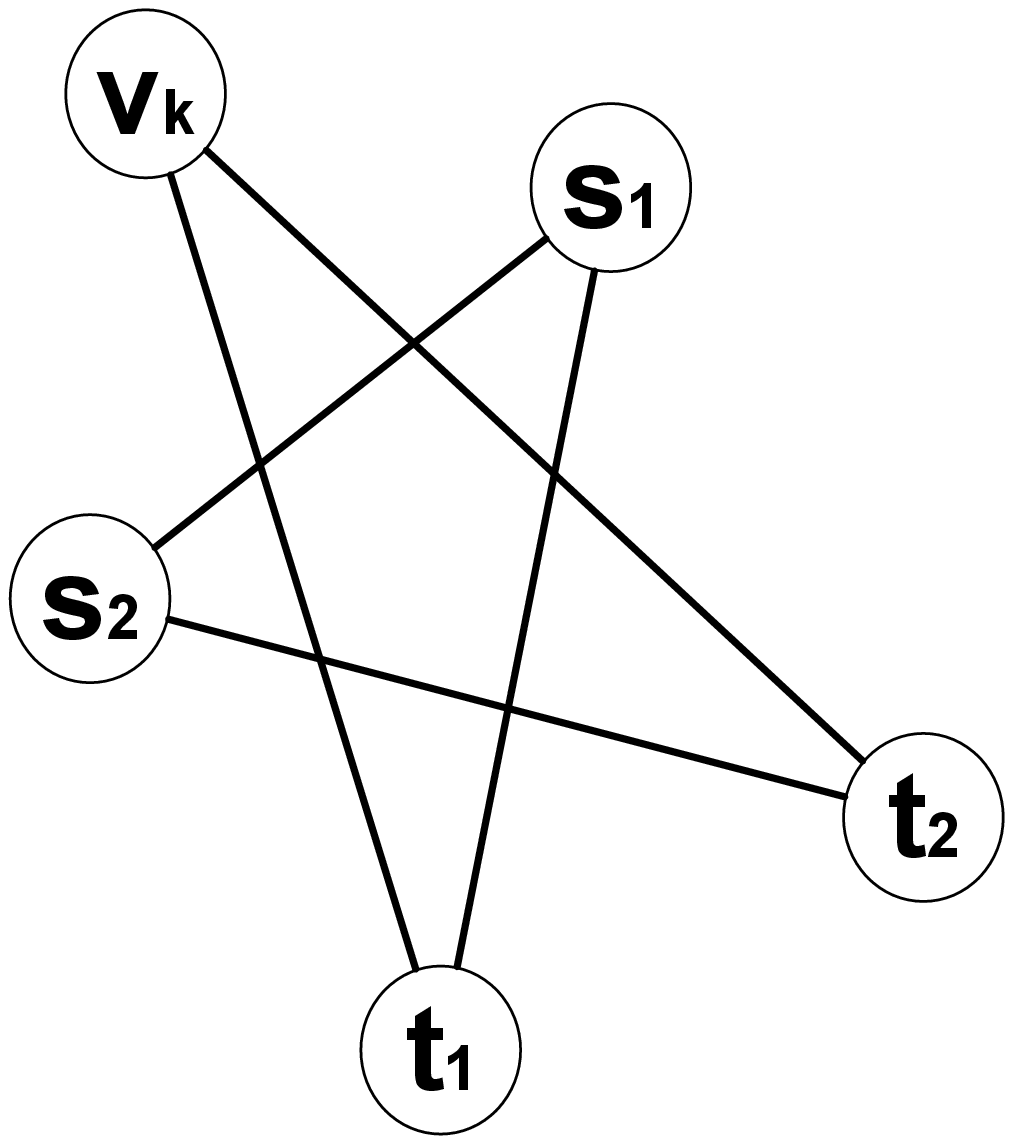}}
\subfigure[]{\includegraphics[width=1.56cm]{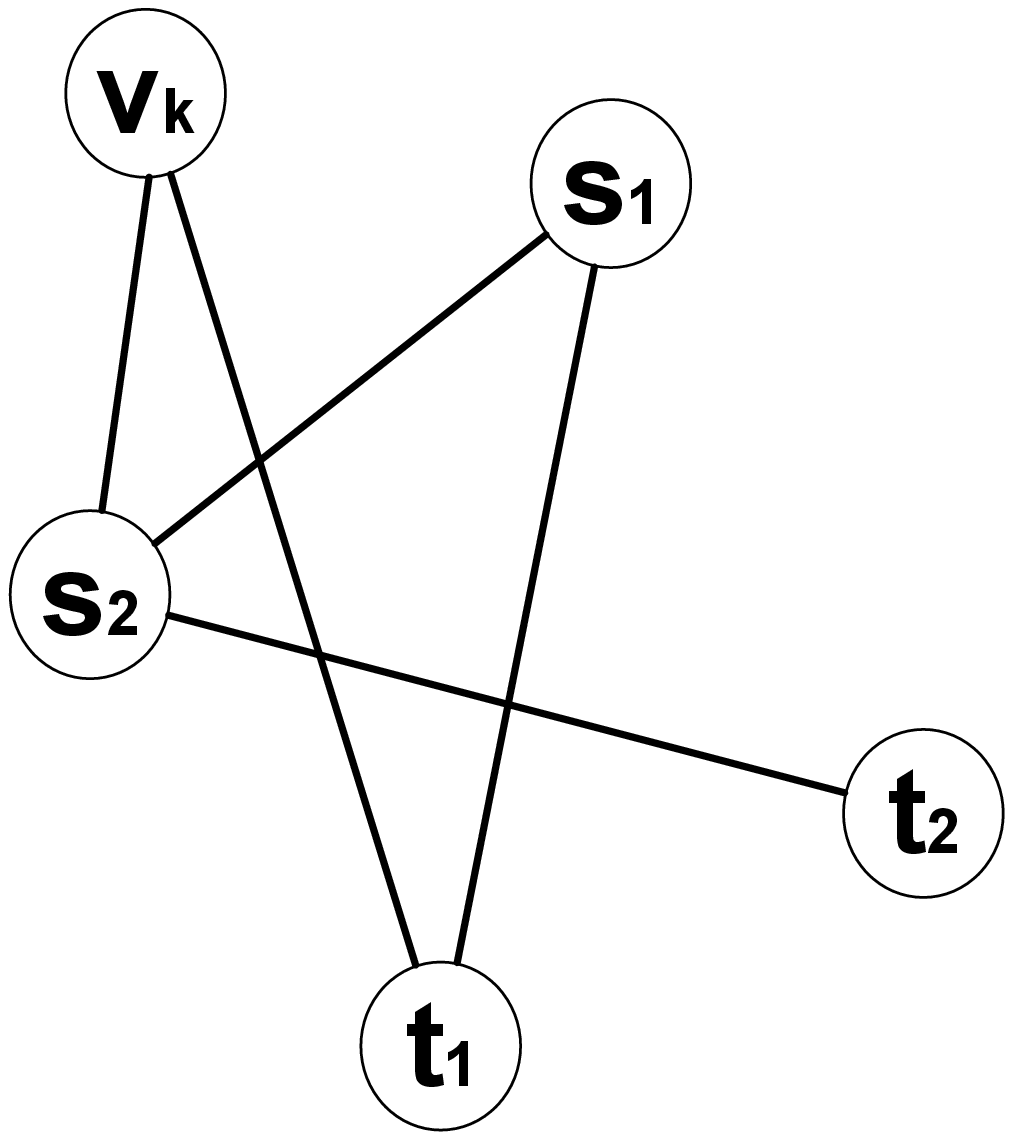}}
\caption{Graphs abiding by situation b) or c) of Proposition \ref{progeneve}, where (a)-(d) and (e)-(j) are designed, respectively, from  the topology structures (e) and (f) of Fig. \ref{Fcdtopo4}.}
\label{Fivnodesf}
\end{center}
\end{figure}
\begin{figure}[H]
\begin{center}
\subfigure[]{\includegraphics[width=1.56cm]{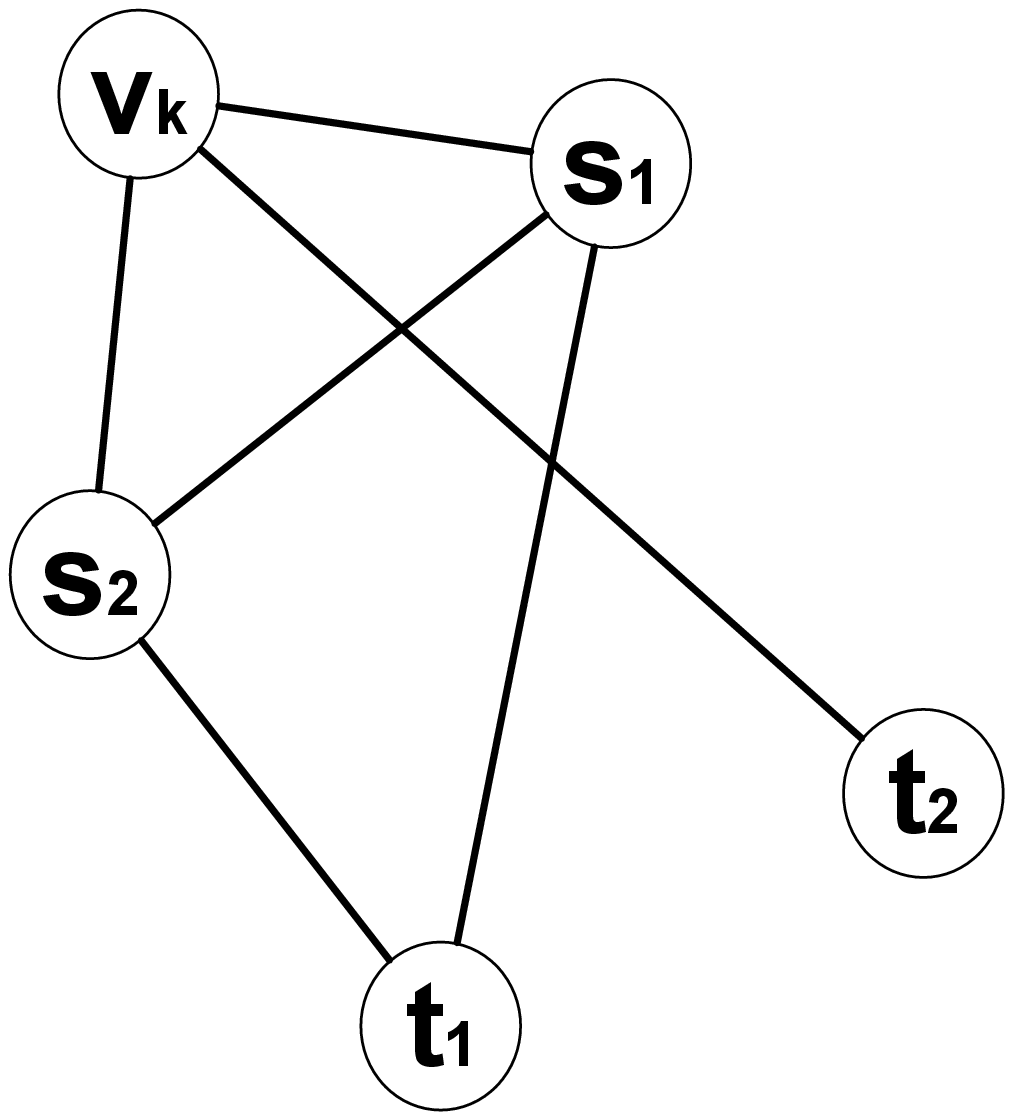}}
\subfigure[]{\includegraphics[width=1.56cm]{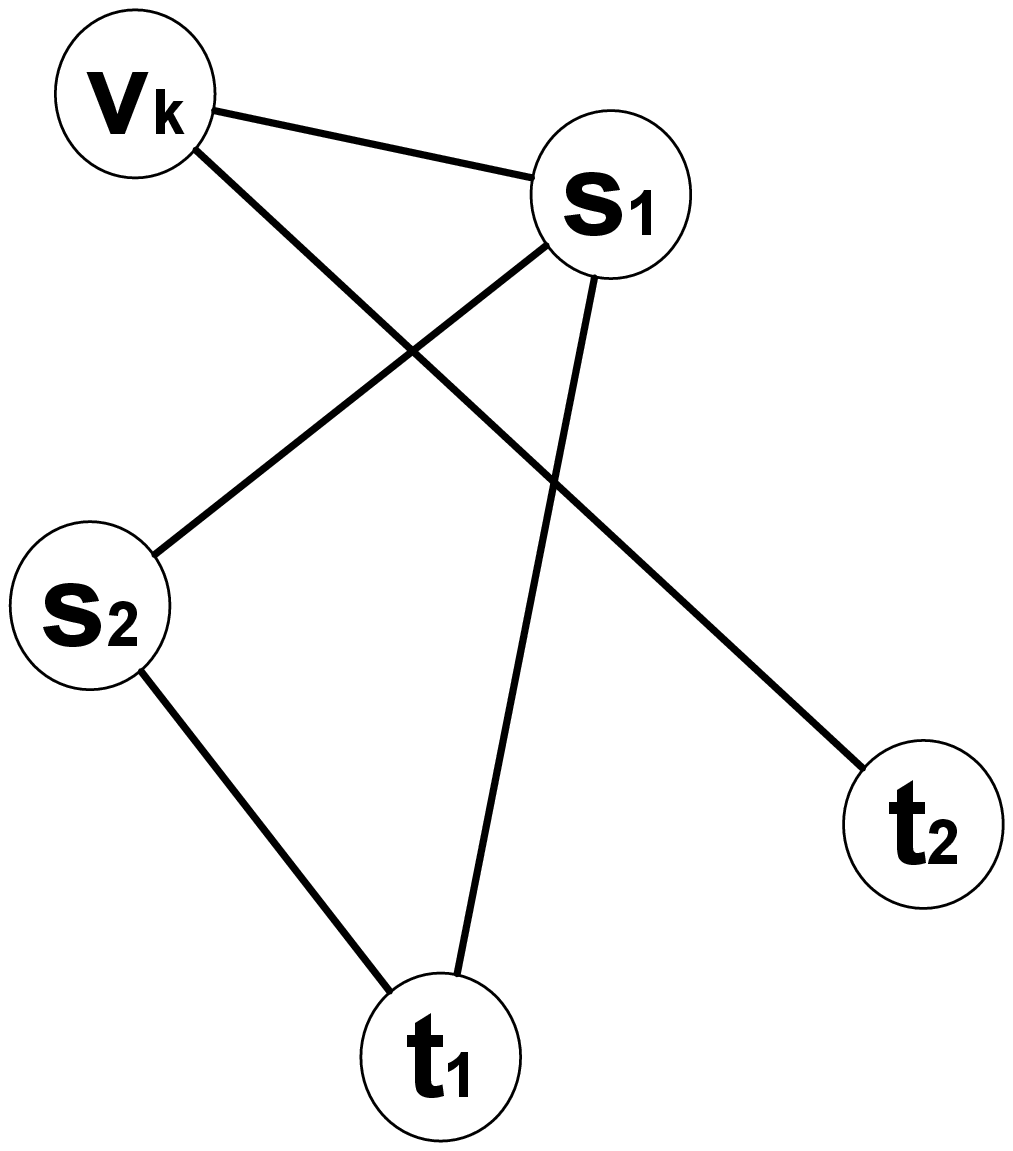}}
\subfigure[]{\includegraphics[width=1.56cm]{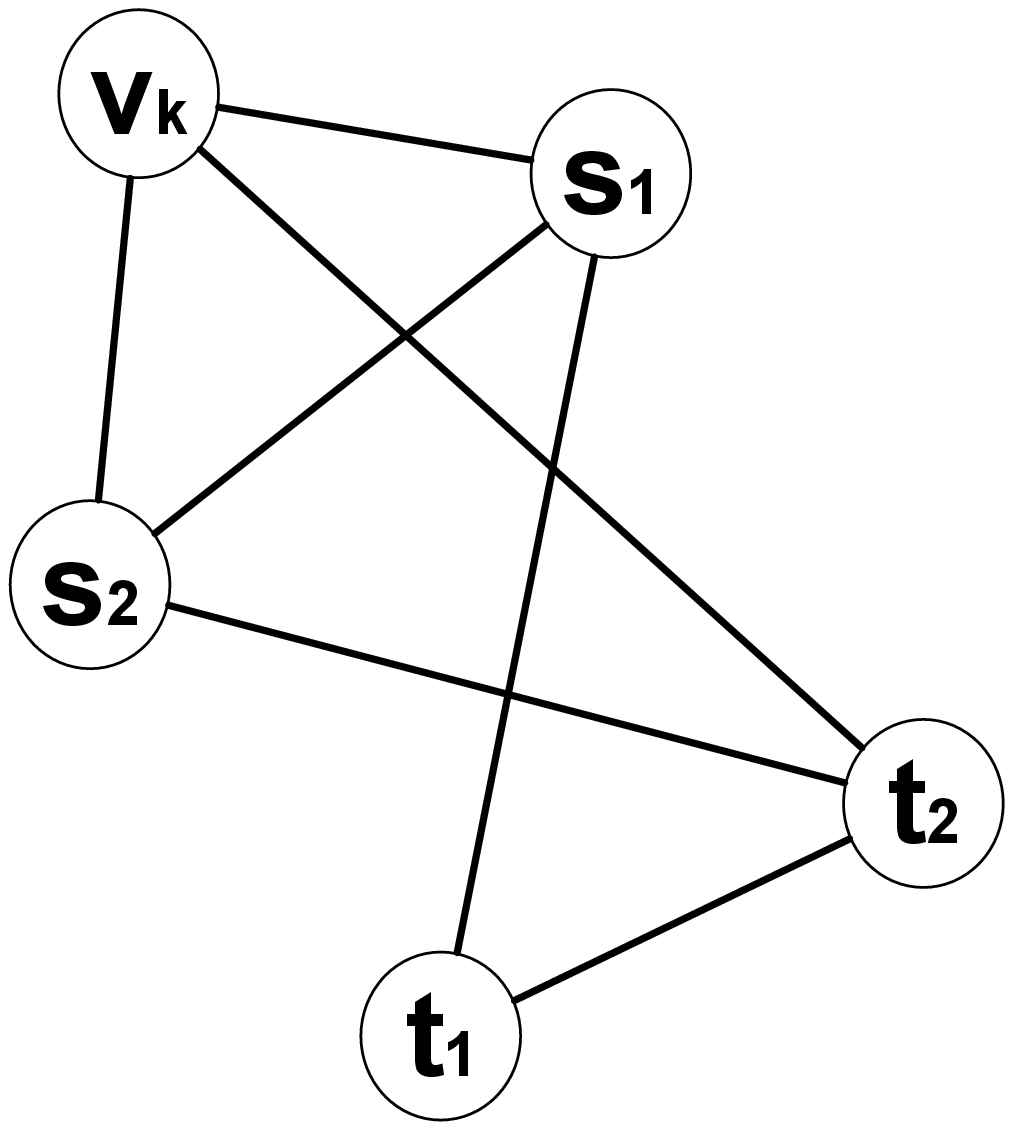}}
\subfigure[]{\includegraphics[width=1.56cm]{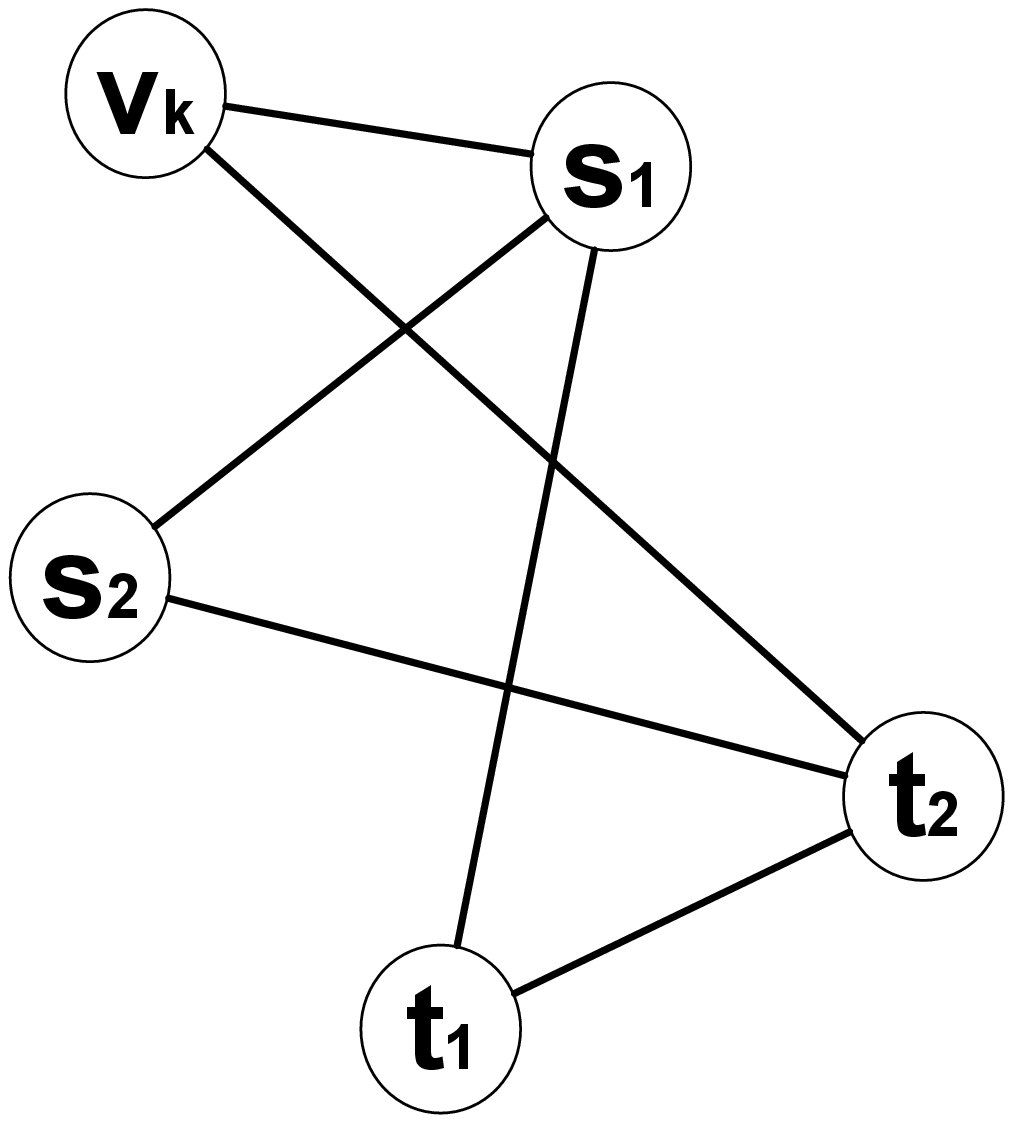}}
\subfigure[]{\includegraphics[width=1.56cm]{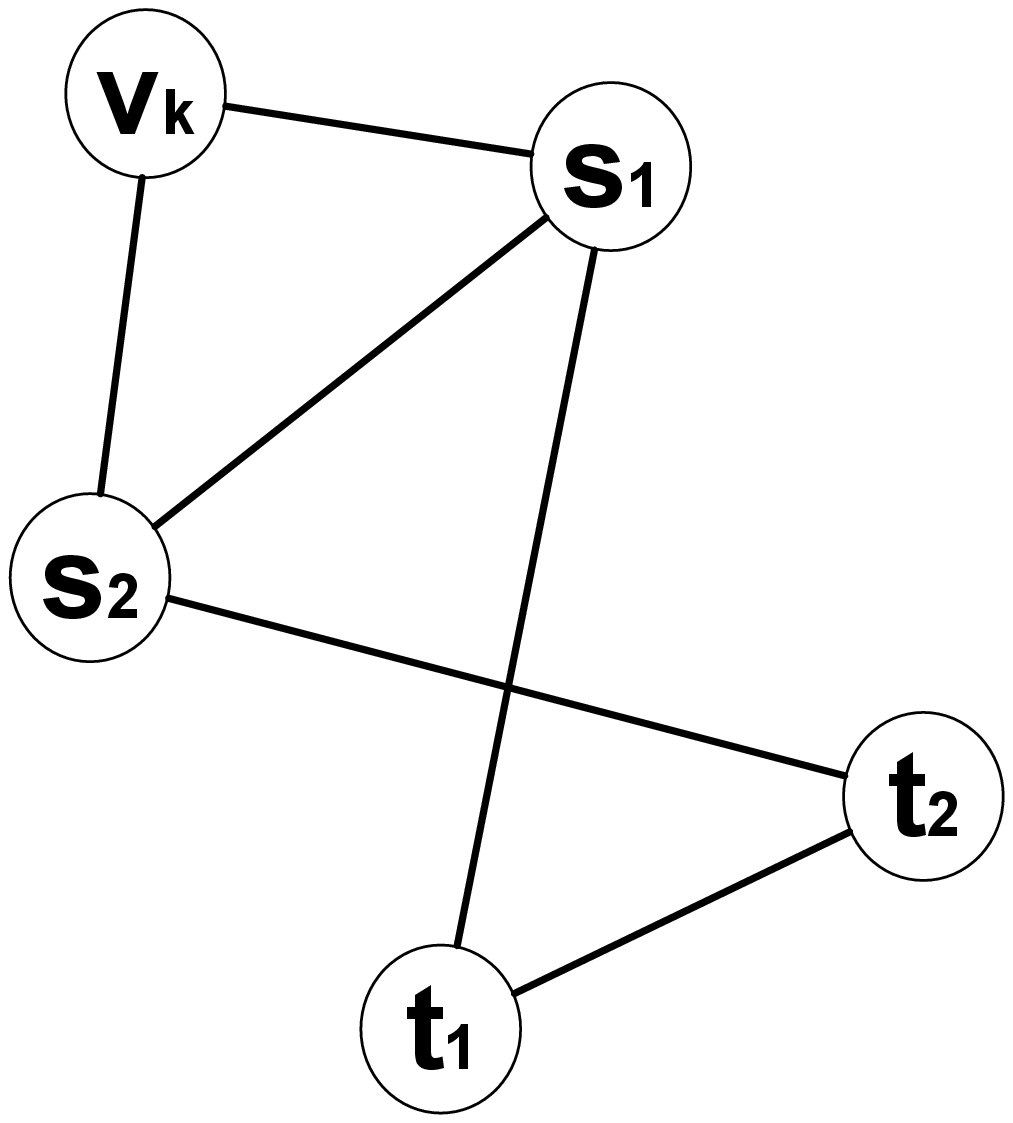}}
\subfigure[]{\includegraphics[width=1.56cm]{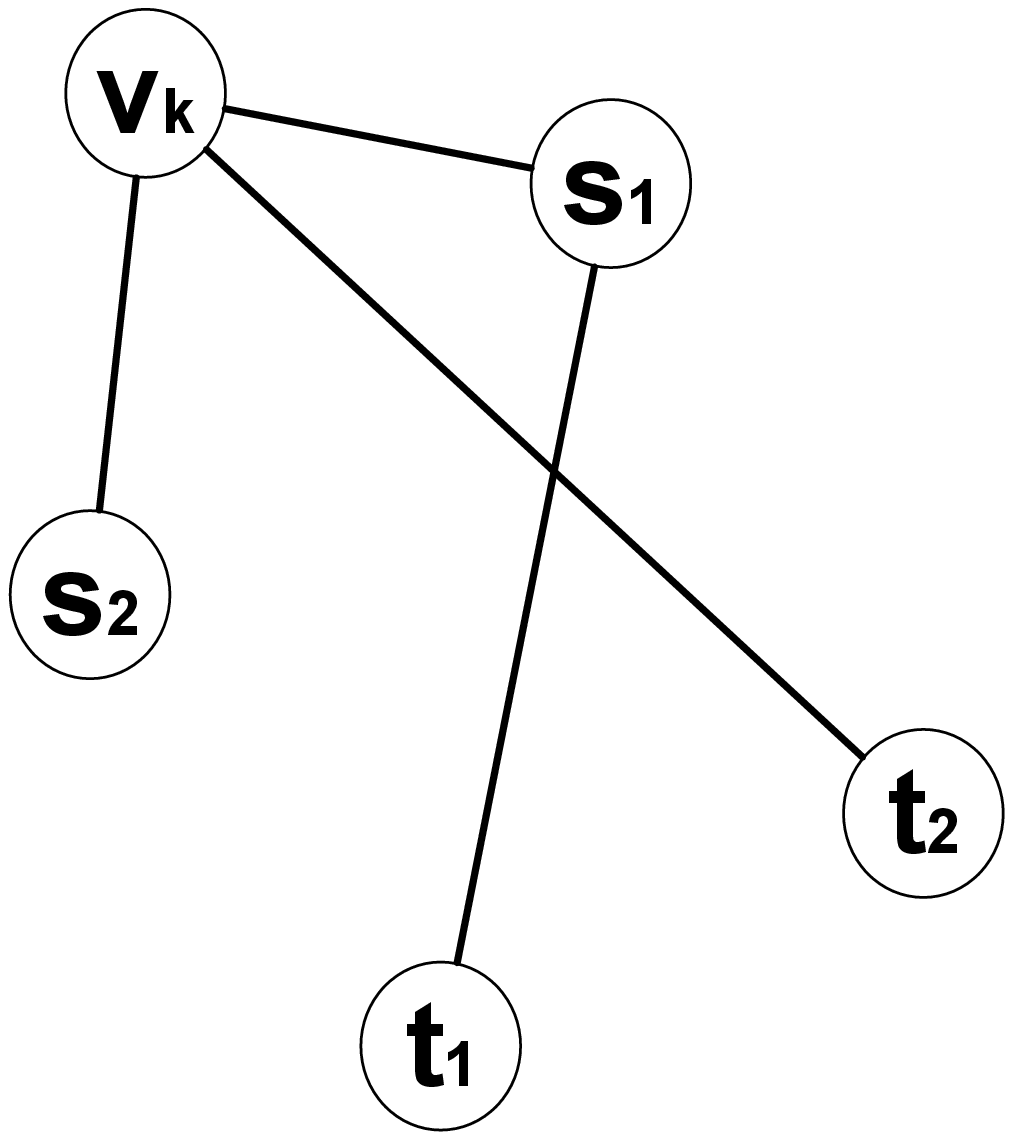}}
\subfigure[]{\includegraphics[width=1.56cm]{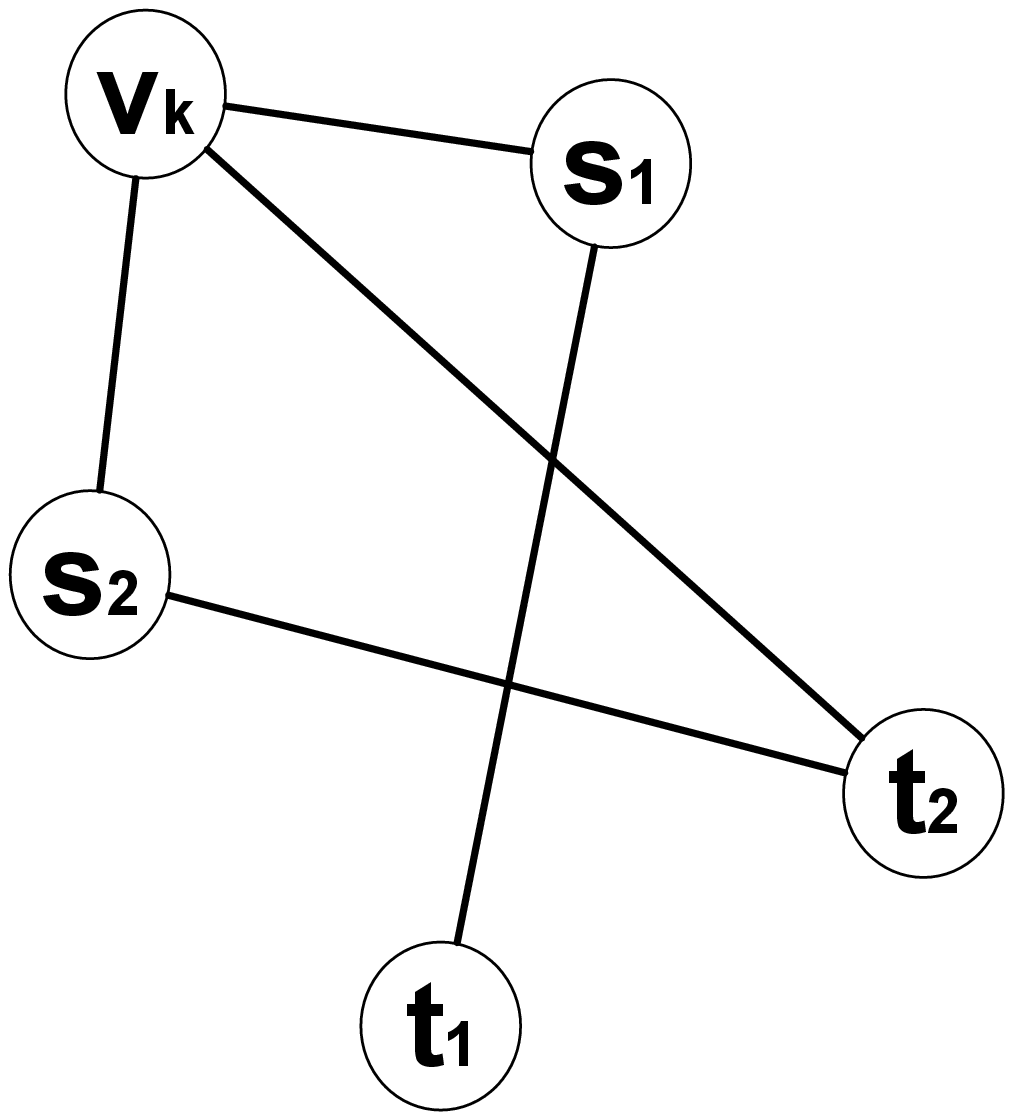}}
\subfigure[]{\includegraphics[width=1.56cm]{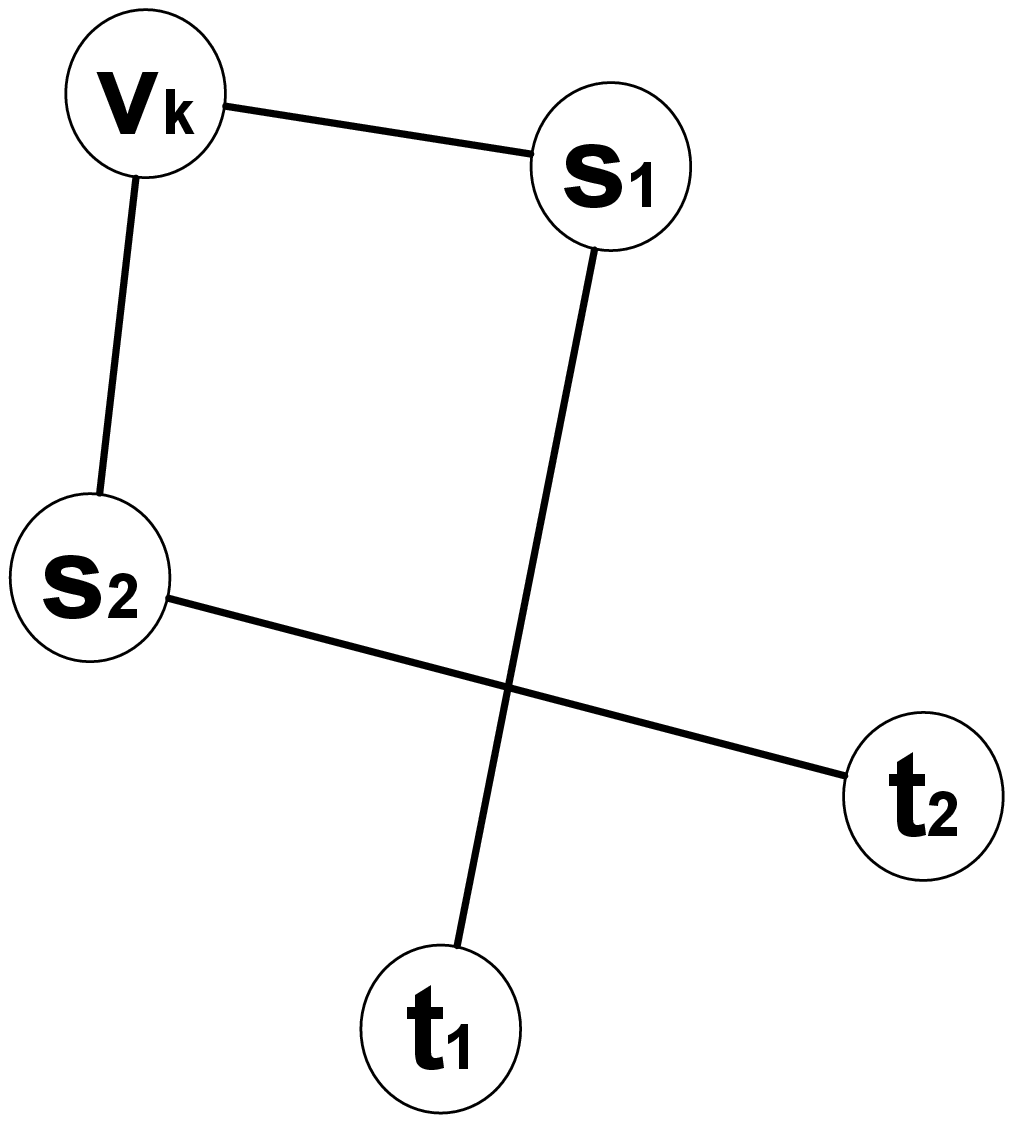}}
\caption{Graphs abiding by situation b) or c) of Proposition \ref{progeneve}, where (a)(b); (c)(d); (e); (f)(g)  are designed, respectively, from the topology structures (g) (h) (i) and (j) of Fig. \ref{Fcdtopo4}.}
\label{Fivnodesij}
\end{center}
\end{figure}

\begin{definition}\label{FCDdef}
For a graph consisting of five vertices $v_k, v_{s_1},v_{s_2},v_{t_1},v_{t_2}$, any four of them, say $v_{s_1},v_{s_2},v_{t_1},v_{t_2}$ are said to be quadruple controllability destructive (QCD) nodes if they conform to any of the following topologies:
\begin{itemize}
\item $v_{s_1},v_{s_2},v_{t_1},v_{t_2}$ take any of the topology structures of Fig. \ref{Fcdtopo4} with $v_k$ incident to all of them. In this case, the corresponding eleven graphs are depicted in Fig. \ref{Fivnodeelv}.

\item $v_{s_1},v_{s_2},v_{t_1},v_{t_2}$ take the topology structure (f) of Fig. \ref{Fcdtopo4} with $v_k$ incident to either $v_{s_1},v_{s_2}$ or $v_{t_1},v_{t_2}.$ The corresponding graphs are respectively (g) (i) of Fig. \ref{Fivnodesf}.

\item $v_{s_1},v_{s_2},v_{t_1},v_{t_2}$ take the topology structures (h) (j) of Fig. \ref{Fcdtopo4} with $v_k$ incident to $v_{s_1},v_{s_2}.$ In this case, the corresponding graphs are respectively (e) (h) of Fig. \ref{Fivnodesij}. 
\end{itemize}
\end{definition}

Relabel  $v_k=v_1, v_{s_1}=v_2,$ $v_{t_2}=v_3, v_{t_1}=v_4, v_{s_2}=v_5.$

\begin{lemma}\label{fivelem}
For a graph $\mathcal{G}$ consisting of five vertices, $\bar y=[0,y_2,y_3,y_4,y_5]$ with $y_2,y_3,y_4,y_5\neq 0$ is an eigenvector of $\mathcal{L}$ if and only if $v_2,v_3,v_4,v_5$ are QCD nodes of $\mathcal{G}.$ 
\end{lemma}
\begin{proof}
(Necessity)
Let $\bar y$ be an eigenvector of $\mathcal{L}.$ Since $\mathcal{V}\setminus\{v_{2}, v_{3}, v_{4}, v_{5}\}$ contains only one element $v_1$ for a graph of five vertices, situation e) of Proposition \ref{progeneve} cannot occur (or else, $v_1$ will be isolated), and any two of a) b) c) do not arise simultaneously. Thus all connected graphs complying with i) or ii) of Proposition \ref{progeneve} can be generated by just following one and only one of a) b) c), and accordingly, by Proposition \ref{progeneve}, constitute all the possible graphs of five nodes with $\bar y$ being an eigenvector. All these graphs are shown in Figures \ref{Fivnodeelv} to \ref{Fivnodesij}.  

First, consider graphs designed from (a) of Fig. \ref{Fcdtopo4}. Calculations show that the necessary condition (\ref{situaeqn}) of Lemma \ref{fivenolem} is met by graph (a) of Fig. \ref{Fivnodeelv}, and condition (\ref{foucases}) is not met by (a) (b) of Fig. \ref{Fivnodesb}, nor is condition (\ref{sitaceqn}) met by (c) (d) of Fig. \ref{Fivnodesb}. Thus (a) (b) (c) (d) of Fig. \ref{Fivnodesb} are excluded from the graphs with $\bar y$ being an eigenvector. For graphs designed from the other topologies of Fig. \ref{Fcdtopo4}, similar arguments yield that only (g) (i) of Fig. \ref{Fivnodesf} and (e) (h) of Fig. \ref{Fivnodesij} satisfy the associated necessary conditions of $\bar y$ being an eigenvector. Thus if $\bar y$ is an eigenvector, $v_{2}, v_{3}, v_{4}, v_{5}$ are QCD nodes. 

(Sufficiency)
For graph (a) of Fig. \ref{Fivnodeelv} with QCD nodes $v_2,v_3,v_4,v_5;$ $d_{1}=d_{2}=4,$ $d_{3}=d_{4}=d_{5}=2.$ For $v_1,$ the eigencondition requires  
$
4{y_1} - ({y_2} + {y_3} + {y_4} + {y_5}) = \lambda {y_1}.
$
Set $y_1=0,$ then ${y_2} + {y_3} + {y_4} + {y_5} = 0.$ For $v_5, v_4,$ the eigencondition respectively yields $2{y_5} - {y_2} = \lambda {y_5}$ and $2{y_4} - {y_2} = \lambda {y_4}.$ Thus
$
(2 - \lambda )({y_4} - {y_5}) = 0.
$
Similarly, for  $v_3,$ $(2 - \lambda )({y_3} - {y_4}) = 0$ and for $v_2,$ 
$
4{y_2} - ({y_3} + {y_4} + {y_5}) = \lambda {y_2}.
$
Take $y_3=y_4=y_5,$ the above arguments show that ${y_2} =  - 3{y_3}.$ Hence $\bar y=[0,-3,1,1,1]$ is an eigenvector of graph (a) of Fig. \ref{Fcdtopo4} with the corresponding eigenvalue $\lambda=5.$ It can be verified in the same way for the other graphs with QCD nodes that $\mathcal{L}$ has an eigenvector $\bar y.$ 
\end{proof}

\begin{theorem}\label{fivenosin}
For a communication graph consisting of five vertices, there is a single leader, denoted by $v_1,$ such that the multi-agent system with single-integrator dynamics (\ref{singmul}) is controllable if and only if the following three conditions are met simultaneously:
\begin{itemize}
\item $\mathcal{V}\setminus\{v_1\}=\{v_2,v_3,v_4,v_5\}$ do not constitute a group of QCD nodes;

\item any three of $v_2,v_3,v_4,v_5$ are not TCD nodes;

\item any two of $v_2,v_3,v_4,v_5$ are not DCD nodes.
\end{itemize}
\end{theorem}
\begin{proof}
Based on Lemma \ref{fivelem}, the result can be proved in the same vein as Theorem \ref{tripDcd}. 
\end{proof}

\begin{remark}
For a graph consisting of five vertices, Theorems \ref{doubDcd}, \ref{tripDcd}, \ref{fivenosin} conspire to answer the following question: with all different selections of leaders, what are the graph topology based necessary and sufficient conditions under which the system is controllable? Theorems \ref{fivenosin}, \ref{tripDcd}, \ref{doubDcd} answer this question with respect to, respectively, the case of single, double and triple leaders. In this sense, these three theorems together provide a complete graphical characterization for the controllability with communication graphs consisting of five vertices.
\end{remark}

\section{Conclusion}

The increasingly widespread use of networks calls for reasonable design and organization of network topologies.  For controllability of multi-agent networks, the problem was tackled in the paper by identifying 
the topology structures formed by the proposed controllability destructive nodes. These discovered communication structures not only reveal uncontrollable topologies but also result in several necessary and sufficient graphical conditions on controllability. The results exhibit a novel method of coping with controllability by which a complete graph based characterization is presented for graphs consisting of five nodes.


\begin{thebibliography}{1}
\bibitem{Cao2013}
M.~Cao, S.~Zhang, and M.~K. Camlibel.
``A class of uncontrollable diffusively coupled multi-agent systems
  with multi-chain topologies.''
\emph{IEEE Transactions on Automatic Control}, 58(2):465--469, Feb. 2013.

\bibitem{Chapman2012}
A.~Chapman, M.~Nabi-Abdolyousefi, and M.~Mesbahi.
``On the controllability and observability of cartesian product
  networks."
In \emph{51st IEEE Conference on Decision and Control}, pages 80--85,
  Maui, Hawaii, USA, December 10-13, 2012 2012.

\bibitem{Egerstedt2012a}
M.~Egerstedt, S.~Martini, M.~Cao, K.~Camlibel, and A.~Bicchi.
``Interacting with networks: How does structure relate to
  controllability in single-leader consensus networks?"
\emph{IEEE Control Systems Magazine}, 32(4): 66--73, Aug. 2012.

\bibitem{Godsil2001}
C.~Godsil and G.~Royle.
\emph{Algebraic graph theory}.
Springer, 2001.

\bibitem{JiMengACC2007}
M.~Ji and M.~Egerstedt.
``A graph theoretic characterization of controllability for multi-agent
  systems."
In \emph{Proceedings of the 2007 American Control Conference}, pages
  4588--4593, Marriott Marquis Hotel at Times Square, New York City, USA, July
  11-13 2007.

\bibitem{JiWCon}
Z.~Ji, Z.~D. Wang, H.~Lin, and Z.~Wang.
``Interconnection topologies for multi-agent coordination under
  leader-follower framework."
\emph{Automatica}, 45(12):2857--2863, 2009.

\bibitem{JiIJC}
Z.~Ji, Z.~Wang, H.~Lin, and Z.~Wang.
``Controllability of multi-agent systems with time-delay in state and
  switching topology."
\emph{International Journal of Control}, 83(2): 371--386, 2010.

\bibitem{Ji2012b}
Z.~Ji, H.~Lin, and H.~Yu.
``Leaders in multi-agent controllability under consensus algorithm and
  tree topology."
\emph{Syst. Contr. Lett.}, 61(9):918--925, July
  2012.

\bibitem{Jisubmitted2013}
Z.~Ji, H.~Lin, and J.~Gao.
``Eigenvector based design of uncontrollable topologies for networks of
  multiple agents."
In \emph{Proceedings of the 32st Chinese Control Conference}, pages
  6797--6802, Xi'an China, July 2013.

\bibitem{JiTAC2015}
Z.~Ji, H.~Lin, and H.~Yu.
``Protocols design and uncontrollable topologies construction for
  multi-agent networks."
\emph{IEEE Transactions on Automatic Control}, 60(3): 781-786, 2015.

\bibitem{Liu2008}
B.~Liu, T.~Chu, L.~Wang, and G.~Xie.
``Controllability of a leader-follower dynamic network with switching
  topology."
\emph{IEEE Trans. Automat. Contr.}, 53(4):1009--1013, 2008.

\bibitem{Lou2012}
Y.~Lou and Y.~Hong.
``Controllability analysis of multi-agent systems with directed and
  weighted interconnection."
\emph{International Journal of Control}, 85(10):1486--1496, May 2012.

\bibitem{Lozano2008}
R.~Lozano, M.~W. Spong, J.~A. Guerrero, and N.~Chopra.
``Controllability and observability of leader-based multi-agent
  systems."
In \emph{Proceedings of the 47th IEEE Conference on Decision and
  Control}, pages 3713--3718, Cancun, Mexico, Dec. 9-11 2008.

\bibitem{Martini2008}
S.~Martini, M.~Egerstedt, and A.~Bicchi.
``Controllability decompositions of networked systems through quotient
  graphs."
In \emph{47th IEEE Conference on Decision and Control}, Fiesta
  Americana Grand Coral Beach, Cancun, Mexico, December 9-11 2008.

\bibitem{Notarstefano2013}
G.~Notarstefano and G.~Parlangeli.
``Controllability and observability of grid graphs via reduction and
  symmetries."
\emph{IEEE Transactions on Automatic Control}, 58(7):1719--1731, July 2013.

\bibitem{Parlangeli2012}
G.~Parlangeli and G.~Notarstefano.
``On the reachability and observability of path and cycle graphs."
\emph{IEEE Transactions on Automatic Control}, 57(3): 743--748, March 2012.

\bibitem{Rahimian2013}
M.~A. Rahimian and A.~G. Aghdam.
``Structural controllability of multi-agent networks: Robustness
  against simultaneous failures."
\emph{Automatica}, 49(11):3149--3157, 2013.

\bibitem{Rahmani2006}
A.~Rahmani and M.~Mesbahi.
``On the controlled agreement problem."
In \emph{Proceedings of the 2006 American Control Conference}, pages
  1376--1381, Minneapolis, Minnesota, USA, Jun. 14-16 2006.

\bibitem{Rahmani2009}
A.~Rahmani, M.~Ji, M.~Mesbahi, and M.~Egerstedt.
``Controllability of multi-agent systems from a graph-theoretic
  perspective."
\emph{SIAM Journal on Control and Optimization}, 48(1):162--186, Feb. 2009.

\bibitem{Tanner2004a}
H.~G. Tanner.
``On the controllability of nearest neighbor interconnections."
In \emph{Proceedings of the 43rd IEEE Conference on Decision and
  Control}, pages 2467--2472, Atlantis, Paradise Island, Bahamas, Dec.14-17
  2004.

\bibitem{Yuan2013}
Z.~Yuan, C.~Zhao, Z.~Di, W.~Wang, and Y.~Lai.
``Exact controllability of complex networks."
\emph{Nature Communications}, 4:2447, doi: 10.1038/ncomms3447, 2013.

\bibitem{Zhangshuo2014}
S.~Zhang, M.~K. Camlibel, and M.~Cao.
``Upper and lower bounds for controllable subspaces of networks of
  diffusively coupled agents."
\emph{IEEE Transactions on Automatic control}, 59(3):745--750, 2014.

\end{thebibliography}
\end{document}